\documentclass[hidelinks,11pt]{article}

\makeatletter
\def\input@path{
  {tables/}
}
\makeatother

\usepackage{geometry}
\usepackage[onehalfspacing]{setspace}

\usepackage[table]{xcolor}

\usepackage[hyperfootnotes=false]{hyperref}
\hypersetup{
    colorlinks=false,
    linktoc=all
}

\usepackage[english]{babel}
\addto\extrasenglish{

}
\usepackage{csquotes}

\usepackage[backend=biber,style=apa,autocite=inline,hyperref=false]{biblatex}
\DeclareLanguageMapping{english}{english-apa}
\usepackage{microtype}
\usepackage{amsmath}
\usepackage{mathtools}
\usepackage{amssymb}
\usepackage{amsthm}
\usepackage{gensymb}
\usepackage{soul}

\newtheorem{proposition}{Proposition}
\newtheorem{corollary}[proposition]{Corollary}
\newtheorem{remark}[proposition]{Remark}
\newtheorem{assumption}{Assumption}

\DeclareMathOperator{\Var}{Var}

\usepackage{graphicx}
\graphicspath{
  {plots/}
}
\usepackage{float}

\usepackage{tikz}
\usetikzlibrary{arrows.meta, positioning}

\usepackage{booktabs}
\usepackage{multirow}
\usepackage{makecell}

\usepackage{threeparttable}
\usepackage{adjustbox}
\usepackage[section]{placeins}

\usepackage[nolist]{acronym}

\usepackage{enumerate}
\newcommand{\acknowledgement}{This work was funded by the Centre for Innovative Data in Economics Research and the Social Sciences and Humanities Research Council.}

\newcommand{\appref}[1]{Appendix~\ref{#1}}

\title{Modeling Shipping Emissions: Machine Learning, Engineering, and Policy Counterfactuals}

\author{
  Hiroyuki Kasahara, Allen Peters, Oliver Xu
  \thanks{Kasahara: University of British Columbia, Vancouver, Canada, \href{mailto:hiroyuki.kasahara@ubc.ca}{hiroyuki.kasahara@ubc.ca}. Peters: University of Connecticut, Storrs, USA, \href{mailto:allen.peters@uconn.edu}{allen.peters@uconn.edu}. Xu: Harvard University, Cambridge, USA, \href{mailto:jxu@hks.harvard.edu}{jxu@hks.harvard.edu}. \acknowledgement}
}
\date{September 30, 2026}

\begin{document}
\begin{titlepage}
\maketitle
\begin{abstract}
Machine learning predicts outcomes well, but predictive accuracy does not ensure reliable counterfactual responses. We examine how to combine machine learning and theory for measurement and counterfactual analysis, using maritime $\text{CO}_2$ emissions where physics provides a benchmark speed response. Matching hourly tracking data for dry bulk and container ships to annual fuel consumption reported under EU regulations, we compare engineering calculations, structural regressions, hybrid models, and machine learning. Out of sample, all estimated models predict within a few percent of reported totals, outperforming standard engineering calculations. Yet pure machine learning and unrestricted structural regression imply attenuated speed responses. Hybrids preserve the structural component's speed response by excluding speed-related inputs from the machine learning component. A cost-benefit analysis of speed reductions illustrates the policy stakes: an attenuated speed response can flip the sign of net benefits. Accurate aggregate predictions therefore cannot substitute for scrutiny of the restrictions determining counterfactual responses.

\end{abstract}

\acrodef{UBC}{University of British Columbia}
\acrodef{IO}{industrial organization}
\acrodef{FC}{fuel consumption}
\acrodef{R2}[R$^2$]{coefficient of determination}
\acrodef{t}{tonnes}
\acrodefplural{GHG}[GHG]{greenhouse gases}

\begin{acronym}
    \acro{AE}[AE]{auxiliary engine}
    \acro{AIS}[AIS]{Automatic Identification System}
    \acro{BO}[BO]{boiler}
    \acro{CII}[CII]{Carbon Intensity Indicator}
    \acro{CO2}[CO$_2$]{carbon dioxide}
    \acro{DWT}[DWT]{deadweight tonnage}
    \acro{EEDI}[EEDI]{Energy Efficiency Design Index}
    \acro{EEXI}[EEXI]{Energy Efficiency Existing Ship
    Index}
    \acro{EU}[EU]{European Union}
    \acro{GDP}[GDP]{gross domestic product}
    \acro{GHG}[GHG]{greenhouse gas}
    \acro{GT}[GT]{gross tonnage}
    \acro{HFO}[HFO]{heavy fuel oil}
    \acro{iid}[iid]{independent and identically distributed}
    \acro{IMO}[IMO]{International Maritime Organization}
    \acro{IPCC}[IPCC]{Intergovernmental Panel on Climate Change}
    \acro{IV}[IV]{instrumental variables}
    \acro{MAE}[MAE]{mean absolute error}
    \acro{MAPE}[MAPE]{mean absolute percentage error}
    \acro{ME}[ME]{main engine}
    \acro{ML}[ML]{machine learning}
    \acro{MMSI}[MMSI]{Maritime Mobile Service Identity}
    \acro{MRV}[MRV]{Monitoring, Reporting, and Verification}
    \acro{MSE}[MSE]{mean squared error}
    \acro{nm}[n miles]{nautical miles}
    \acro{NT}[NT]{net tonnage}
    \acro{OECD}[OECD]{Organisation for Economic Co-operation and Development}
    \acro{OLS}[OLS]{ordinary least squares}
    \acro{RMSE}[RMSE]{root-mean-squared error}
    \acro{RTE}[RTE]{relative total error}
    \acro{SCC}[SCC]{social cost of carbon}
    \acro{SOG}[SOG]{speed over ground}
    \acro{SFC}[SFC]{specific fuel consumption}
    \acro{STW}[STW]{speed through water}
    \acro{UNCTAD}[UNCTAD]{United Nations Conference on Trade and Development}
    \acro{USD}[USD]{United States dollars}
    \acro{WFR}[WFR]{World Fleet Register}
    \acrodef{R2}[R$^2$]{coefficient of determination}
\end{acronym}

\thispagestyle{empty}
\end{titlepage}

\section{Introduction}
\label{sec:MLintroduction}

Environmental and energy economists increasingly use \ac{ML} to construct counterfactual predictions in policy evaluations \parencite{burlig2020machine,christensen2023performance,prest2023rcts,abrell2022effective}.
However, predictive accuracy does not establish that a model captures how an outcome responds to a policy-relevant input \parencite{mullainathan2017machine}.
\ac{ML} predicts levels well by flexibly accommodating determinants that theory leaves unspecified, whereas structure incorporates prior knowledge to better discipline the response to a policy change.
We examine how to combine the two approaches in maritime shipping, where evaluating the welfare effects of regulating sailing speed requires estimates of both the level of fuel consumption and its response to speed.

Both levels and responses matter because policy calculations use them together.
\textcite{kleinberg2015prediction} decompose the effect of a policy decision on its payoff into a prediction term and a causal term, and make the point that choosing a policy generally requires knowing both.
Welfare analysis is a prime example, where the change in an outcome is often computed as a baseline level multiplied by a response.
For instance, \textcite{borenstein2019expecting} derive the distribution of allowance prices in California's cap-and-trade market from a time-series model of business-as-usual emissions combined with separately estimated price elasticities of abatement.
They find that uncertainty in business-as-usual emissions dwarfs the abatement likely to occur within the politically acceptable price range.
Because the level and the response enter multiplicatively, accuracy in one does nothing to offset error in the other.

Our setting is \ac{CO2} emissions from maritime shipping. Ships carry more than 80\% of world trade by volume, and fuel consumption from international shipping emitted 740 million tonnes of \ac{CO2} in 2018, about 2\% of the global total \parencite{unctad2017review,faber2020fourth}. %
Regulators are working to improve shipping emissions estimates while simultaneously iterating on policies to reduce them. The \acs{OECD} has switched between the two extremes of the approaches we study: until 2025 its shipping emissions database used an \ac{ML} model trained on reported fuel efficiency, and it now uses the calculation of \ac{IMO} \parencite{oecd2025maritime,clarke2023co2}.
The \ac{IMO} has implemented a range of policies targeted at net-zero emissions by or around 2050 \parencite{imo2023strategy}, and the \ac{EU} Emissions Trading System has covered shipping since January 2024 \parencite{eu2023directive959}. %

The speed--emissions relationship is central to policy effectiveness in the short run.\footnote{The \ac{IMO}'s efficiency rules for existing ships leave few options other than slowing down.} It also offers something few economic settings do: a physical benchmark for the response. %
In the engineering calculation behind the \ac{IMO}'s emissions inventories, main-engine power rises with the cube of speed, so ship tracking data and a ship's technical characteristics together yield predicted fuel consumption at any speed.
Engineering projections are nonetheless unreliable in practice: realized savings from energy-efficiency programs fall well short of projections \parencite{fowlie2018does,christensen2023performance,burlig2020machine}.
Our results suggest that shipping is no exception.

We therefore develop a range of estimators and judge each on two criteria: the accuracy of predicted aggregate fuel consumption relative to reported totals, and the fidelity of the estimated speed response to the engineering benchmark. We find that models with similar aggregate bias imply speed responses that differ enough to reverse the sign of net benefits in an illustrative cost-benefit analysis of slow steaming for dry bulk ships.

The engineering calculation overpredicts reported fuel consumption in the 2021 test sample by 6.5\% for dry bulk ships and 13.9\% for container ships, and overpredicts for dry bulk ships in every training year.\footnote{The \ac{IMO}'s own comparison of its 2018 inventory with the reports to the \ac{EU} finds an overstatement of 5.5\% overall, with errors of $-0.2\%$ for bulk carriers and 6\% for container ships \parencite{faber2020fourth}. Our figures refer to 2021 and to ship-years whose tracked and reported distances agree under our implementation of the data cleaning procedures, so the two sets of numbers are not directly comparable.} %
\ac{ML}, by contrast, predicts annual fuel consumption within a few percent.
Agreement in annual totals, however, may say little about the response to speed.
Speed is the main short-run margin on which policy acts: lower operating speeds drove much of the fall in shipping's carbon intensity after 2008 \parencite{faber2020fourth}, and policy pushes in the same direction. %

We compare estimators that vary in terms of flexibility and physics content. At one end of this spectrum is the engineering calculation itself, as calibrated for the \ac{IMO} inventories; at the other are two classes of pure \ac{ML} models with no engineering form imposed, one replicating the \acs{OECD}'s model \parencite{clarke2023co2,oecd2025maritime} and the other, which we call physics-consistent, using inputs derived from the physics. %
Between them lie structural regressions, which keep the engineering equation but estimate its coefficients, and hybrids, which combine an engineering calculation or a structural regression with a flexible correction. %

A hybrid has two components.
We call a component structural when it takes the engineering form, whether its coefficients are calibrated or estimated, and flexible when the data alone determine its shape.
Flexibility comes in degrees: estimating the engineering coefficients rather than calibrating them is a first step, and \ac{ML}, which imposes no form at all, is the limit.
In the hybrids, the flexible correction sets the level of fuel consumption and the structural component sets its response to speed.
The division rests on one design choice: the correction never receives speed, or any variable derived from it, as an input, so it adjusts the level while each ship's response to speed remains that of its structural component.\footnote{A correction that multiplies the engineering calculation preserves each ship's elasticity of fuel consumption with respect to speed, and a correction added to a structural regression preserves the derivative in levels; both properties hold with the inputs of the correction held fixed (\autoref{prop:preservation} in \appref{app:hybrid}, available online). Whether the correction carries over to a counterfactual rests on the assumption that the relation of its retained inputs to fuel consumption does not shift with speed, which cannot be verified outside the training data.} %
Pairing a flexible level with a structural response is implicit in the partially linear model of \textcite{robinson1988root} and in the separation of a causal parameter from flexibly estimated nuisance functions in double machine learning \parencite{chernozhukov2018double}.
We make the pairing explicit by using theory to decide which variables enter each component.

We study dry bulk carriers and container ships, the two ship types that emit the most \ac{CO2}, over 2019--2021. %
Our target is the annual fuel consumption that ships report under the \ac{EU}'s Monitoring, Reporting and Verification (\acs{MRV}) regulation for voyages to, from, and between ports in the European Economic Area.\acused{MRV}\footnote{The regulation applies to ships above 5,000 gross tonnage and covers virtually all ships of the two types we study. Accredited third parties verify the reports, and no binding regulation of \ac{CO2} emissions applied to these ships during our sample period, so ships had little incentive to misreport \parencite{eu2015regulation}.} %
The predictors are ship characteristics and annual variables built from each ship's hourly \ac{AIS} transmissions, which provide positions and speeds throughout the reported voyages. %
Requiring tracked and reported voyage distances to agree leaves 2,443 ship-years for dry bulk carriers and 1,036 for container ships.\footnote{The main text reports dry bulk results; container ship results are in the online appendix. In the speed counterfactual, container ships follow the same pattern, with responses closer to the benchmark, although the pattern across speed quartiles differs.} %

We organize the evaluation around distance from the training data. %
A counterfactual evaluates the model at values of a variable it saw in training, but the new values need not lie within the support of the training data.
We measure each observation's distance from that support rather than assume it is small.\footnote{The measure is the Mahalanobis distance of each observation from the training sample, computed on each model's own inputs (\autoref{sec:mlexplanation}). It is one-sided: a large distance flags a departure from the training data, whereas a small distance does not establish that an observation lies within them. By this measure the fleet and quartile hold-outs place many observations unusually far from the training data for some models' inputs.}
We test prediction in scenarios that deviate increasingly from the support of the training data: a held-out year, then held-out quartiles of ship size, distance, and speed, and a held-out ship type. %
A policy counterfactual deviates further still, beyond observable outcomes, so we instead compare each model's predicted response with the physical benchmark, the response of the engineering calculation itself. %
The exercise imposes a common percentage change in at-sea speeds, holding distances fixed, and records each model's predicted change in aggregate fuel consumption. %

Four results stand out.
First, every estimated model predicts aggregate fuel consumption within a few percent out of sample, an improvement on the engineering calculation. %
Second, the physics-consistent \ac{ML} models dominate the \acs{OECD}'s fuel-efficiency model: they are similarly accurate in level, but the former respond to speed while the latter does not. %
Third, level accuracy and response fidelity trade off, and the trade-off takes two forms.
Models anchored to the engineering calculation keep the benchmark response yet inherit part of its level bias. %
Models that let the data determine the shape are less biased in levels but lose most of the response relative to the benchmark: in the pure \ac{ML} models the response is nearly flat, and in the unrestricted structural regression it is attenuated and concave in speed where the cube law is convex. %
Fourth, the remedies differ.
The hybrids keep their structural component's response because the correction never receives the speed-related variables. %
The structural regression returns to the benchmark when constrained at the design load, the engine load at which the engineering model places peak efficiency. %
Two models without \ac{ML} and without an intercept come closest on both counts: the rescaled engineering calculation, which keeps the benchmark response by construction, and the constrained structural regression, which recovers it by estimation. %
Each is nearly unbiased out of sample and nearly matches the benchmark response. %

We consider a structural restriction on the coefficients of the engine efficiency curve that sets its slope to zero at the design load, so that the slope is tied to the curvature.\footnote{The structural regression is linear in the first three powers of relative engine load summed over hours, with coefficients those of the efficiency curve $s(w)=\beta_1+\beta_2 w+\beta_3 w^2$. Setting its slope to zero at the design load $\bar w=0.78$ imposes $\beta_2=-2\bar w\beta_3$, so the constrained regression replaces the squared and cubed load terms with the single regressor $\sum_h (w_{ih}^3-2\bar w\,w_{ih}^2)$ and estimates two coefficients by least squares. Ship-year mean loads lie near $\bar w$ in both fleets, so the constrained curve is nearly flat over the loads ships use and fuel per hour is nearly proportional to load, which is the cube law.} %
The unrestricted coefficients reject the constraint for both fleets, yet the constrained regression reproduces the benchmark response.\footnote{Under a correction for wave resistance, the constrained regression's response to slow-downs falls to about a third of its uncorrected size while its response to speed-ups changes little.} %
We also investigate why the pure \ac{ML} models' response is flat and the unrestricted regression's attenuated, but none of the three candidate mechanisms we examine (extrapolation, measurement error, and the nonlinearity of the structural regression) fully explains either failure.\footnote{Extrapolation contributes modestly to the flat response of the tree-based \ac{ML} models, and measurement error of the kinds and sizes that the data admit reproduces neither response.}

Two applications show the models at work. %
The first is a welfare assessment of a change in ship speeds, the margin through which fuel levies and efficiency rules act. %
For dry bulk ships and five illustrative models, we price the fuel, carbon, and time-at-sea consequences of imposed changes in at-sea speed of 5\% and 10\% in either direction; the time cost is the same for every model, so the models differ only in predicted fuel savings. %
Those that carry the benchmark speed response, the engineering calculation and the constrained structural regression, find that slowing down pays, because fuel and carbon savings exceed the cost of longer voyages; those with a flat response find the reverse.\footnote{We impose the speed changes rather than model how operators would choose speeds under a policy, hold distance, fleet, and trade fixed, and abstract from general equilibrium effects, so the exercise supports no conclusion about the policy itself; it shows how far the choice of model propagates into a policy calculation. \textcite{peters2024beached} gives a dynamic, general equilibrium treatment.} %
The sign therefore turns on a model's response to speed, not on its level accuracy: the unrestricted structural regression, among the least biased level predictors, reaches the same verdict as the pure \ac{ML} models. %

The second application disaggregates \ac{CO2} emissions by month and by origin--destination country pair over all global ship movements, similar to what the \acs{OECD} has recently begun to publish based on engineering calculations \parencite{oecd2025maritime}. %
The task looks like pure prediction, yet a ship-month covers far less distance than the ship-years the models were trained on, and global movements include trips the training sample never saw; the distance diagnostic accordingly places a sizeable share of the ship-months and route aggregates beyond the support of the pure \ac{ML} models but not of the structural ones. %
The models agree on month-to-month movements, such as the dip at each Chinese New Year and the pandemic slump in container shipping, and on which routes dominate; they differ mainly in level, and where the monthly series part ways the divergence traces to their speed responses. %

To researchers whose models specify one part explicitly and learn the rest from data, we offer three things.
The first is a taxonomy of estimators for settings with some structure, organized by how the components are combined. %
It includes a hybrid that corrects the engineering calculation with a flexible component yet preserves its response to speed by withholding the speed-related variables.\footnote{\textcite{almeida2023machine} boost parametric option-pricing models with a neural network trained on their pricing errors. Their correction withholds none of the parametric model's inputs, whereas ours withholds the speed-related ones.} %
The second is the examination of the three candidate mechanisms reported above. %
A single constraint at the design load then brings the response back relative to the benchmark. %
The third is a caution and a tool: aggregation can move a task far from the training data, and our diagnostic flagged the pure \ac{ML} models' ship-months in advance. %
An unflagged task can still leave a flexible model far from the benchmark.\footnote{In the evaluation there is little correspondence between the tree-based models' errors and the diagnostic's statistics, and the diagnostic did not single out the unrestricted structural regression's failure for container ships.} %
For emissions measurement, the official calculation overpredicts reported fuel with our data, our replication of the \acs{OECD}'s model shows no response to speed for dry bulk ships, and the cost-benefit sign turns on that response. %

\paragraph*{Related literature.}
The papers cited at the outset are illustrative of how \ac{ML} is being used in environmental economics. \textcite{prest2023rcts} train on pre-treatment consumption to predict the counterfactual for a pricing experiment and recover the experimental estimate.
Similarly, \textcite{burlig2020machine} pair an \ac{ML} counterfactual for the level with a fixed-effects estimate of the response to evaluate energy-efficiency upgrades in schools.
Closest to ours, \textcite{abrell2022effective} read the effect of the UK carbon price support off \ac{ML} predictions with and without the tax, which enters through relative fuel prices observed before and after its introduction.
In the first two the policy is a treatment; our counterfactual instead changes speed, a continuous input present throughout training. %
\textcite{prest2023rcts} also suggest that predictive accuracy is an imperfect proxy for recovering a treatment effect; our level-versus-response result is the same warning for a continuous input that the counterfactual changes.

In the broader debate on \ac{ML} against structural models, \textcite{nevo2010taking} argue that, with credible identification and unchanged parameters, structural models predict responses to changes not yet observed.
They add that such models support welfare calculations where changes in well-being cannot be estimated directly as treatment effects.
In their application to risk preferences, \textcite{andrews2025transfer} find that black-box algorithms fit slightly better within domain while economic models transfer more reliably.
An alternative approach to ours in blending \ac{ML} and structure is that of \textcite{chen2026teaching}, who pre-train a network on synthetic structural data and fine-tune it on real data.
Our focus differs in that we evaluate the response to speed as well as prediction error.

Finally, the paper can help to inform economic work on shipping that uses or addresses the speed--fuel relationship.
\textcite{kalouptsidi2014time} assumes fuel consumption proportional to the cube of speed in the operating costs of a structural model of bulk shipping.
With vessel-tracking data, \textcite{klotz2022local} evaluate a regional fuel rule and find that behavioral adjustments, slower steaming in the regulated zone among them, together change its welfare accounting.
In turn, \textcite{hummels2013time} estimate the value of time in transit as an ad valorem equivalent per day, which our welfare exercise sets against fuel savings.
Closer to our exercise, \textcite{adland2020optimal} test the cubic law on the daily fuel reports of a small tanker fleet.
They find it holds near design speed, with a lower elasticity at the speeds the ships usually sail, which bears on the open question above.
Here the relationship itself is the object of study, at the scale at which policy is evaluated.

\bigskip
We describe the industry and the engineering model in \autoref{sec:MLbackground}, the data in \autoref{sec:MLdata}, and the estimators in \autoref{sec:MLmodels}.
We evaluate them in \autoref{sec:MLevaluation}, explain their performance in \autoref{sec:MLexplanation}, apply them to disaggregated emissions and a speed-change welfare assessment in \autoref{sec:MLprediction}, and conclude in \autoref{sec:MLconclusion}.

\section{Background}
\label{sec:MLbackground}

\subsection{Industry and Emissions}

Maritime shipping carries over 80\% of world trade by volume \parencite{unctad2017review} and contributes approximately 3\% of global \ac{GHG} emissions \parencite{faber2020fourth}. We focus on the two largest \ac{CO2}-emitting sectors: dry bulk carriers (transporting unpackaged commodities) and container ships (carrying manufactured goods). Because \ac{CO2} emissions are directly proportional to fuel use, we model fuel consumption.\footnote{Nearly all ships in our sample period burn \acl{HFO} or, following the 2020 \acs{IMO} sulfur cap, low-sulfur variants with the same carbon content per tonne, though slightly different energy density. Lower emissions fuels are being adopted, but the transition will likely take decades.}

The dominant physical determinant of fuel consumption is hydrodynamic drag, whose power requirement scales as the cube of speed, making speed the most immediate margin for reducing emissions. Speed is also the margin on which operators respond to economic conditions such as freight rates and fuel prices. Accordingly, container ships operate at higher speeds than bulk carriers, reflecting the time value of their cargo, so the two sectors occupy different regions of the speed distribution. Given this and other differences in operation and design, we treat the two sectors as separate domains throughout the analysis.

The \acf{IMO} is the multilateral organization responsible for monitoring and regulating emissions from international shipping. Its official emissions reports are at the annual level and published infrequently; the most recent \ac{GHG} study dates from 2020 \parencite{faber2020fourth}. The regulatory environment has tightened considerably since our sample period. The revised 2023 \ac{GHG} strategy of the \ac{IMO} targets net-zero emissions from international shipping by or around 2050 \parencite{imo2023strategy}, and two measures now in force are relevant: the \ac{CII} rates ships on operational carbon intensity and the \ac{EEXI} enforces an efficiency standard. Since technological abatement options are limited, speed reduction is the predominant compliance mechanism. Separately, the \ac{EU} applied its Emissions Trading System to maritime transport in January 2024 \parencite{eu2023directive959}, in theory reducing speeds by raising the effective price of fuel.

\subsection{Theoretical Fuel Consumption}\label{subsec:engcalc}
\label{sec:theory}

Engineering formulas used by the \ac{IMO} for its official emissions reports \parencite{faber2020fourth} capture the core physics of drag and engine efficiency but rely on calibrated parameters and coarse approximations of factors such as weather, loading, and hull condition. Our analysis takes these formulas as the theoretical baseline. The fundamental components are outlined here, with further details in \appref{app:calculation}.

A ship $i$'s hourly fuel consumption sums the contributions of the \acf{ME}, \acf{AE}, and \acf{BO}:
\begin{equation}\label{eqn:fc_inst}
    \mathrm{FC}_{ih} = W_{\text{ME},ih} \cdot \mathrm{SFC}_{\text{ME},ih} + W_{\text{AE},ih} \cdot \mathrm{SFC}_{\text{AE},ih} + W_{\text{BO},ih} \cdot \mathrm{SFC}_{\text{BO},ih},
\end{equation}
where $W$ denotes demanded power and $\mathrm{SFC}$ denotes \acl{SFC}, the fuel burned per unit of power delivered, for each component. \Acl{AE} and \acl{BO} power values are assigned based on the ship's operating phase, as determined by its speed and distance from shore, and their \ac{SFC} values are constants taken from \textcite{faber2020fourth}.

The \acl{ME} power is given by the Admiralty formula, which is nonlinear in speed $v_{ih}$ and draft $t_{ih}$:
\begin{equation}\label{eqn:admiralty}
    W_{\text{ME}, ih} = C_i \cdot W_{\text{ME},\textit{ref}, i} \cdot \left( \frac{t_{ih}}{t_{\textit{ref}, i}} \right)^{0.66} \cdot \left( \frac{v_{ih}}{v_{\textit{ref}, i}} \right)^{3},
\end{equation}
where $C_i$ is a ship-specific constant that corrects for both weather and hull fouling, and \textit{ref} subscripts denote reference values from ship specifications. \Acl{ME} efficiency is quadratic (u-shaped) in the engine load $w_{ih} = W_{\text{ME},ih}/W_{\text{ME},\textit{ref},i}$:
\begin{equation}\label{eqn:sfc}
    \mathrm{SFC}_{\text{ME},ih} = \mathrm{SFC}_{\text{base}, i} \cdot \left(0.455\, w_{ih}^2 - 0.710\, w_{ih} + 1.280\right),
\end{equation}
where $\mathrm{SFC}_{\text{base}, i}$ depends on the ship's engine type, fuel type, and year built, following Table 19 of \textcite{faber2020fourth}. The quadratic reaches its minimum at 78\% of rated power, so engines running well below or above that load burn more fuel per unit of power delivered.

Annual fuel consumption $\mathrm{FC}_i$ is the sum over all hours $h$ in a year:
\begin{equation}\label{eqn:annual_fc}
    \mathrm{FC}_i = \sum_h \mathrm{FC}_{ih}.
\end{equation}
The nonlinearity of this construction matters for what follows. Because power is convex in speed, predicting annual consumption from average speed understates it by Jensen's inequality. Additionally, errors in either speed (\acs{AIS} reports \acl{SOG} rather than the \acl{STW} that governs drag) or $C_i$ (weather and fouling effects) do not average out, and slow speeds that coincide with currents or adverse weather bias the estimated speed relationship, which we discuss further in \autoref{sec:structuralexplanation}.

\section{Data}
\label{sec:MLdata}

We employ ship-level data on the global fleet of dry bulk and container ships spanning 2019--2021. Three data sets provide tracking data, reported fuel consumption, and vessel characteristics. A key challenge is a temporal and spatial mismatch: the fuel consumption target is an \textit{annual} aggregate for voyages into and out of the \ac{EU}, while tracking data is \textit{hourly} and global. To align the two, we extract the subset of ship movements on \ac{EU} voyages and aggregate to the ship-year level.

Tracking data consists of \acs{AIS} messages in which each ship broadcasts its position, speed, and other operational information. We purchase hourly transmissions from Spire for all dry bulk and container ships. Raw \ac{AIS} data contains erroneous and missing observations, so we apply standard cleaning and interpolation procedures (\appref{app:MLDataCleaning}).

Our target is annual fuel consumption reported under the \ac{EU}'s \ac{MRV} program, which covers voyages into and out of the European Economic Area by virtually all bulk and container ships \parencite{eu2015regulation}.\footnote{The program applies to ships above 5,000 \ac{GT}. Reports are verified by accredited third parties, and during our sample period there were no binding emissions regulations and therefore no clear incentive to misreport.}

We purchase a fleet register containing vessel characteristics such as ship type, age, size, engine power, and rated speed from Clarksons Research (the \ac{WFR}). This also provides each ship's vessel and \ac{AIS} identifiers (\ac{IMO} number and \ac{MMSI}), which link the above data sets. For robustness exercises, we additionally draw on publicly available ocean reanalysis data from the Copernicus Marine Service for currents and waves (\appref{app:weather}).

To harmonize the tracking data with the \ac{MRV} data, we identify port calls from observed speed and proximity to land, assign each to the country whose exclusive economic zone it lies within, and define a trip as a ship's movement between consecutive port calls (\appref{app:MLDataCleaning}). Any trip with at least one port call in the European Economic Area is an \ac{EU} trip. We aggregate \ac{EU} trips to the annual level to construct the model features and join the data sets on \ac{IMO} number.

Given the potential for errors in the raw \ac{AIS} data and the trip detection procedure, we train models on the subset of ship-years for which observed annual distance on \ac{EU} trips agrees with the reported distance to within $\pm$500 \ac{nm}.\footnote{We additionally require that less than 0.3\% of distance is traveled at an implied speed above 25 knots and that at least ten hourly observations are directly observed rather than interpolated. As a robustness check, we apply an alternative distance concordance criterion of $\pm$10\% (\appref{app:distance}).} This yields 2,443 dry bulk and 1,036 container ship observations (\autoref{tab:obscounts}). As \autoref{fig:datasupporthistograms} shows, the resulting data set is limited in size but reasonably representative of all \ac{EU} trip observations. The most noticeable difference is a skew toward shorter \ac{EU} trip distances, a dimension along which we directly assess out-of-support performance in \autoref{sec:MLevaluation}. Summary statistics, along with the container ship counterparts of \autoref{tab:obscounts} and \autoref{fig:datasupporthistograms}, are provided in \appref{app:MLsummarystats}.

\begin{table}
    \centering
    \begin{threeparttable}
        \caption{Observation counts for data subsets, dry bulk ships.}
        \label{tab:obscounts}
        
\begin{tabular}[t]{lrr}
\toprule
Data Subset & \# Ships & \# Observations\\
\midrule
EU Trips & 4478 & 9488\\
Train + Test & 1999 & 2443\\
Monthly Prediction & 8469 & 331780\\
Annual Route Prediction & 8455 & 328985\\
\bottomrule
\end{tabular}

    \end{threeparttable}
\end{table}

\begin{figure}
    \centering
    \includegraphics[width=0.9\textwidth]{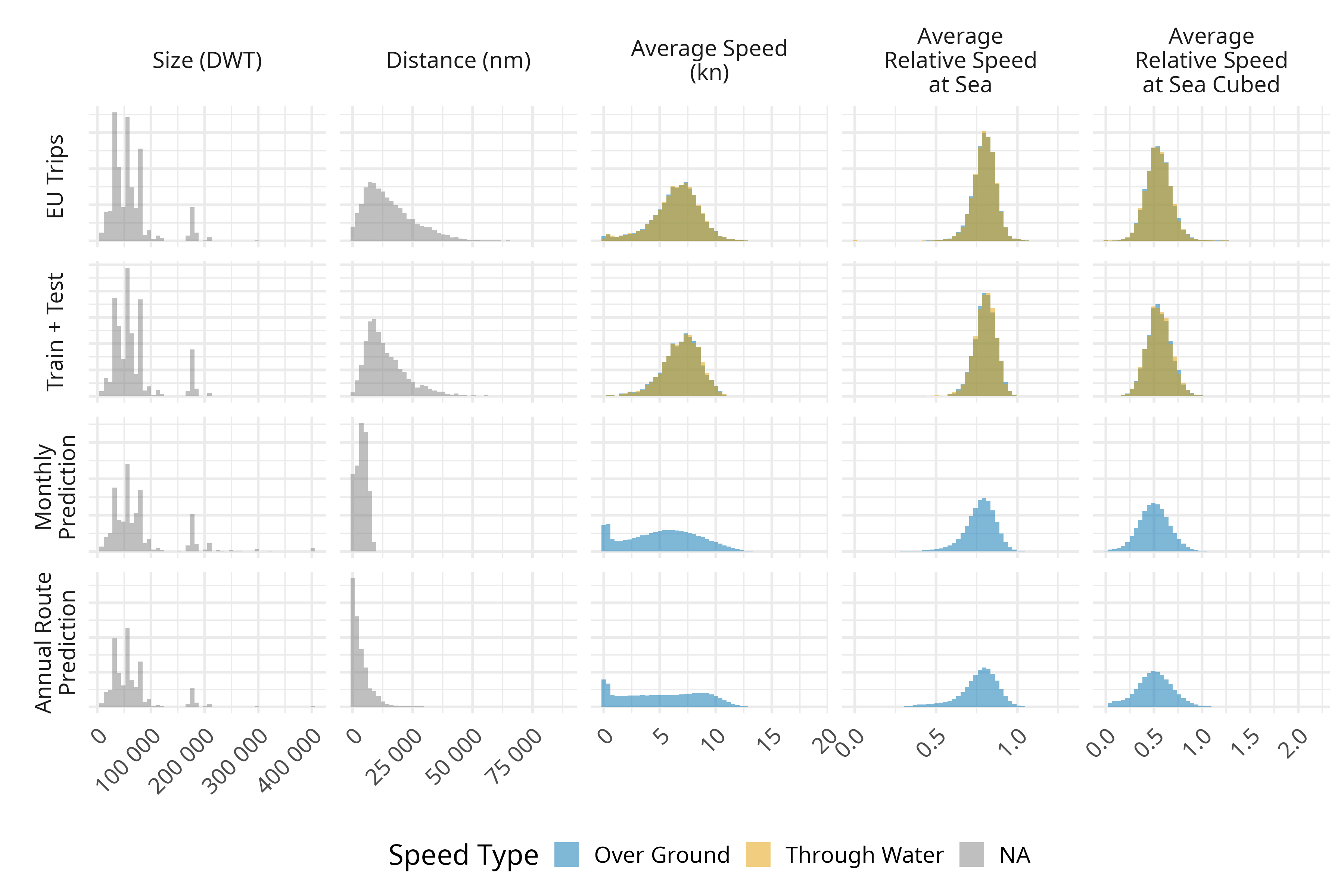}
    \caption{Histograms of variables of interest relevant for prediction applications across data subsets, dry bulk ships. Axes vary.}
    \label{fig:datasupporthistograms}
\end{figure}

For the prediction exercises, we aggregate all global \ac{AIS} movements at two further levels: by ship and month, and by ship, year, and origin--destination country pair (\appref{app:disaggregation}). Both yield an order of magnitude more observations than the training data (\autoref{tab:obscounts}), with substantially shorter distances and more mass at slower average speeds, though less so at sea, while ship sizes are similar (\autoref{fig:datasupporthistograms}). Over half of all ships appear in the \ac{MRV} data for at least one year. We use \ac{SOG} throughout and current-corrected \ac{STW} only for robustness checks; the figure shows the difference is minimal.

\section{Empirical Models}
\label{sec:MLmodels}

We consider a range of models that vary in the degree to which they rely on the underlying engineering and physics. At one extreme, we take as our benchmark the \ac{IMO} approach (\autoref{subsec:engcalc}), which calibrates, rather than estimates, parameters of the engineering model (hereafter the calibration). Any inappropriately calibrated parameters guarantee prediction error, as do any factors affecting fuel consumption for which a tractable engineering model does not exist. At the opposite extreme lies a family of flexible \ac{ML} models that are purely data-driven and need not take any account of the physics. Representing this extreme, we replicate the fuel efficiency model of \textcite{clarke2023co2}, which targets per-distance fuel consumption, and in doing so abstracts from the key physical nonlinearity of the power model.

We contribute models that lie between these two extremes, blending theory and data in varying degrees and configurations. Actual fuel consumption can reasonably be expected to diverge from what theory predicts in various ways, such as through variation in parameter values, slight deviations in functional form, or the contribution of factors that are not captured in the engineering model. We are ex ante agnostic as to which of these are important, as this undoubtedly varies with the specific context and application. 

To discipline the potentially infinite space of blended models, we focus on those that are separable in structural (physics-informed) and nuisance (form-free) components. We write the family of models as:
\begin{equation*}
    Y_i = f(X_{f,i}) \circ g(X_{g,i}),
\end{equation*}
where $f$ is the structural component, a function of the physics-informed features $X_f$, and $g$ is the nuisance component, a function of the nuisance features $X_g$. The $\circ$ operator stands for either addition or multiplication, as we consider both types of separability. Different specifications of $f$ and $g$ then vary the flexibility of the model.

The result is the taxonomy in \autoref{fig:modeldiagram}, which organizes the models by how the two components compose and which of them departs from the calibration. The calibrated model is the special case $f = f_c$ with no nuisance component; it sits at the top of the diagram as the fully theory-constrained model. Grayed-out models lack a physical basis and are shown only to complete the taxonomy.\footnote{The additive calibration hybrid, $f_c + g$, is well defined but weakly motivated: the factors a nuisance component absorbs on the main engine, such as weather, fouling, and load, scale consumption rather than add to it. We therefore estimate the calibration hybrid only in its multiplicative form.} We describe the remaining models below.

\begin{figure}
    \centering
    \begin{adjustbox}{width=0.7\textwidth}
        \resizebox{\textwidth}{!}{%
\begin{tikzpicture}[x=1mm, y=1mm, font=\small,
    box/.style={draw, rounded corners=1.5pt, align=center,
        minimum width=36.0mm, minimum height=17.0mm, inner sep=1.2mm},
    ghost/.style={box, densely dashed, draw=black!40, text=black!45},
    detached/.style={box, densely dotted},
    fan/.style={circle, draw, fill=black, inner sep=0pt, minimum size=1.4mm},
    arr/.style={-{Latex[length=2mm]}, semithick, black!75},
    tag/.style={font=\scriptsize\itshape, inner sep=1pt, fill=white,
        align=center,
        text=black!70}]
\node[box] (calc) at (0.0, 0.0) {Calibration\\[0.3ex] $f_c(X_f)$};
\node[box] (m_scaled) at (-34.0, -30.0) {Calibration\\power law\\[0.3ex] $c\,f_c(X_f)^{\beta}$};
\node[box] (a_scaled) at (34.0, -30.0) {Scaled\\calibration\\[0.3ex] $c + \beta f_c(X_f)$};
\node[ghost] (m_struct) at (-62.0, -64.0) {$c\prod_k f_k^{\beta_k}$};
\node[box] (m_nuis) at (-21.0, -64.0) {Calibration\\hybrid\\[0.3ex] $f_c(X_f)\,g(X_g)$};
\node[box] (a_struct) at (21.0, -64.0) {Structural\\[0.3ex] $c + \sum_k \beta_k f_k$};
\node[ghost] (a_nuis) at (62.0, -64.0) {$f_c(X_f) + g(X_g)$};
\node[ghost] (m_joint) at (-42.0, -90.0) {$g(X_g)\prod_k f_k^{\beta_k}$};
\node[box] (a_joint) at (42.0, -90.0) {Structural\\hybrid\\[0.3ex] $\sum_k \beta_k f_k + g(X_g)$};
\node[box] (flexible) at (0.0, -114.0) {Flexible\\[0.3ex] $m(X_f, X_g)$};
\node[detached] (efficiency) at (0.0, -136.0) {Efficiency\\[0.3ex] $d\,m(X)$};
\draw[arr] (calc) -- node[tag, pos=0.50, left] {multiplicative\\separability} (m_scaled);
\draw[arr] (calc) -- node[tag, pos=0.50, right] {additive\\separability} (a_scaled);
\draw[arr] (m_scaled) -- node[tag, pos=0.58, left] {structural\\flexibility} (m_struct);
\draw[arr] (m_scaled) -- node[tag, pos=0.58, right] {nuisance\\flexibility} (m_nuis);
\draw[arr] (a_scaled) -- node[tag, pos=0.58, left] {structural\\flexibility} (a_struct);
\draw[arr] (a_scaled) -- node[tag, pos=0.58, right] {nuisance\\flexibility} (a_nuis);
\draw[arr] (m_nuis) -- (m_joint);
\draw[arr] (m_struct) -- (m_joint);
\draw[arr] (a_struct) -- (a_joint);
\draw[arr] (a_nuis) -- (a_joint);
\draw[arr] (m_joint) -- (flexible);
\draw[arr] (a_joint) -- (flexible);
\draw[arr, black!45] (-96.0, 8.5) -- (-96.0, -144.5);
\node[anchor=south west, font=\footnotesize\itshape, text=black!60] at (-98.0, 10.5) {constrained};
\node[anchor=north west, font=\footnotesize\itshape, text=black!60] at (-98.0, -146.5) {naive};
\end{tikzpicture}%
}

    \end{adjustbox}
    \caption{Model architectures. Models in grayed-out boxes are not estimated (see text).}
    \label{fig:modeldiagram}
\end{figure}
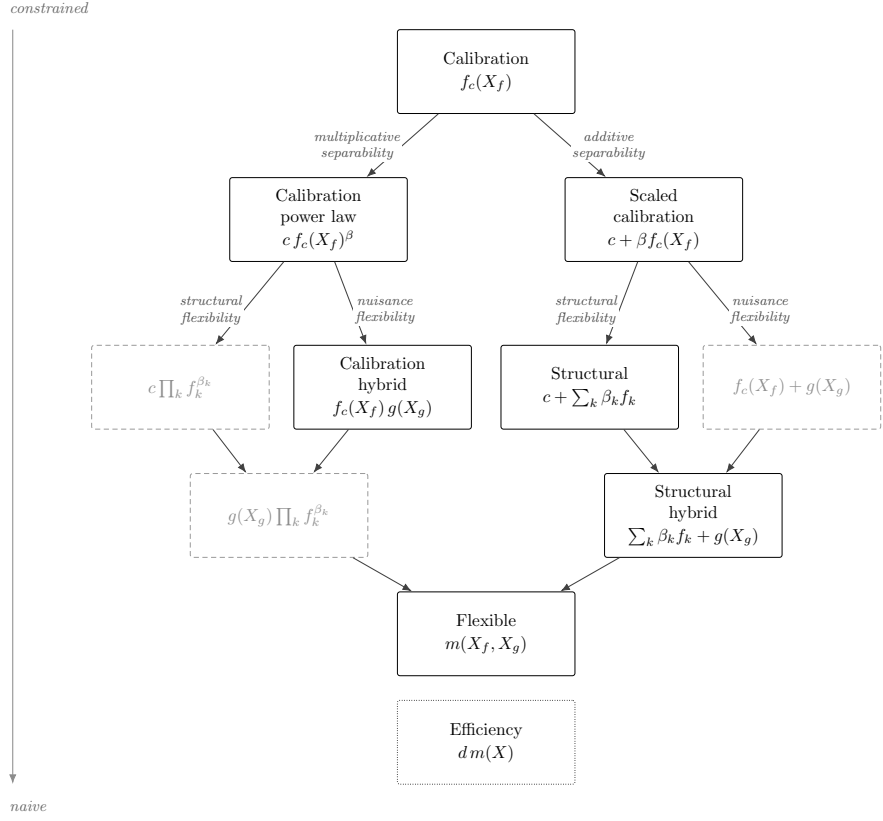

\subsection{Nuisance Component}
\label{sec:nuisance}

We refer to any flexible \ac{ML} component of a model as the nuisance component, following the double/debiased machine learning literature \parencite{chernozhukov2018double}. Nuisance components vary in learner and in feature set. We consider both linear (ordinary least squares, lasso, ridge) and tree-based (random forest, gradient boosting machine, CatBoost) learners.\footnote{Neural networks cannot be reliably estimated given our limited sample size.} The `Efficiency' class of \textcite{clarke2023co2} and the `Flexible' class, at the bottom of \autoref{fig:modeldiagram}, consist of nuisance components alone. The `Flexible' class lies above the `Efficiency' class because its target and features are consistent with the physics. The `Efficiency' class assumes fuel efficiency is fixed within the ship-year, so that average speed and draft suffice as features. Because fuel efficiency is nonlinear in speed, the `Flexible' class instead aggregates the Admiralty terms of \eqref{eqn:admiralty} hour by hour, so that its speed features measure the work done to move the ship through the water rather than its average speed.

Beyond these physics-derived terms, the `Flexible' feature set includes ship characteristics such as size and rated engine power, activity measures such as distance traveled and number of trips, and indicators of tracking data quality such as the share of interpolated observations at sea (\appref{app:MLfeatures}). Given the hundreds of characteristics available and the potentially infinite variations on aggregating tracking data, we select this single broad set based on domain knowledge and feature importance results from preliminary analyses. Flexible learners leave the data to determine how these features interact. Including components of the engineering formulas lets a model deviate from the exact functional form while remaining consistent with the physics, but may also induce overfitting or instability where the components are highly correlated. For linear learners we fit in logs, so that the product of terms in the engineering calculation becomes a sum and the linear learners loosely mimic the separability of the structural equation, while tree-based learners, being invariant to monotone transformations of the features, need none.

\subsection{Structural Component}
\label{sec:structural}

Structural components derive functional form from the engineering model, but estimate parameters rather than calibrating them. Their flexibility is set by how many parameters are freed and at what level of the model. The most constrained versions free just two parameters that modulate the calibrated calculation $f_c$ as a whole. On the additive branch of \autoref{fig:modeldiagram}, the `Scaled calibration' class estimates a scale and intercept, while on the multiplicative branch, the `Calibration power law' class estimates a scale and exponent.

Finer structure comes from estimating the coefficients of the engine efficiency curve \eqref{eqn:sfc} rather than taking them from \textcite{faber2020fourth}. Combining \eqref{eqn:admiralty}--\eqref{eqn:annual_fc} and dividing through by the reference power and base \ac{SFC}, main engine fuel consumption satisfies
\begin{equation*}
    \frac{\mathrm{FC}_{\text{ME},i}}{W_{\text{ME},\textit{ref},i}\,\mathrm{SFC}_{\text{base},i}}
    = c + \beta_1 \sum_h w_{ih} + \beta_2 \sum_h w_{ih}^2 + \beta_3 \sum_h w_{ih}^3,
\end{equation*}
which is linear in the three aggregate load terms $\sum_h w_{ih}^k$, the $f_k$ of the `Structural' class in the diagram. The calibration is the special case $c = 0$ and $(\beta_1, \beta_2, \beta_3) = (1.280, -0.710, 0.455)$. For this class the target is the normalized main engine consumption, obtained by subtracting the calibrated auxiliary engine and boiler consumption from reported fuel consumption as per \eqref{eqn:fc_inst}. The normalization thus imposes the reference power, base \ac{SFC}, and auxiliary and boiler components, and we estimate the coefficients by ordinary least squares.

Between the scaled calibration and the unrestricted polynomial lie restrictions motivated by the physics, which we treat as variants of the `Structural' class. Our main variant sets the slope of the efficiency curve to zero at the design load of \eqref{eqn:sfc}, where the calibrated curve has its minimum, leaving two free coefficients. Other departures from the calibrated curve, such as allowing asymmetry about its minimum, are also possible.

\subsection{Hybrid Models}
\label{sec:hybridmodels}

Hybrid models combine a structural and a nuisance component, and the structural component can be hybridized more or less flexibly by constraining its coefficients: pinned at their calibrated values, scaled, or free. We do not estimate every combination, only one on each branch of \autoref{fig:modeldiagram}. The `Calibration hybrid' takes the multiplicative branch and pins the calibration at its coefficient of one, so the nuisance component learns only the correction to it. This boosts a theory-based model with a data-driven one: the calibration is the first stage, and the nuisance learner corrects its errors.\footnote{\textcite{almeida2023machine} boost parametric option pricing models in the same way, training a neural network on the model-implied pricing errors. Their correction sees the same inputs as the parametric model, whereas we withhold key features, as discussed below.} The `Structural hybrid' takes the additive branch and frees both components, with the polynomial of \autoref{sec:structural} as the structural component. Since a flexible nuisance can absorb variation that the structural coefficients would otherwise explain, the two are identified by orthogonalization: following \textcite{robinson1988root} and \textcite{chernozhukov2018double}, the coefficients are estimated from the variation in the target and in the load terms that the nuisance features do not predict (\appref{app:hybrid}).

A more subtle design choice is the selection of nuisance features. For the imposed structure to bind, the nuisance component cannot contain features that would let it re-learn the structural relation in a more flexible form. In both hybrids, the nuisance feature set $X_g$ excludes every variable derived from speed: the calibrated value, the aggregated speed and power terms, and the rated speed, together with the at-port fraction, which changes with time at sea at fixed distance. This exclusion delivers the hybrids' defining property: because the nuisance has no speed argument, the model's response to speed is exactly that of its structural component, whatever the nuisance learns. Proposition~\ref{prop:preservation} in \appref{app:hybrid} states this formally: the multiplicative hybrid preserves the elasticity of the calibration, and the additive hybrid the level derivative of its polynomial, for every realization of the nuisance fit. If the structural component is misspecified, the error in the preserved object is bounded by the structural error alone, whatever the data, sample size, or learner (Corollary~\ref{cor:preservation}). The composition sets which invariant is preserved, and the nuisance contaminates the other: a multiplicative nuisance rescales the level derivative, and an additive one shifts the elasticity.

The exclusion rests on different grounds in the two hybrids. In the structural hybrid it is inherent to the partially linear model: the load terms are the regressors of interest, and the controls are the ship characteristics and data-quality indicators. In the calibration hybrid nothing is identified, and the exclusion is a modeling restriction that we postulate: a nuisance that is to carry over to a counterfactual should not see the variables the counterfactual would alter, so we withhold from it every variable that is not invariant to a counterfactual shift in speed. Its nuisance therefore retains distance, number of trips, and mean draft, which a fixed-distance counterfactual leaves unchanged, alongside the ship characteristics and data-quality indicators (\appref{app:MLfeatures}). The restriction guarantees the preserved elasticity, but the transfer of the nuisance term to the counterfactual rests on the assumption that the retained features' relation to fuel consumption does not shift with speed, which cannot be verified outside the training support.

\subsection{Implementation}
\label{sec:implementation}

We implement multiplicatively separable models by regressing in logs, using $\log(1+Y)$ to retain values of zero.\footnote{We use Python's scikit-learn library for all learners except CatBoost, for which we use the library of the same name \parencite{prokhorenkova2018catboost}, and LightGBM in the orthogonalized estimator. Implementation details, including software versions, are in \appref{app:MLfeatures}.} Additively separable models, namely the structural and efficiency classes, are fit in levels. The calibration hybrid is fit as a nuisance learner on the log residual $\log Y_i - \log f_c(X_{f,i})$, with any of the learners above. Retransformation bias in predictions from the log-scale models is corrected with the smearing estimate of \textcite{duan1983smearing} (\appref{app:MLfeatures}). As a robustness check that avoids retransformation, we also fit the `Flexible' class by Poisson and gamma pseudo-maximum likelihood \parencite{silva2006log}.

Where relevant, hyperparameters are tuned by grid search with 5-fold cross-validation on the training data, selecting the parameter set with the lowest mean squared error averaged across folds (\appref{app:MLfeatures}).\footnote{The Poisson and gamma pseudo-maximum likelihood fits are instead tuned on their own deviance.} We select on squared error because the training criterion should match the use of the predictions \parencite{gneiting2011making}: our applications aggregate predictions, so the conditional mean, which squared error targets, is the object of interest.

\section{Model Evaluation}
\label{sec:MLevaluation}

Accurate prediction and valid counterfactual analysis require models that generalize beyond the training data. Following \textcite{andrews2025transfer}, we distinguish \textit{within-domain error}, the prediction error on data drawn from the same distribution as the training set, from \textit{transfer error}, the prediction error when the test distribution differs from it. In their application to risk preferences, \textcite{andrews2025transfer} find that black-box algorithms fit slightly better within domain while the economic models transfer more reliably. Our evaluation is similar in spirit, and we present a suite of experiments to assess the performance of models within and out of domain. This section describes the results; \autoref{sec:MLexplanation} examines why the models perform as they do.

Our target applications predict emissions aggregated at various levels rather than at the ship level, so we report error as the \ac{RTE}, the total prediction error as a share of total reported consumption. The disadvantage of an aggregate measure is that offsetting errors within a sample cancel, which can mask ship-level error. To show sampling variability, we present 95\% percentile intervals from a ship-clustered bootstrap of the test sample with the fitted models held fixed. For reference we report \ac{RMSE}, which is immune to this cancellation, in \appref{app:MLfeatures}. We present results for dry bulk ships in the main text. Unless noted, container ship results are qualitatively similar (\appref{app:container_evaluation}).

The remainder of the section proceeds as follows. \autoref{sec:oos} evaluates the models \textit{out of sample}, on a later year drawn from the same population, which measures within-domain error. \autoref{sec:oosupport} evaluates the models \textit{out of support}, where the test data lie beyond the range of the training data along at least one dimension, which measures transfer error. We train on three quartiles of a policy-relevant variable and predict the fourth (\autoref{sec:quartiles}), train on one ship type and predict the other (\autoref{sec:domain}), and shift speed counterfactually beyond the observed range (\autoref{sec:elasticity}). \autoref{sec:comparison} summarizes the trade-off between within-domain fit and transfer.

\subsection{Out of Sample}
\label{sec:oos}

We train on 2019--2020 and test on 2021, a roughly two-thirds/one-third split over which reported fuel consumption is stable (\autoref{tab:calstats}). The calibration overpredicts reported consumption, with an \ac{RTE} of 6.5\% on the test set (\autoref{fig:twoway_cal}). It also overpredicts in every year of the training data, so the bias is systematic rather than a feature of the test year.

\autoref{fig:samplescores} shows the out-of-sample \ac{RTE} of every model specification and sub-variant, arranged loosely by class from the most to the least theory-constrained. Every estimated model lies within a few percentage points of the true value, a clear improvement from the calibration. The `Structural' models tend to perform best in general, while the pure \ac{ML} models underpredict. The Poisson and gamma pseudo-maximum likelihood fits do so by the same margin as the log-linear fits, suggesting the shortfall is not retransformation bias. The calibration-anchored models almost all overpredict by similar amounts. Perhaps surprisingly, no hybrid improves on its structural counterpart, so adding a flexible nuisance component buys no aggregate accuracy within domain. Variation in performance is typically larger across model architectures than across nuisance learners. The intervals of the data-driven and calibration-anchored groups generally do not overlap, so the sign pattern is not an artifact of the test draw.

\begin{figure}
    \centering
    \includegraphics[width=0.9\textwidth]{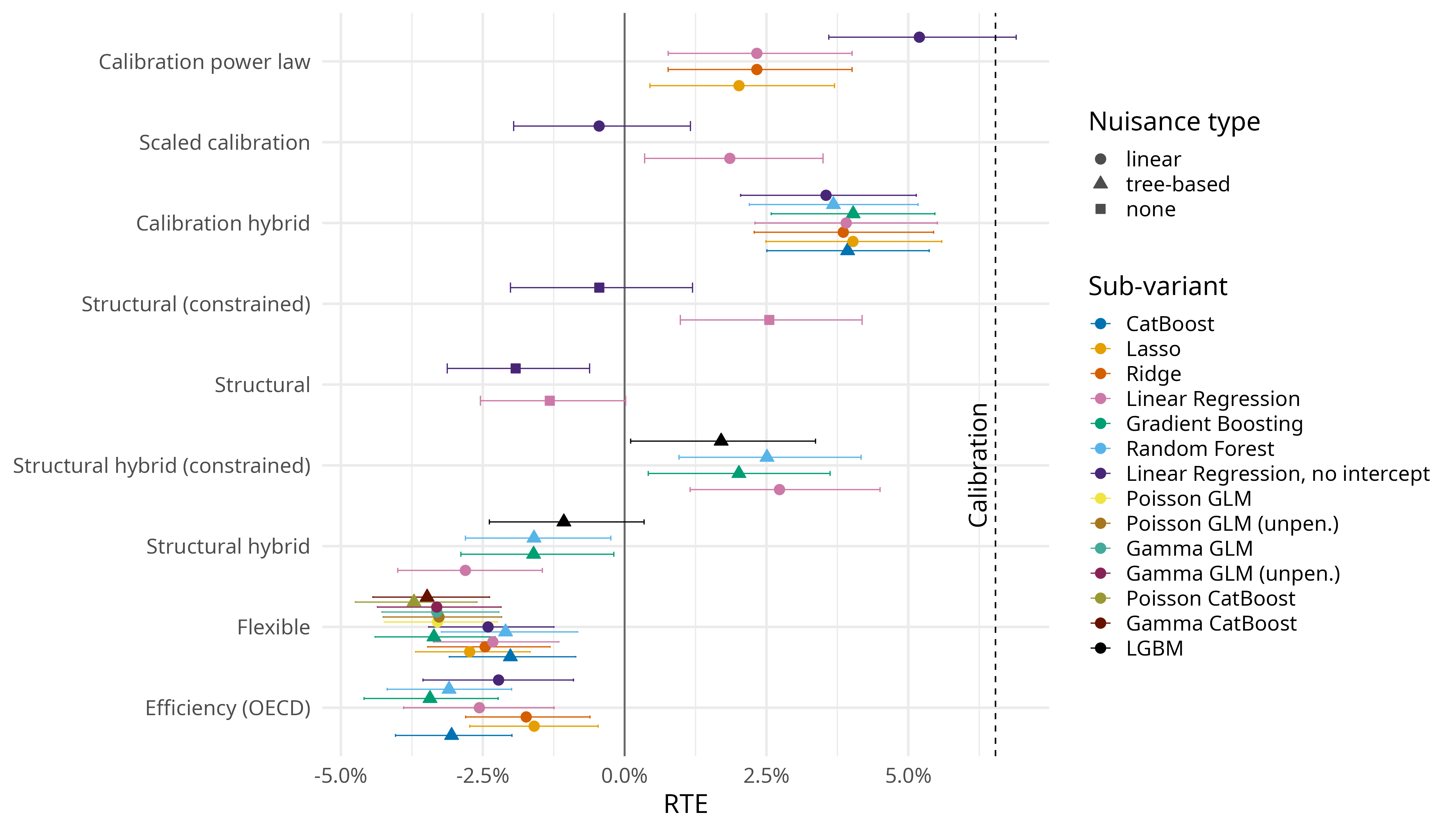}
    \caption{Out-of-sample \ac{RTE} by model specification and sub-variant, dry bulk ships. The dashed line is the calibration. Bars are 95\% ship-clustered bootstrap intervals over the test sample, fitted models held fixed.}
    \label{fig:samplescores}
\end{figure}

These patterns are similar for container ships, where all errors are slightly larger (\appref{app:container_evaluation}).

\subsection{Out of Support}
\label{sec:oosupport}

Out of support, the test data lie beyond the range of the training data along at least one dimension, so a model must extrapolate. We assess this in three ways. First, we hold out one tail of the observed sample along a continuous dimension, training on three quartiles of a variable and predicting the fourth. Second, we train on one ship type and predict the other, where the test data lie outside the training support on nearly every dimension at once. In both, outcomes are observed, so transfer error is measured directly. Third, we shift speeds counterfactually beyond the observed range. No outcomes exist, so predictions are compared against the engineering benchmark instead.

\subsubsection{Quartile Segmentation}
\label{sec:quartiles}

We split the sample into quartiles of a variable, train on three quartiles, and predict the fourth. To illustrate extrapolation performance, we report prediction for the two outer quartiles. The first two splits are by ship size, measured by \ac{DWT}, and by distance traveled, two dimensions along which subsets of the fleet differ in disaggregation. Because speed is a key policy margin, the third split is by the mean of the cube of hourly at-sea speed relative to rated speed, the quantity that drives consumption.\footnote{Splitting on average speed would be non-physical, grouping a ship that travels at a constant moderate speed with one that is stationary half the time and fast the other half.}

The left panel of \autoref{fig:support_quartiles} shows the size split. Tree-based nuisance learners tend to show a larger discrepancy between the two outer quartiles than linear ones, overpredicting for the smallest ships and underpredicting for the largest, and the gap is widest for the `Flexible' and `Efficiency' pure \ac{ML} classes. The structural fits, whose response to size is fixed by the physics, tend to perform better. The distance split in the middle panel shows a similar pattern, except for the `Efficiency' class: its target is consumption per unit distance and the prediction is multiplied by distance, so it fits the distance split by construction.

\begin{figure}
    \centering
    \includegraphics[width=\textwidth]{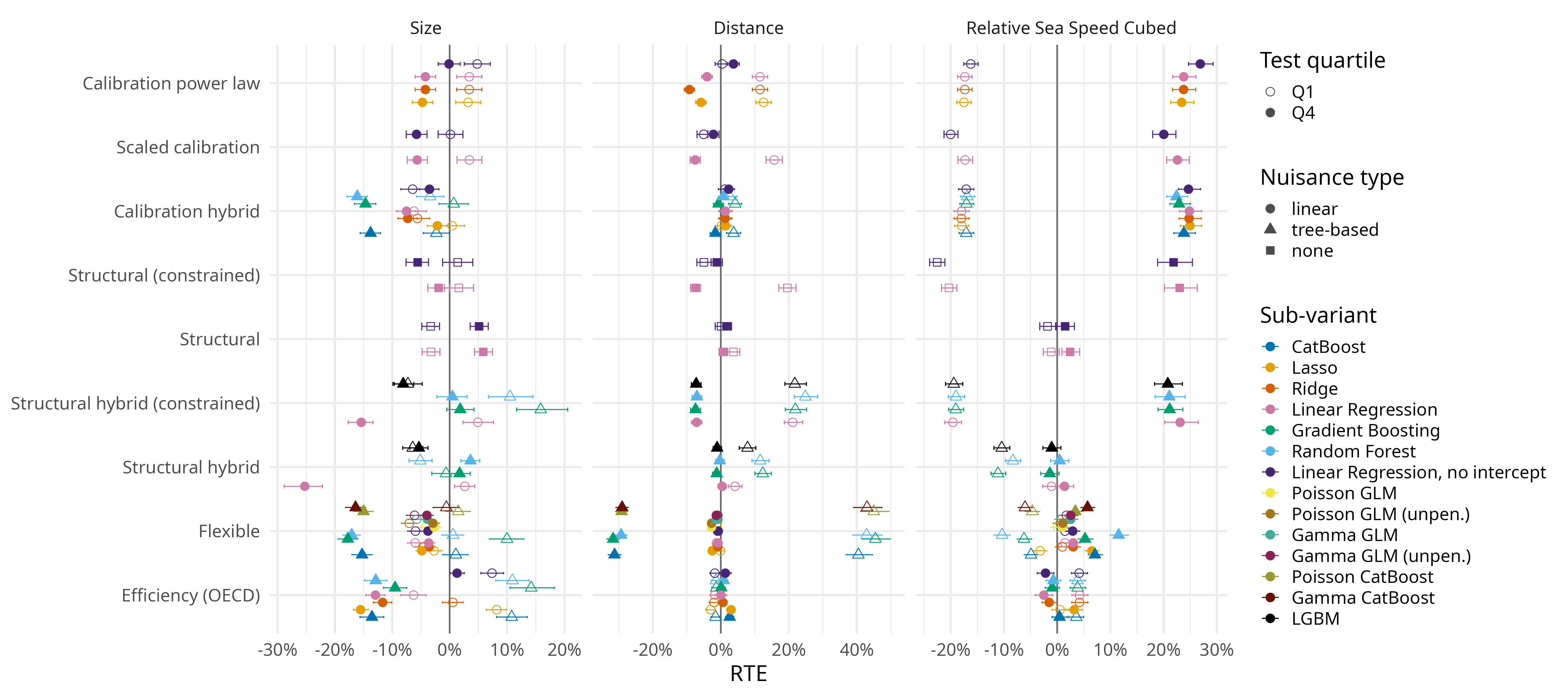}
    \caption{Out-of-support \ac{RTE} for three quartile splits, dry bulk ships. Each model is trained on three quartiles and tested on the fourth. Bars are 95\% ship-clustered bootstrap intervals over the test quartile, fitted models held fixed.}
    \label{fig:support_quartiles}
\end{figure}

The speed split in the right panel shows much more mixed results. For dry bulk ships, models that fix the speed relation at its calibrated form miss on both outer quartiles, while models that estimate the speed relation, the unrestricted structural fits and the pure \ac{ML} classes, do considerably better. The sign pattern of the quartiles indicates the slope of each model's speed relation relative to the data: a steeper relation underpredicts the slowest quartile and overpredicts the fastest, a flatter one the reverse. The `Efficiency' class is the one model that implies a flatter relation than the data, as imposed by its fixed-efficiency assumption. The speed split differs substantially for container ships, where the calibration-anchored models transfer and the unrestricted structural fit underpredicts the fastest quartile (\autoref{fig:support_quartiles_container}). The two fleets evidently differ in more than the range of their features, which is the subject of the next exercise.

\subsubsection{Ship-Type Transfer}
\label{sec:domain}

We next train on one ship type and predict the other. Container ships travel faster than dry bulk ships and differ in hull form, engine configuration, and operating pattern, so the held-out fleet lies outside the training support on nearly every feature at once, and the relation between features and consumption need not be the same on both sides. \autoref{fig:domain} shows the \ac{RTE} in both directions.

Here we observe that transfer error generally grows with flexibility. The two-parameter models that scale or exponentiate the calibration transfer within a few percent in either direction, and the constrained structural fits are close behind. The unrestricted structural fits and the structural hybrids do worse, overpredicting for container ships when trained on dry bulk ships. The pure \ac{ML} classes transfer worst: their tree-based learners overpredict for dry bulk ships and underpredict for container ships by large margins, and the linear learners of the `Flexible' class underpredict in both directions.

\begin{figure}
    \centering
    \includegraphics[width=0.9\textwidth]{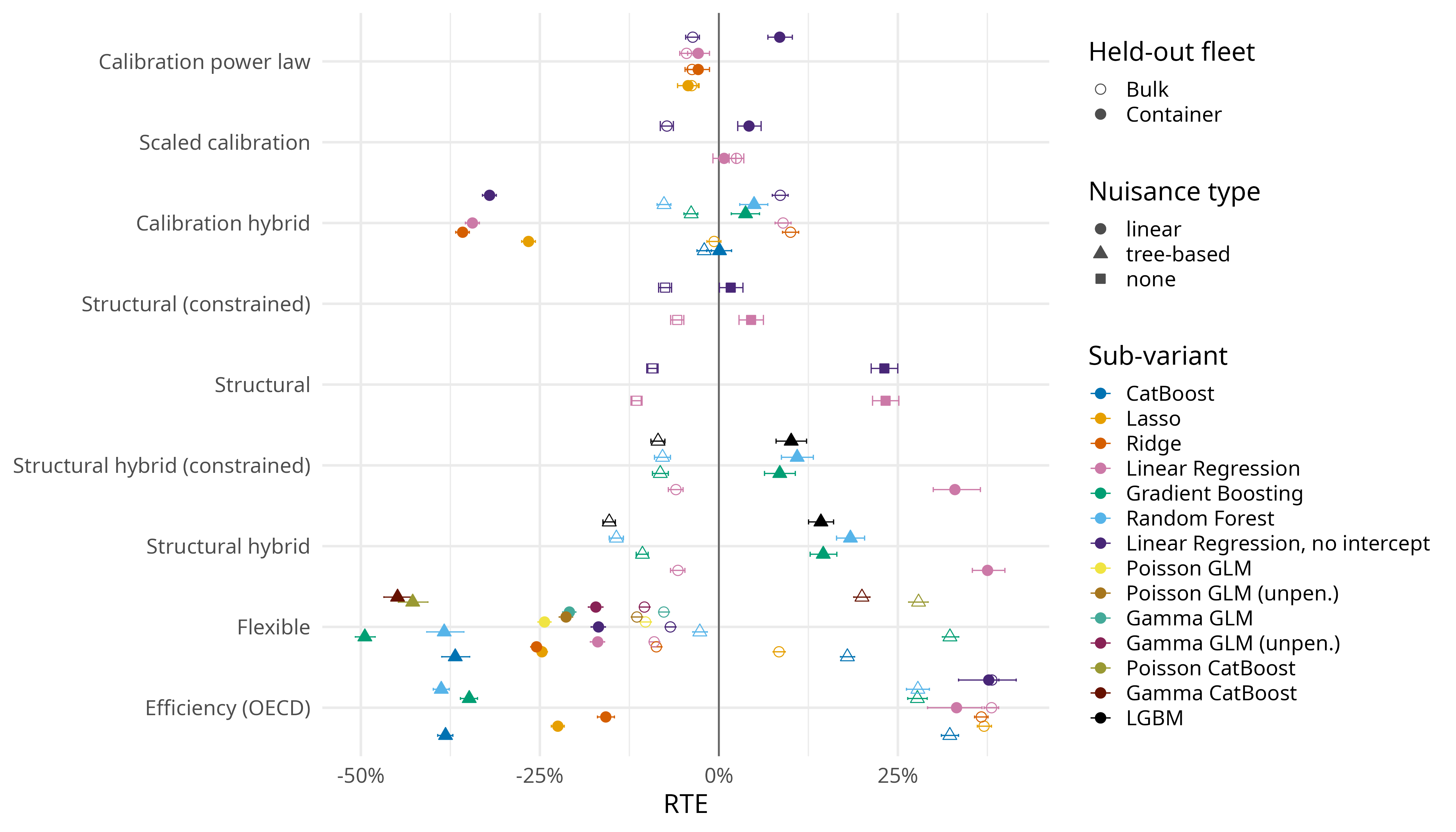}
    \caption{Transfer \ac{RTE} across ship types by model specification and sub-variant. Each model is trained on one fleet and tested on the other. Bars are 95\% ship-clustered bootstrap intervals over the test fleet, fitted models held fixed.}
    \label{fig:domain}
\end{figure}

\subsubsection{Speed Counterfactual Prediction}
\label{sec:elasticity}

The quartile split tests each model's speed relation within the range of historically observed speeds. However, speed regulation and market shocks can move speeds beyond that range. To evaluate such larger variation, we construct counterfactual data by adjusting the hourly speed of every at-sea observation by a common percentage.\footnote{For greater realism, we adjust only observations at sea, where the operator chooses speed, and leave low speeds where ships are likely constrained by congestion or safety unchanged.} Distance is held fixed, so time at sea varies inversely with speed. Every feature derived from speed, including the calibration itself, is recomputed from the adjusted hourly data, while distance, number of trips, draft, and ship characteristics are unchanged. Each model, trained on the full sample, then predicts consumption for every ship-year at each adjustment, and we sum the predictions.\footnote{The exercise is a variant of the partial dependence plot of the \ac{ML} literature, with the adjustment applied to the hourly tracking data rather than to the aggregate feature.} No outcomes exist for these counterfactuals, so we report the percentage change in predicted aggregate consumption from its unadjusted value and compare each model's response with the `Scaled calibration' (without intercept) as the theoretical benchmark.\footnote{This sub-variant multiplies the calibration by a single estimated scale, so its percentage response to a speed adjustment is exactly that of the calibration, whatever the estimated scale.}

\autoref{fig:elasticity} shows the predicted response for adjustments ranging from $-20\%$ to $+10\%$ by model class; the dashed curve in each panel is the benchmark. At one extreme, the `Calibration power law' and `Calibration hybrid' classes follow the benchmark quite closely, as these models strongly impose the \ac{IMO} power law. In stark contrast, the pure \ac{ML} classes exhibit a nearly flat speed response, with some `Efficiency' models even predicting a decrease in fuel consumption at higher speeds. The structural models fall in-between; the unconstrained variants exhibit both an attenuated slope and an inverted curvature relative to the benchmark. Constraining the slope of the \ac{SFC} component to zero at its calibrated design load of 0.78 disciplines this and aligns the structural models much more closely with the benchmark. Both hybrid classes follow their structural component in both respects, which depends on the exclusion of speed-related features from the nuisance component (\autoref{sec:hybridmodels}); retaining them lets the nuisance re-learn part of the speed relation and attenuates the speed response.
\begin{figure}
    \centering
    \includegraphics[width=\textwidth]{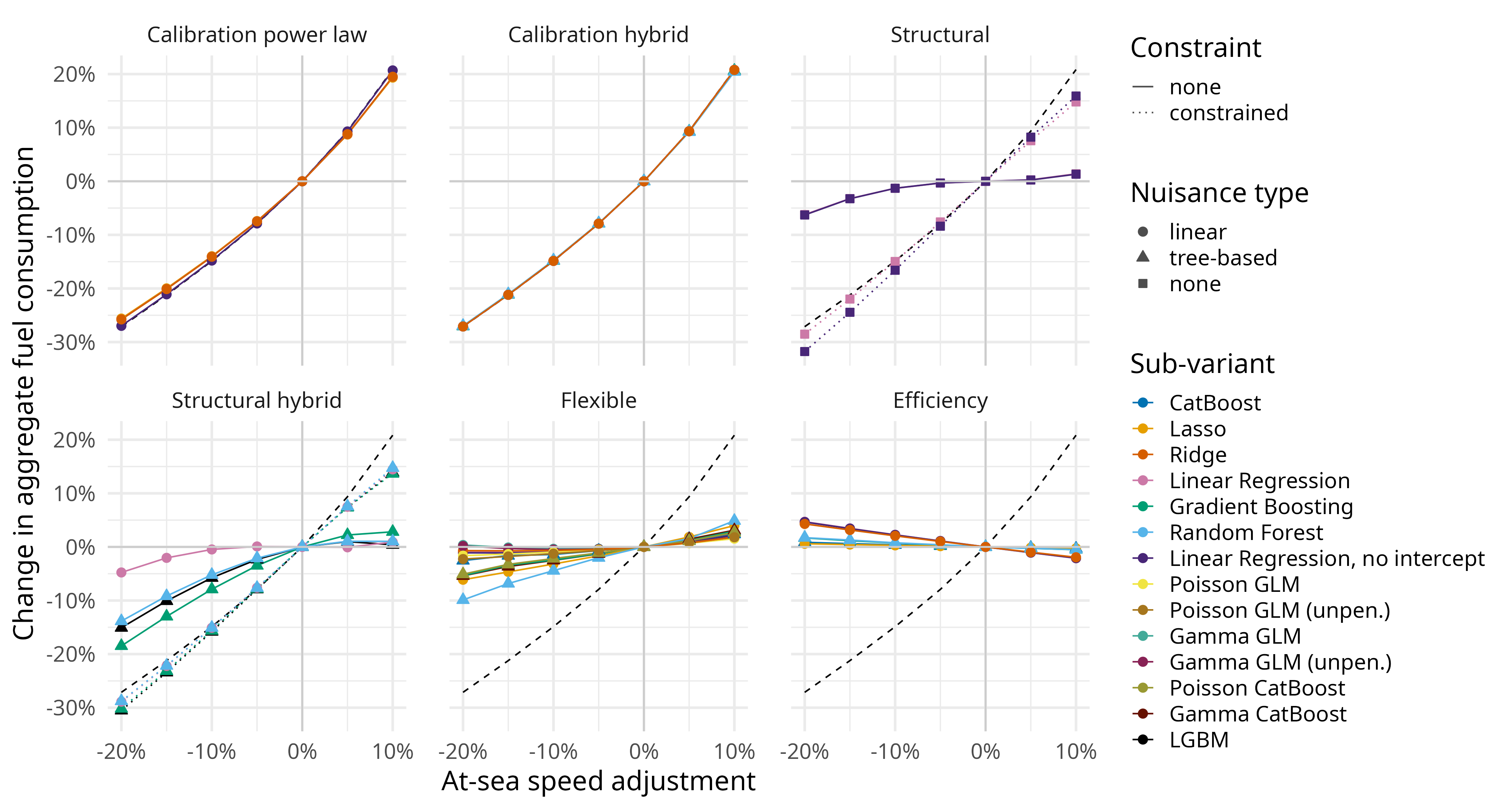}
    \caption{Predicted change in aggregate fuel consumption in the speed-adjusted counterfactual exercise, dry bulk ships. The dashed line is the theoretical benchmark.}
    \label{fig:elasticity}
\end{figure}

For container ships the models' relative performance follows the same pattern, but with a more pronounced response to speed adjustments, closer to the benchmark (\autoref{fig:elasticity_container}).

\subsection{Comparison}
\label{sec:comparison}

The previous exercises evaluated out-of-sample and out-of-support performance in isolation, but in many applications both matter: a disaggregated inventory must be accurate in level, and an assessment of a speed policy must also capture the response to speed. \autoref{fig:bias_elasticity} plots the two axes against each other, the out-of-sample \ac{RTE} of \autoref{sec:oos} against the speed elasticity of aggregate consumption implied by the counterfactual exercise of \autoref{sec:elasticity}. We condense the speed curve comparison into a single measure using the arc elasticity at $\pm 10\%$. The vertical line marks zero bias and the dashed line the elasticity of the calibration, so their crossing is the ideal. This highlights that some models are strictly dominated, but for others there is a trade-off.

The models fall into two groups. The calibration-anchored classes and the hybrids with a constrained structural component reproduce the calibration's speed response but overpredict, while the pure \ac{ML} classes and the unrestricted structural fits are less biased but have lost most of the speed response. Within the second group the `Efficiency' class is dominated, with a bias similar to the `Flexible' class and no speed response at all. Between the groups the trade-off is between level and slope, and across the nuisance learners of a class the points move along the bias axis but hardly at all along the elasticity axis: the hybrids inherit the speed response of their structural component, and the nuisance component shifts only the level.

The exceptions are the two no-intercept fits, the `Scaled calibration' and the constrained `Structural' fit, which sit close to the ideal on both axes. Each is nearly unbiased out of sample and responds to speed nearly as the calibration does, the first by construction and the second by estimation. The picture is similar for container ships (\autoref{fig:bias_elasticity_container}), where the same two fits lie near the ideal.

\begin{figure}
    \centering
    \includegraphics[width=0.8\textwidth]{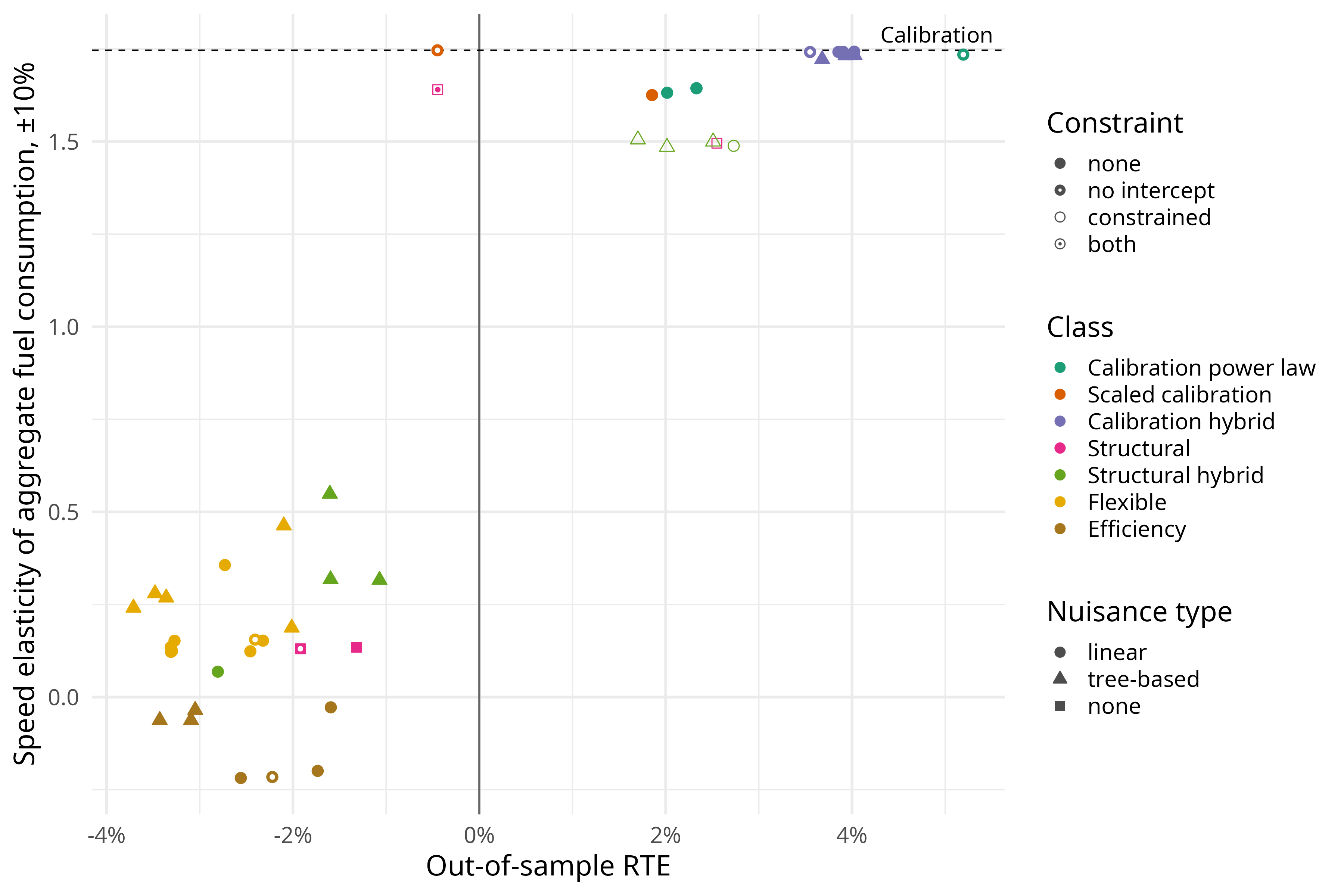}
    \caption{Out-of-sample \ac{RTE} against the speed elasticity of aggregate fuel consumption, dry bulk ships. The elasticity is the arc elasticity over the $\pm 10\%$ counterfactual; the dashed line denotes the calibration.}
    \label{fig:bias_elasticity}
\end{figure}

\section{Performance Explanation}
\label{sec:MLexplanation}

Our results are largely in line with the previous literature and general intuition \parencite{nevo2010taking,andrews2025transfer}: pure \ac{ML} models perform well within domain but fail out of support, while theory-constrained models are biased within domain but generally transfer better. The correspondence with a structure-versus-\ac{ML} narrative is not exact, however: the unrestricted structural fits, which keep the engineering form but estimate its coefficients, lost most of the response to speed in the counterfactual, as the pure \ac{ML} models did. In this section we explore why models failed where they did, adjudicating three candidate explanations: extrapolation beyond the training support, measurement error in the features, and the nonlinearity of the structural model.

\subsection{Out-of-Support Extrapolation}
\label{sec:mlexplanation}

The extrapolation behavior of tree-based learners is well known. A regression tree partitions the feature space into cells and predicts a constant within each, and an ensemble of trees, whether bagged or boosted, is a sum of such piecewise-constant functions \parencite{hastie2009elements}. Beyond the range of the training data along any feature, every observation falls into the outermost cell along that dimension, so the model returns the value it learned at the edge of the support.\footnote{Neural networks fail in a similar fashion but instead follow an affine function out of support \parencite{xu2021how}.} In our results the `Flexible' and `Efficiency' models employing tree-based learners performed noticeably worse than the linear ones out of support.\footnote{For the distance split, the `Efficiency' models performed better by construction, as mentioned previously. For the hybrid classes, whose nuisance components also use tree-based learners, the pattern is less clean.} We might therefore suspect the trees to be the problem. 

To illustrate the impact of the learner extrapolation behavior in our setting, we present an exercise in which we replace the dependent variable with the fuel consumption predicted by the \ac{IMO} calculation. This isolates the effect of the learner by removing measurement error. \autoref{fig:sim_flex} shows the outcome of the same speed counterfactual exercise as in \autoref{sec:elasticity}. The linear learners, whose feature set includes the calibrated value itself, recover the benchmark exactly. The tree-based learners depart from it at the edges of the range to varying degrees. The departure is largest for gradient boosting, but much less pronounced for random forest and CatBoost. Extrapolation is therefore a likely contributor to the flat response of the flexible models, but within the range of the counterfactual exercise its magnitude is too small to be the only factor.

\begin{figure}
    \centering
    \includegraphics[width=0.9\textwidth]{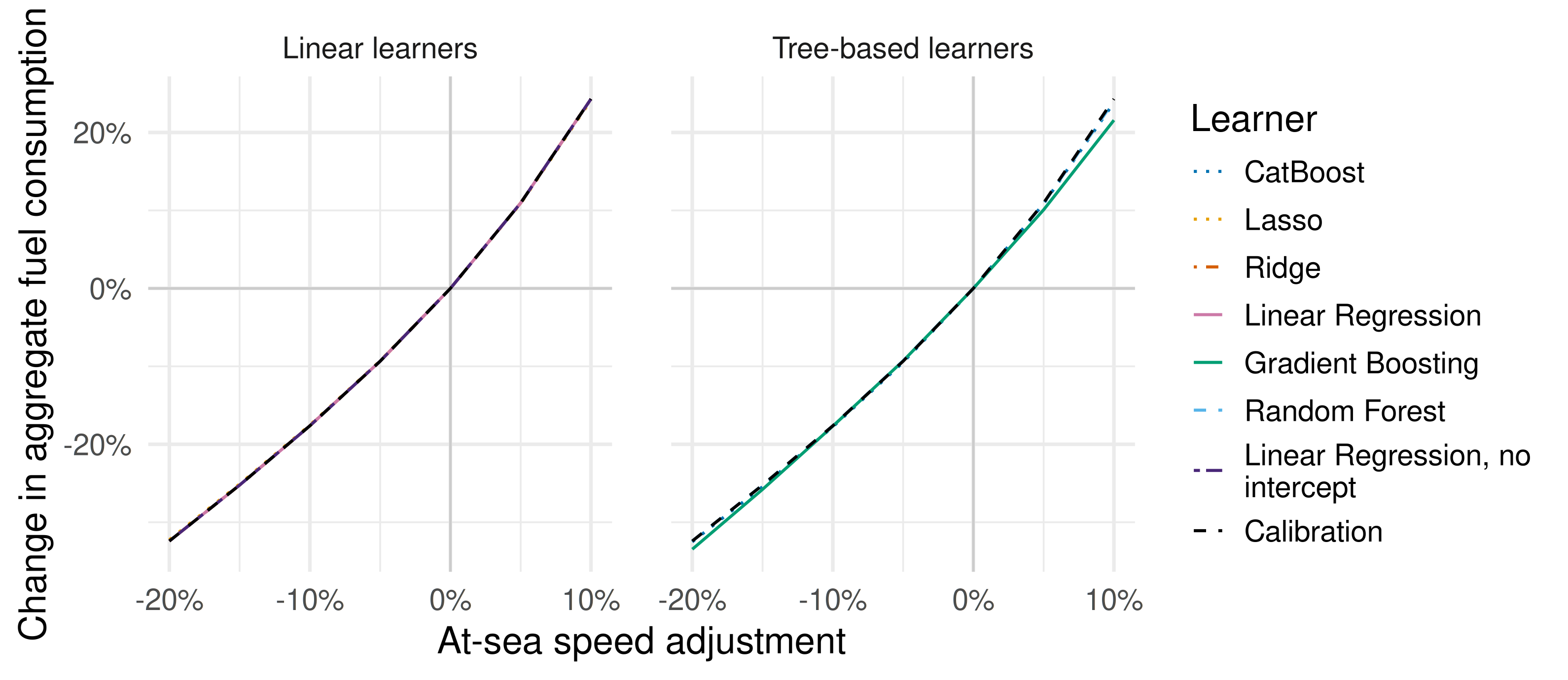}
    \caption{Speed-adjusted counterfactual exercise with calibrated fuel consumption as the fitted target to remove measurement error. `Flexible' class, dry bulk ships. The dashed line is the calibration.}
    \label{fig:sim_flex}
\end{figure}

While the impact we find here is modest, whether the extrapolation behavior matters for a given application is not obvious in advance. For the extrapolation failure discussed above, we should expect accuracy to depend on the number of out-of-support observations and how far out of support they are for the model being used. We therefore propose a diagnostic that indicates how far a test set is from the training set in terms of feature space.

We borrow from the \ac{ML} literature, which uses the Mahalanobis distance for out-of-distribution detection \parencite{lee2018simple,chen2026teaching}. For an observation with feature vector $x_i$, the squared distance from the training sample is
\begin{equation*}
    D_i^2 = (x_i - \hat{\mu})^\top \hat{\Sigma}^{-1} (x_i - \hat{\mu}),
\end{equation*}
where $\hat{\mu}$ is the vector of feature means in the training sample and $\hat{\Sigma}$ their covariance matrix. Weighting by the inverse covariance makes this a Euclidean distance along the principal components of the training data, each scaled by its standard deviation, so that a displacement counts for more in a direction in which the training data vary little and correlated features are not double counted (\appref{app:support}). We compute it on each model's own feature set, on the scale on which the model is fit.

\autoref{fig:support_quartiles_mahal} illustrates the Mahalanobis distance for the quartile splits of \autoref{sec:quartiles} for three feature sets: those of the constrained and unrestricted `Structural' fits and of the `Flexible' class.\footnote{The container ship counterpart, in \appref{app:container_evaluation}, is quite similar.} The same split can be out of support for one feature set and not for another. This is most clearly seen with the 4th size quartile; for the `Flexible' set, the measure exhibits significant mass beyond the 95th percentile of the training distances, perhaps not entirely out of support, but with thin support. In contrast, the test set mass lies well within the training set mass for both structural models.

\begin{figure}
    \centering
    \includegraphics[width=\textwidth]{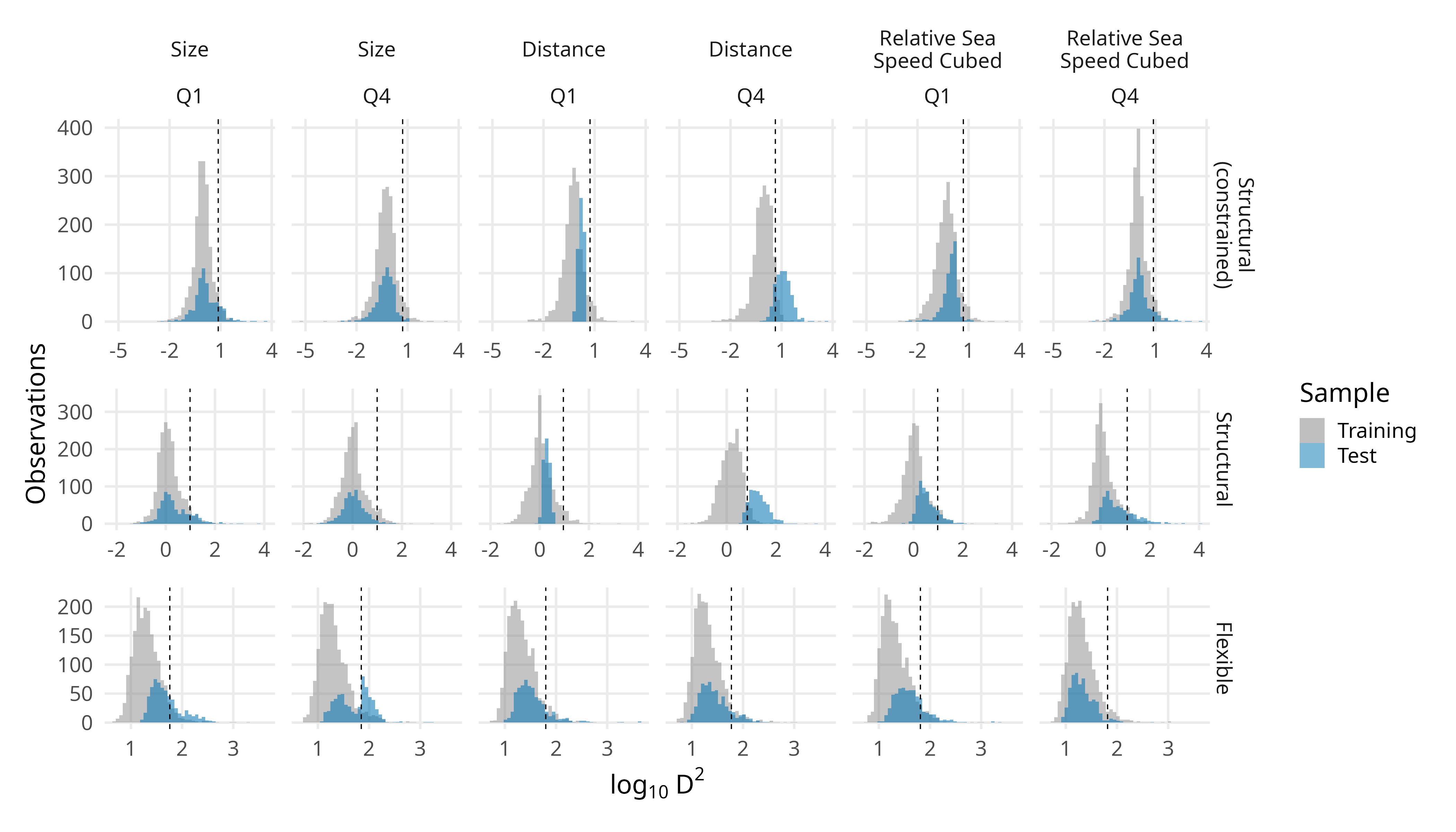}
    \caption{Mahalanobis distance from the training sample for the three quartile splits and three feature sets, dry bulk ships. Gray: training observations, each measured leave-one-out against the rest of the training sample; blue: the held-out quartile. The dashed line is the 95th percentile of the training distances.}
    \label{fig:support_quartiles_mahal}
\end{figure}

A formal comparison requires a summary statistic, and the raw distances are unsuitable because they grow with the number of features. We use two, each relative to the training distances of the same feature set. The first, following \textcite{lee2018simple}, is the share of test observations beyond the 95th percentile of the training distances. The second is our own and, we suggest, better reflects how far out of support a test set lies, since it also counts how far beyond the boundary the observations fall. This standardized mean log excess (hereafter the excess) is
\begin{equation*}
    E = \frac{1}{\sigma}\,\frac{1}{|T|}\sum_{i \in T} \max\{0,\ \log_{10} D_i^2 - \log_{10} q_{95}\},
\end{equation*}
where $T$ is the test set, $q_{95}$ the 95th percentile of the leave-one-out training distances, and $\sigma$ the standard deviation of $\log_{10} D^2$ over those distances. An observation inside the support contributes zero, and the statistic is in units of the training spread, which makes it comparable across feature sets whose spreads differ with the number of features.

\autoref{tab:mahal-quartiles-full} and \autoref{tab:mahal-counterfactual-full} report both statistics for the quartile splits and the speed counterfactual. Both corroborate the simulation finding in that they do not point to out-of-support extrapolation as the main issue; indeed, there is little correspondence between the errors of the tree-based models and the statistics for the `Flexible' set. For reference, in the ship-type transfer of \autoref{sec:domain} the held-out fleet lies almost entirely beyond the support of the `Flexible' set: the share is 0.90 for dry bulk ships and 1.00 for container ships, with an excess of 0.7 and 2.9.

\begin{table}
    \centering
    \begin{threeparttable}
        \caption{Support of the quartile splits by feature set, dry bulk ships.}
        \label{tab:mahal-quartiles-full}
        
\begin{tabular}[t]{lrrrrrr}
\toprule
\multicolumn{1}{c}{ } & \multicolumn{2}{c}{Size} & \multicolumn{2}{c}{Distance} & \multicolumn{2}{c}{Relative Sea Speed Cubed} \\
\cmidrule(l{3pt}r{3pt}){2-3} \cmidrule(l{3pt}r{3pt}){4-5} \cmidrule(l{3pt}r{3pt}){6-7}
 & Q1 & Q4 & Q1 & Q4 & Q1 & Q4\\
\midrule
\addlinespace[0.3em]
\multicolumn{7}{l}{\textbf{Share beyond the training 95th percentile}}\\
\hspace{1em}Structural (constrained) & 0.15 & 0.02 & 0.00 & 0.86 & 0.02 & 0.10\\
\hspace{1em}Structural & 0.16 & 0.02 & 0.00 & 0.91 & 0.15 & 0.25\\
\hspace{1em}Flexible & 0.32 & 0.45 & 0.15 & 0.13 & 0.19 & 0.03\\
\addlinespace[0.3em]
\multicolumn{7}{l}{\textbf{Mean log excess (training SDs)}}\\
\hspace{1em}Structural (constrained) & 0.119 & 0.006 & 0.000 & 0.854 & 0.009 & 0.098\\
\hspace{1em}Structural & 0.155 & 0.013 & 0.000 & 1.074 & 0.085 & 0.286\\
\hspace{1em}Flexible & 0.320 & 0.278 & 0.173 & 0.099 & 0.186 & 0.019\\
\bottomrule
\end{tabular}

        \begin{tablenotes}
            \small
            \item Squared Mahalanobis distance of each held-out ship-year from the remaining three quartiles. The share is the fraction beyond the 95th percentile of the training distances; the standardized mean log excess is the mean $\log_{10}$ distance beyond that percentile, in standard deviation units.
        \end{tablenotes}
    \end{threeparttable}
\end{table}

\begin{table}
    \centering
    \begin{threeparttable}
        \caption{Support of the speed-adjusted counterfactual by feature set, dry bulk ships.}
        \label{tab:mahal-counterfactual-full}
        
\begin{tabular}[t]{lrrrrrr}
\toprule
\multicolumn{1}{c}{ } & \multicolumn{6}{c}{Speed adjustment} \\
\cmidrule(l{3pt}r{3pt}){2-7}
 & -20\% & -15\% & -10\% & -5\% & +5\% & +10\%\\
\midrule
\addlinespace[0.3em]
\multicolumn{7}{l}{\textbf{Share beyond the observed 95th percentile}}\\
\hspace{1em}Structural (constrained) & 0.11 & 0.06 & 0.03 & 0.03 & 0.10 & 0.22\\
\hspace{1em}Structural & 0.23 & 0.16 & 0.09 & 0.05 & 0.09 & 0.21\\
\hspace{1em}Flexible & 0.10 & 0.07 & 0.06 & 0.05 & 0.05 & 0.05\\
\addlinespace[0.3em]
\multicolumn{7}{l}{\textbf{Mean log excess (training SDs)}}\\
\hspace{1em}Structural (constrained) & 0.047 & 0.023 & 0.013 & 0.015 & 0.069 & 0.223\\
\hspace{1em}Structural & 0.148 & 0.092 & 0.050 & 0.028 & 0.066 & 0.213\\
\hspace{1em}Flexible & 0.079 & 0.056 & 0.045 & 0.040 & 0.037 & 0.039\\
\bottomrule
\end{tabular}

        \begin{tablenotes}
            \small
            \item Squared Mahalanobis distance of each speed-adjusted ship-year from the observed sample. The share is the fraction beyond the 95th percentile of the observed distances; the standardized mean log excess is the mean $\log_{10}$ distance beyond that percentile, in standard deviation units.
        \end{tablenotes}
    \end{threeparttable}
\end{table}

\subsection{Measurement Error}
\label{sec:structuralexplanation}

Measurement error affects every model that estimates the speed relation from the data. The tracking data report speed over ground, while drag depends on speed through water, so currents enter as error in speed. The calibration allows for weather only through a fixed correction factor, so wave resistance enters as error in the drag term. The hull state is unobserved, and fouling raises drag over time by an amount that varies across ships. The engine parameters, reference power and base specific fuel consumption, are assigned from standard lookup tables rather than measured, so the assigned value is the same for every ship of a type while the true value varies around it, an error of the Berkson rather than the classical kind. Finally, the annual features are sums over an hourly track that is matched to reported trips and interpolated where transmissions are missing, so trip matching and interpolation add error to distance and to every speed-derived sum.

Errors from trip matching and interpolation cannot be corrected. To gauge them we relax the distance criterion used to subset the training data, admitting ship-years whose observed distance is within $\pm$10\% of the reported distance rather than within $\pm$500 \ac{nm}, which admits more of this error in a larger and noisier sample. Results are very similar to the main ones in the key exercises (\appref{app:distance}).

We attempt to account for currents and waves using reanalysis data. Speed through water is obtained by subtracting reanalysis ocean currents from the observed speed over ground. For wave resistance, we employ a simplified model based on significant wave height and add the resulting power to the calibrated main engine power (\appref{app:weather}).\footnote{Given the uncertainty in the reanalysis data and the simplicity of the wave model, we do not consider this a strict improvement over the baseline approach.} Correcting for currents has almost no effect on any model. Correcting for waves attenuates the speed response of the calibration-anchored models, but it attenuates the theoretical benchmark by a similar amount, since the calibration is recomputed with the same correction, and it leaves the other models unchanged, so the divergence between the models and the benchmark is little changed. For the constrained structural fits, by contrast, the wave correction reduces the response to speed reductions.

The remaining sources, the hull state and the engine parameters, cannot be corrected, so we ask whether errors of these kinds could produce the attenuation we observe. We return to the simulated data of \autoref{sec:mlexplanation}, in which the outcome is the calibrated value, and introduce persistent, multiplicative ship-specific errors into the two calibrated quantities these sources affect: the drag constant (fouling) and the engine term (lookup-table parameters), both of the Berkson kind.\footnote{Classical error, in which the noise sits in the features rather than in the outcome, would work in the opposite direction, attenuating the estimated speed response rather than inflating it.} Each is drawn once per ship, independently of speed, with a standard deviation of 15\% ($\sigma_C$ for the drag constant and $\sigma_E$ for the engine term), the largest magnitude consistent with the fleet's dispersion in speed once the errors are allowed to drive it (\appref{app:errorsim}). We then refit three classes of \autoref{fig:elasticity} and repeat the counterfactual, with both errors present and with the drag error alone (\autoref{fig:sim_independent}). Independent errors of this size leave the estimated speed response of every class at the benchmark, whether both errors are present or the drag error alone.

\begin{figure}
    \centering
    \includegraphics[width=0.75\textwidth]{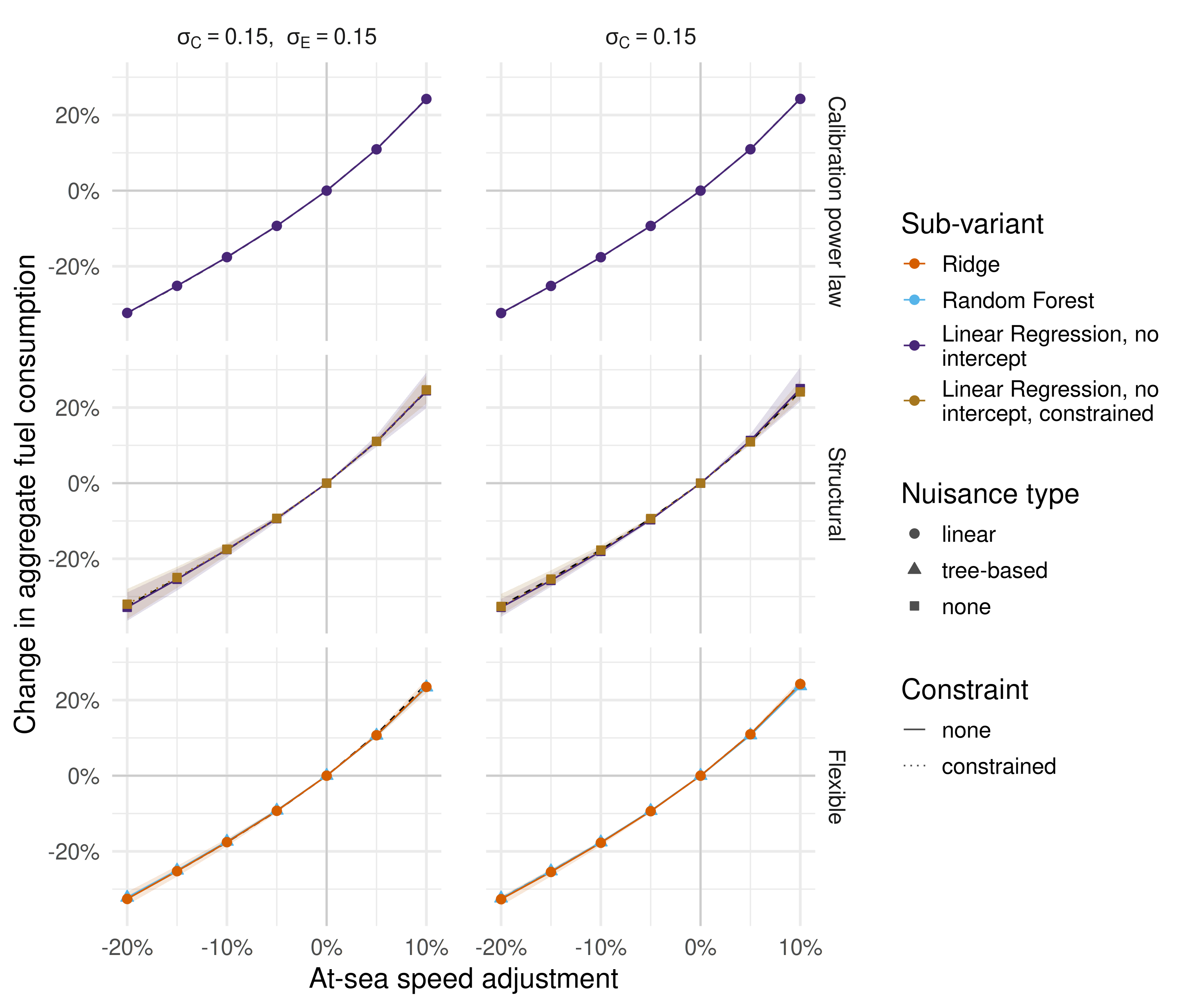}
    \caption{Speed-adjusted counterfactual exercise on simulated data with independent ship-specific errors in the drag constant ($\sigma_C$) and the engine term ($\sigma_E$), dry bulk ships, main engine consumption. Lines are medians across replications, bands their 10th to 90th percentiles. The dashed line is the calibration.}
    \label{fig:sim_independent}
\end{figure}

Errors of the Berkson kind cannot attenuate the estimated response unless they are correlated with speed, which is the last possibility we consider; the wave correction above removed one such source. Here the operator knows the ship's true efficiency and chooses speed accordingly, so we draw the errors conditional on observed speed, which strictly makes this endogeneity rather than measurement error. We choose the response strength from a model of cost-minimizing operator behavior, both with a theoretical speed-price elasticity ($\phi = 1$) and one damped to match an independent empirical estimate ($\phi = 0.243$) (\appref{app:errorsim}). At the empirical strength every fit but ridge regression stays close to the benchmark (\autoref{fig:sim_correlated}). At the theoretical strength the unrestricted structural fit loses most of its response to speed reductions while its response to increases stays near the benchmark; ridge regression predicts higher consumption at lower speeds and no change at higher ones; random forest is moderately attenuated on both sides; and the constrained fit is attenuated for reductions and overshoots for increases. With the drag error alone the attenuation is much weaker.

\begin{figure}
    \centering
    \includegraphics[width=0.95\textwidth]{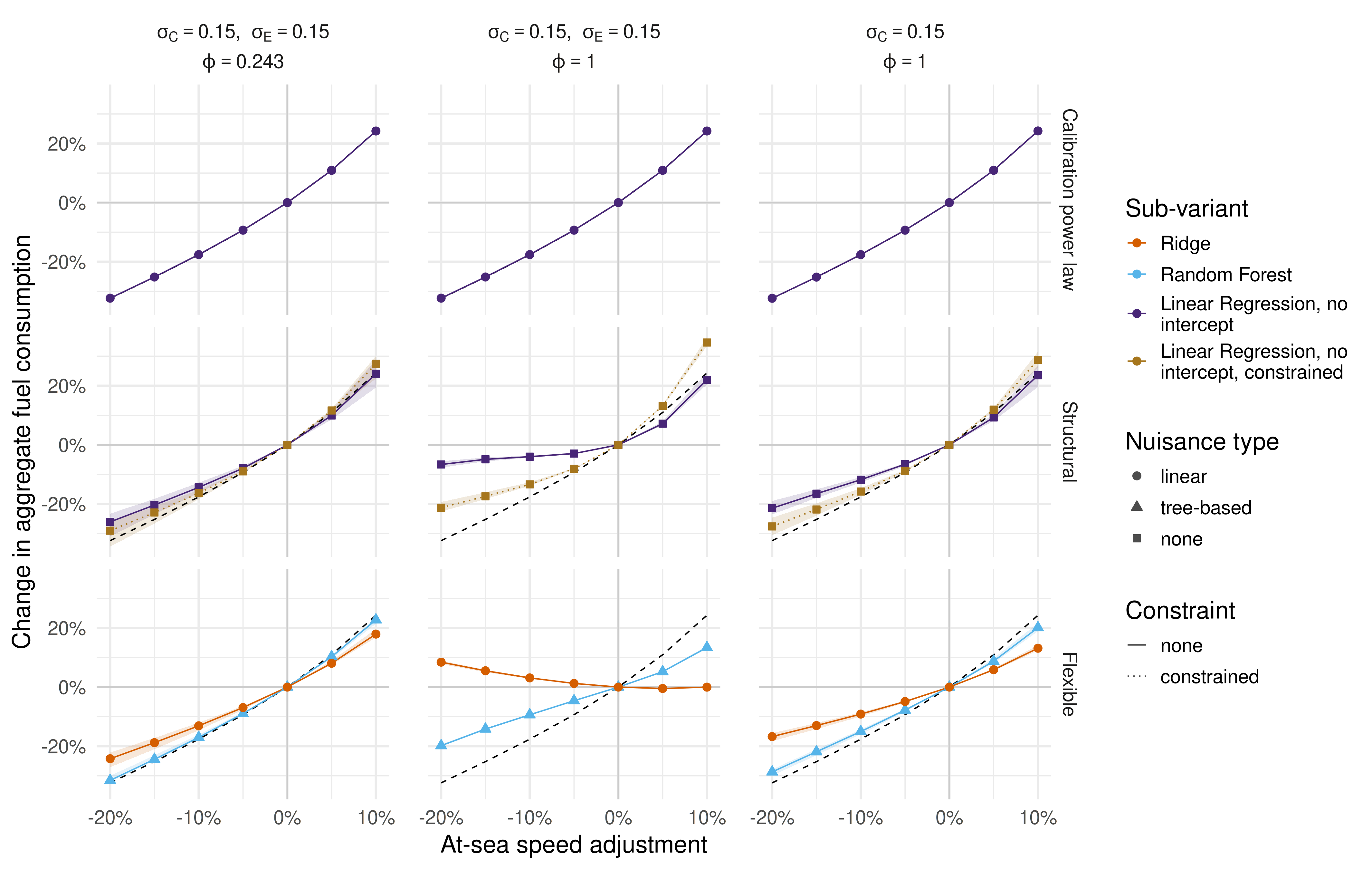}
    \caption{Speed-adjusted counterfactual exercise on simulated data with ship-specific errors correlated with observed speed through an operator response of strength $\phi$, dry bulk ships, main engine consumption. Lines are medians across replications, bands their 10th to 90th percentiles. The dashed line is the calibration.}
    \label{fig:sim_correlated}
\end{figure}

Persistent ship-specific error in the calibrated quantities therefore explains neither the attenuated response of the unrestricted structural fit nor the flat response of the flexible learners. Independent error leaves the response at the benchmark, and speed-correlated error attenuates it materially only at the theoretical response strength, four times the empirical estimate, and then with the opposite pattern to the one we observe: the simulated unrestricted fit is flat for reductions and steep for increases where the estimated one is steeper for reductions and flat for increases, and the simulated constrained fit is attenuated for reductions and overshoots for increases where the estimated one tracks the benchmark for reductions and falls short for increases.

\subsection{Nonlinearity}
\label{sec:nonlinearity}

The third candidate concerns the `Structural' class, whose regressors are the first, second, and third powers of the same relative engine load, summed over hours (\autoref{sec:structural}). Powers of one variable are highly collinear, so the design is ill-conditioned: small perturbations of the data move the coefficients a long way in offsetting directions. The standard remedy is an orthogonal basis \parencite[see][Section~5.5.3]{carroll2006measurement}. We keep the raw powers, since their coefficients are the ones the calibration fixes, and nothing is lost by doing so: a change of basis alters the coefficients but not the fitted values, so within the training support every basis predicts the same consumption.

The constraint of \autoref{sec:structural} restores the response for both fleets, and the reason is instructive. Main engine fuel per hour is the load $w$ times the efficiency curve of \eqref{eqn:sfc} relative to its base value, $w\,s(w)$ with $s(w) = \beta_1 + \beta_2 w + \beta_3 w^2$. Given the freedom, the data choose a falling efficiency curve: in the unrestricted dry bulk fit $s(w)$ falls steeply across the range of loads the ships use, which makes fuel per hour nearly proportional to speed rather than to its cube and fuel per mile nearly independent of speed, the flat response of \autoref{sec:elasticity}. The fit can do this because the slope of the efficiency curve, $s'(w) = \beta_2 + 2\beta_3 w$, and its curvature, $2\beta_3$, are governed by separate parameters. Setting the slope of the efficiency curve to zero at the design load imposes $\beta_2 = -1.56\,\beta_3$, so that
\begin{equation*}
    s'(w) = 2\beta_3\,(w - 0.78),
\end{equation*}
and the efficiency curve can then depart from its value at the design load only by $\beta_3 (w - 0.78)^2$, a departure whose size and sign are set by the single parameter $\beta_3$. The constrained fits choose $\beta_3$ small enough in both fleets that the efficiency curve stays close to flat over the loads the ships use, so fuel per hour stays close to the cube of speed, as in the benchmark. The data resist the tie: the unrestricted coefficients reject the restriction $\beta_2 + 1.56\,\beta_3 = 0$ for both fleets. The constrained fit reproduces the theory because the constraint leaves the data no way to flatten the response, not because the data support the calibrated curve.

Taken together, the candidates account for neither failure completely. Extrapolation contributes to the flat response of the tree-based learners, but modestly, and is not the main issue. Measurement error of the kinds and sizes the data admit explains neither the attenuated response of the unrestricted structural fit nor the flat response of the `Flexible' class, and the data errors we could check, from the inclusion criterion to currents and waves, change the results little. Nor is ill-conditioning the cause: given the freedom, the data prefer a falling efficiency curve to the calibrated one. Whether that preference reflects how consumption is reported or a real departure from the physics, these data cannot say. The constrained fit reproduces the benchmark because the constraint removes the freedom through which the flat response enters, which recommends it for counterfactual use without settling whether the calibrated curve is right.

\section{Illustrative Prediction Exercises}
\label{sec:MLprediction}

We close with applications that highlight the utility of accurate prediction models. First, since the official \ac{IMO} inventories are annual and global \parencite{faber2020fourth}, we disaggregate emissions along those two dimensions, by month and by origin--destination country pair, where the value of a model lies in resolving what the official figures average over. We then illustrate a welfare calculation for a change in ship speeds, which is a key margin to quantify accurately for assessing emissions policy. Throughout, the figures present a select set of illustrative models.\footnote{Within each class we show the sub-variant and learner with the lowest mean absolute \ac{RTE} across the held-out quartiles of \autoref{sec:quartiles}. For the constrained structural fit we use the variant without an intercept, whose speed response is closest to the calibration (\autoref{sec:comparison}).}

\subsection{Disaggregation}
\label{sec:disaggregation}

For both disaggregation exercises we apply the fitted models to all global \ac{AIS} movements rather than to \ac{EU} trips alone. For the temporal exercise, the data are aggregated by ship and month; for the spatial one, by ship, year, and origin--destination country pair. The models are those trained on the full \ac{EU}-trip sample of each fleet, so the prediction data extend the training data in both directions: monthly aggregation shortens distances relative to the annual training observations, and the global movements include trips the training sample never saw. Predicted fuel consumption is converted to CO$_2$ at a fixed factor of 3.114 tonnes per tonne of fuel, given the predominance of conventional fuel oil. Details of the aggregation, the assignment of trips to country pairs, and the treatment of infeasible observations are in \appref{app:disaggregation}.

In the monthly series (\autoref{fig:temporal_disagg}) the most striking feature is the difference in level, and it varies by ship type: for dry bulk ships the constrained structural fit lies below the other four models, while for container ships the calibration lies above the other four. Month-to-month fluctuations are otherwise similar across models. All highlight interesting patterns, including the dip at Chinese New Year each year and the extended slump and subsequent recovery for container ships (much less pronounced for bulk ships) coincident with COVID-19. Two exceptions stand out: in the spring of 2019 the calibration and the constrained structural fit dip while the pure \ac{ML} models do not, and from the start of 2021 onward the `Efficiency' model falls below the models it otherwise tracks.\footnote{Independent data confirm that ship speeds were elevated through 2021.} Both are consistent with the speed responses of the evaluation exercises: variation in speed rather than in distance traveled moves the calibration and the constrained structural fit, is damped in the unrestricted structural fit and the `Flexible' model, and leaves the `Efficiency' model unmoved. Our support diagnostic of \autoref{sec:mlexplanation} tells in advance whether a model is being asked to extrapolate, and here it suggests caution with respect to the two pure \ac{ML} models: a sizeable share of the ship-months lies beyond their training support, while the structural feature sets remain within it (\autoref{fig:support_monthly}).%

\begin{figure}
    \centering
    \includegraphics[width=0.9\textwidth]{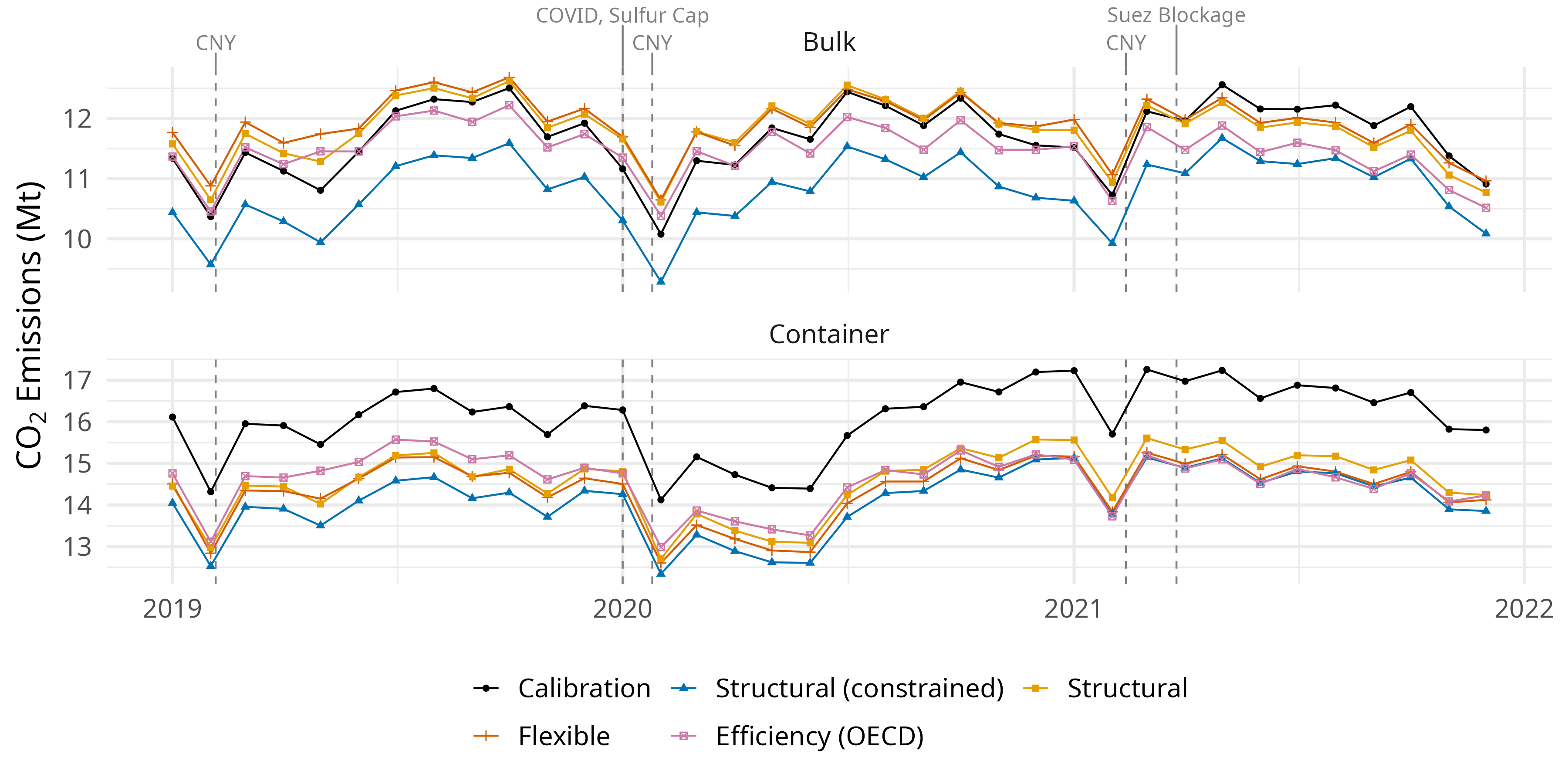}
    \caption{Predicted global monthly CO$_2$ emissions by model, with events marked (CNY = Chinese New Year, COVID = approximate onset of COVID-19, Sulfur Cap = implementation of the low-sulfur fuel regulation).}
    \label{fig:temporal_disagg}
\end{figure}

\begin{figure}
    \centering
    \includegraphics[width=0.9\textwidth]{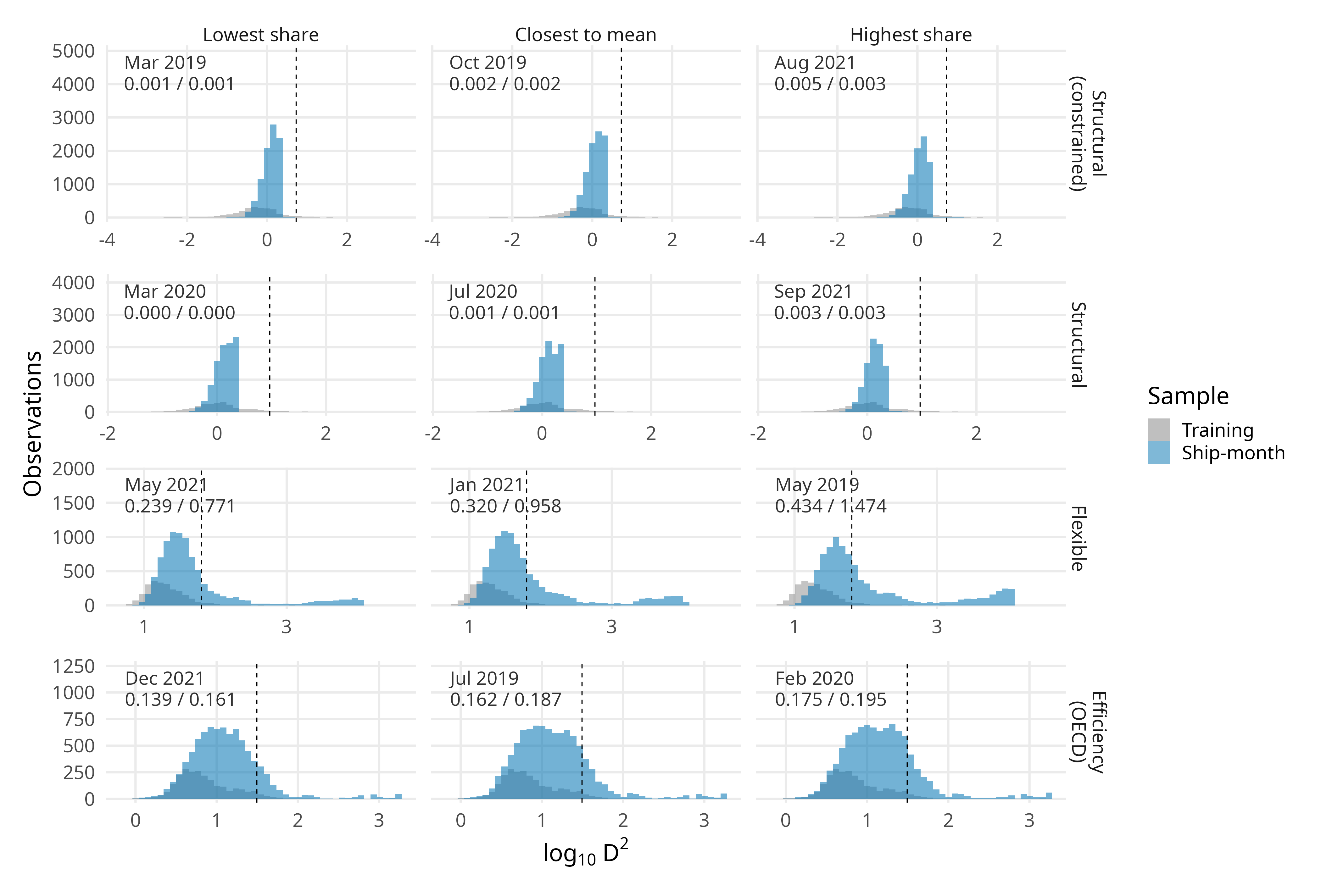}
    \caption{Mahalanobis distance of the ship-month prediction data from the annual training sample by feature set, dry bulk ships, for the months with the lowest, mean, and highest share beyond the training 95th percentile; the month, the share, and the standardized mean log excess are annotated in each panel. Gray: training observations; blue: ship-months. The dashed line is the 95th percentile of the training distances.}
    \label{fig:support_monthly}
\end{figure}

The route disaggregation (\autoref{fig:spatial_disagg}) is a demonstration of the patterns that disaggregation can reveal rather than a test of the models. Emissions from between-country trips are concentrated on a few country pairs: routes to and from China dominate both fleets, with Australia and Brazil the main partners for dry bulk ships and the United States for container ships. The models agree fairly well on both the ranking and the magnitudes, with the calibration again highest for container ships. The same caution for the pure \ac{ML} models applies here, as the route aggregates also lie partly beyond the support of their feature sets (\autoref{tab:mahal-route}).

\begin{figure}
    \centering
    \includegraphics[width=0.95\textwidth]{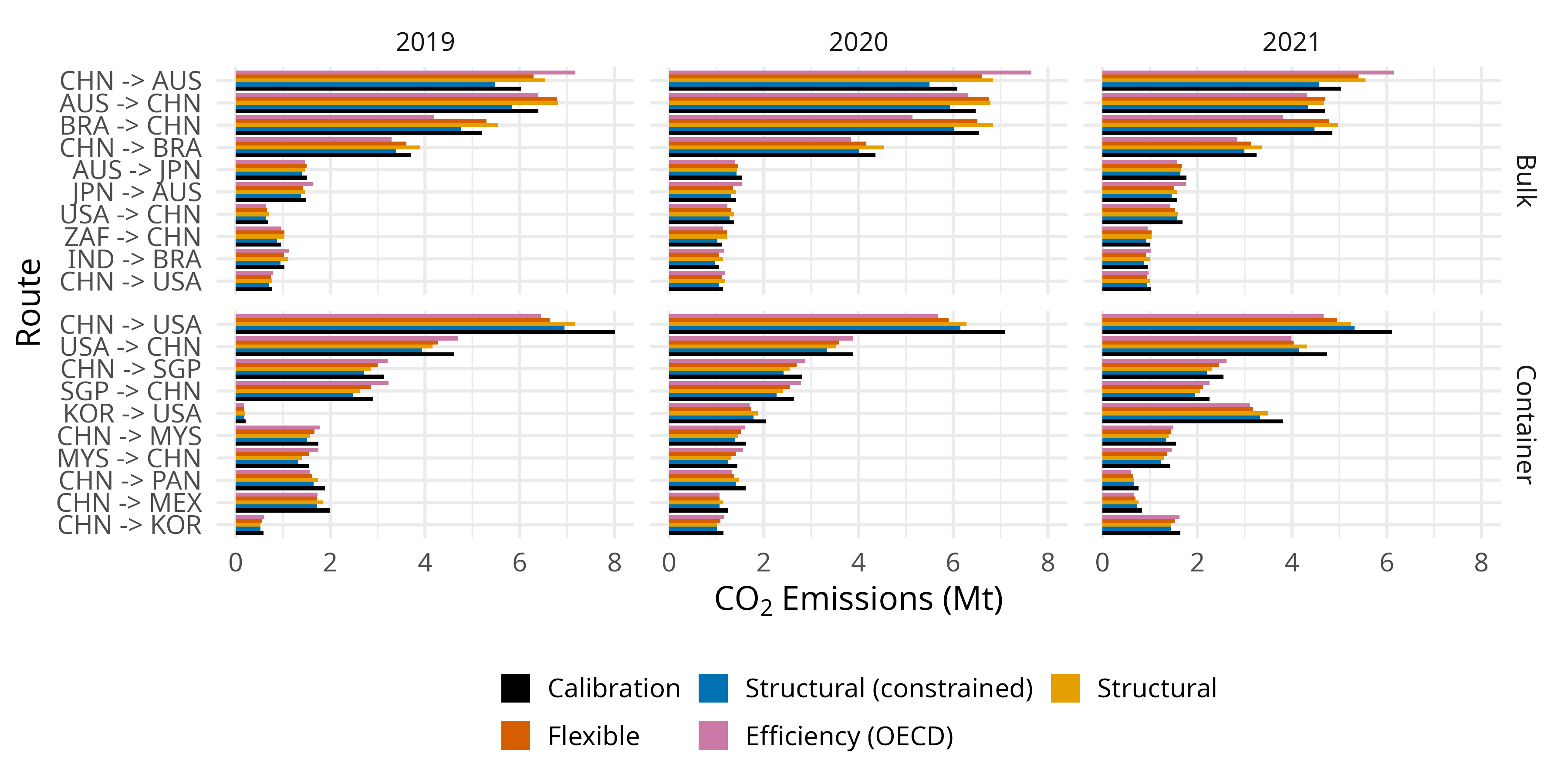}
    \caption{Predicted annual CO$_2$ emissions for the top between-country origin--destination pairs by model.}
    \label{fig:spatial_disagg}
\end{figure}

\begin{table}
    \centering
    \begin{adjustbox}{max width=\textwidth}
    \begin{threeparttable}
        \caption{Support of the route prediction data by feature set.}
        \label{tab:mahal-route}
        
\begin{tabular}[t]{lrrrrrrrr}
\toprule
\multicolumn{1}{c}{ } & \multicolumn{4}{c}{Share beyond the training 95th percentile} & \multicolumn{4}{c}{Mean log excess (training SDs)} \\
\cmidrule(l{3pt}r{3pt}){2-5} \cmidrule(l{3pt}r{3pt}){6-9}
 & Mean & SD & Min & Max & Mean & SD & Min & Max\\
\midrule
\addlinespace[0.3em]
\multicolumn{9}{l}{\textbf{Bulk}}\\
\hspace{1em}Structural (constrained) & 0.008 & 0.008 & 0.000 & 0.028 & 0.003 & 0.003 & 0.000 & 0.010\\
\hspace{1em}Structural & 0.016 & 0.018 & 0.001 & 0.056 & 0.005 & 0.007 & 0.000 & 0.021\\
\hspace{1em}Flexible & 0.169 & 0.091 & 0.065 & 0.371 & 0.243 & 0.151 & 0.111 & 0.615\\
\hspace{1em}Efficiency (OECD) & 0.236 & 0.161 & 0.025 & 0.547 & 0.307 & 0.343 & 0.010 & 1.120\\
\addlinespace[0.3em]
\multicolumn{9}{l}{\textbf{Container}}\\
\hspace{1em}Structural (constrained) & 0.004 & 0.007 & 0.000 & 0.023 & 0.003 & 0.007 & 0.000 & 0.023\\
\hspace{1em}Structural & 0.005 & 0.008 & 0.000 & 0.026 & 0.005 & 0.011 & 0.000 & 0.036\\
\hspace{1em}Flexible & 0.132 & 0.085 & 0.046 & 0.301 & 0.122 & 0.077 & 0.042 & 0.286\\
\hspace{1em}Efficiency (OECD) & 0.027 & 0.020 & 0.000 & 0.066 & 0.011 & 0.008 & 0.000 & 0.028\\
\bottomrule
\end{tabular}

        \begin{tablenotes}
            \small
            \item Squared Mahalanobis distance of each ship-year on a route from the training sample of its fleet; statistics over the routes of \autoref{fig:spatial_disagg}. Statistics as in \autoref{tab:mahal-quartiles-full}.
        \end{tablenotes}
    \end{threeparttable}
    \end{adjustbox}
\end{table}

\subsection{Welfare Assessment of Speed Policy}
\label{sec:welfare}

The speed counterfactual of \autoref{sec:elasticity} supports a simple welfare assessment of regulation that affects ship speeds. Regulation rarely mandates speed directly. It acts through incentives, and fuel, one of the largest costs of operating a ship, is the main one: a fuel levy has been discussed repeatedly at the \ac{IMO} without being adopted, while the efficiency requirements adopted for existing ships leave operators few technical options, so many comply by slowing down. We take the same counterfactual as before, a uniform proportional change in at-sea speed with distance held fixed, and price its three consequences.\footnote{The exercise holds distance, fleet, and trade fixed and so abstracts from general equilibrium effects; see \textcite{peters2024beached} for a dynamic, general equilibrium treatment.} Fuel is valued at 600 USD per tonne and CO$_2$ at 190 USD per tonne.\footnote{Fuel expenditure is in part a transfer between operators and fuel suppliers rather than a resource cost. We include it at its market price as the resource cost of the fuel burned. Excluding it would lower every model's total by the value of its predicted fuel saving, so the contrast between models below would remain while the levels would not.} Time at sea is valued at 0.6\% of cargo value per day, the lower bound of \textcite{hummels2013time}, which is an upper bound for dry bulk commodities since their estimates cover goods that travel by air as well as by sea, with cargo valued at 100 USD per tonne from coal and iron ore prices. The time cost is applied to every capacity tonne-hour at sea, laden or not, as an opportunity cost.

\autoref{fig:speed_welfare} shows the three components and their total for the same five models at speed changes of 5 and 10\% in either direction. The time cost is identical across models by construction, so every difference between them is a difference in predicted fuel consumption. The models that carry the theoretical speed response, the calibration and the constrained structural fit, predict that slowing down pays: the fuel and carbon savings exceed the time cost at both reductions, and speeding up costs. The models with a flat speed response predict the reverse, a large loss from slowing and a gain from speeding, because their predicted fuel savings are too small to offset the time cost. The sign of the assessment therefore turns on a model's speed response and not on its level accuracy: the unrestricted structural fit, among the least biased out of sample, reaches the same verdict as the pure \ac{ML} classes. Given the assumptions involved, we draw no conclusion about the policy itself; the exercise shows how far the choice of model propagates into a policy calculation.

\begin{figure}
    \centering
    \includegraphics[width=\textwidth]{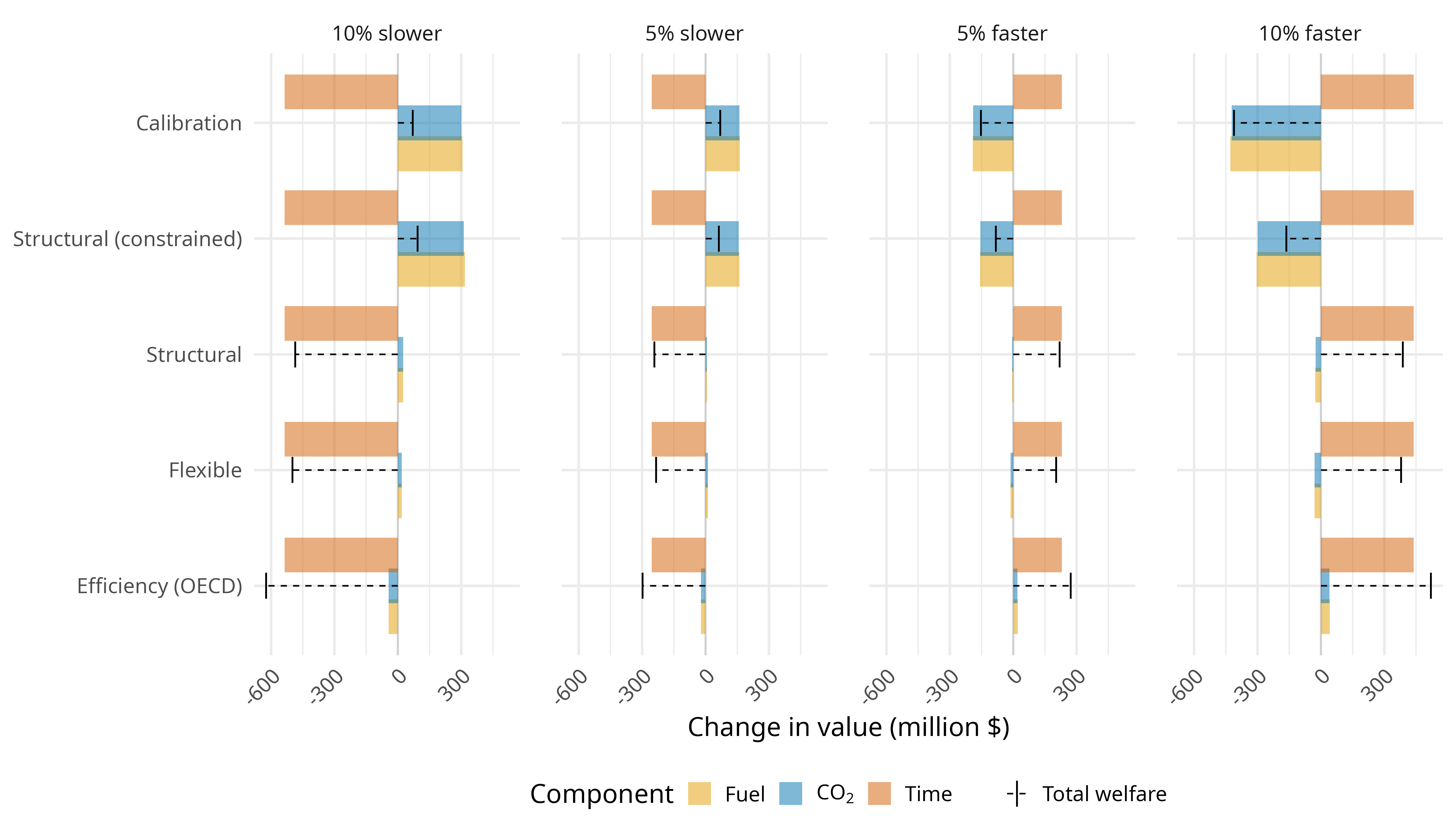}
    \caption{Welfare change under counterfactual at-sea speed adjustments for selected models, dry bulk ships, million USD. Bars are the fuel, CO$_2$, and time components; the black marker is the total.}
    \label{fig:speed_welfare}
\end{figure}

\section{Conclusion}
\label{sec:MLconclusion}

This paper develops and evaluates combinations of \ac{ML} and theory-based structure to predict ship fuel consumption and its response to counterfactual changes in sailing speed. Both quantities are essential for shipping emissions policy. We assess aggregate predictions against reported consumption and compare estimated speed responses with an engineering benchmark. Several of our models outperform those used for official estimates on at least one of these criteria, while performing comparably on the other.

In the 2021 test sample, every estimated model predicts aggregate fuel consumption within a few percent of reported totals, improving upon the calibrated engineering calculation. Yet the models imply markedly different speed responses. For dry bulk ships, pure \ac{ML} and the unrestricted structural regression yield responses far below the benchmark. Calibration-anchored models generally track the benchmark response closely, though they retain some aggregate overprediction. In the main specification, two simple models without \ac{ML} combine low aggregate prediction bias with responses close to the benchmark: the rescaled engineering calculation and the structural regression constrained at the engine's design load, both estimated without an intercept. This restriction sets the slope of the efficiency curve to zero at the design load rather than leaving its slope and curvature unconstrained.

The illustrative cost-benefit exercise demonstrates why the distinction between predictive accuracy and speed responsiveness matters. Among the five models considered, the calibrated engineering calculation and the constrained structural regression imply positive net benefits from the imposed speed reductions, whereas the selected pure \ac{ML} models and the unrestricted structural regression imply negative net benefits. Models with similarly small aggregate prediction errors can thus yield opposite cost-benefit conclusions under identical valuation assumptions.

We draw three methodological lessons. First, the response used in a counterfactual should be evaluated separately from statistical fit. Second, the separable hybrids studied here offer a transparent way to prevent a flexible correction from altering a specified structural property. Excluding speed-related variables from the correction leaves it unchanged when speed changes and other inputs are held fixed. Consequently, the multiplicative hybrid preserves the structural component's ship-level speed elasticity, while the additive hybrid preserves its speed derivative in levels. This property neither guarantees small bias in levels nor validates the preserved response, and transferring the correction to a counterfactual requires an additional invariance assumption. Third, distance in a model's feature space can flag departures from the training distribution. Our diagnostic serves as a warning indicator rather than a test of counterfactual validity: a small distance neither establishes that an observation lies within the training support nor guarantees agreement with the benchmark.

For emissions accounting, the same issues arise when models are applied beyond annual totals. Our replication of the \acs{OECD} fuel-efficiency model matches reported totals closely but exhibits essentially no speed response for dry bulk ships. Moving from ship-years to ship-months or to global route aggregates can likewise place prediction inputs far outside the training distribution. These applications call for assessing both the relevant input responses and the conditions under which predictions are transferred.

Several limitations qualify these findings. Estimation relies on three years of reported consumption for two vessel types on voyages to, from, and between European Economic Area ports, retaining ship-years whose tracked and reported distances agree within the specified tolerance. The engineering response is a maintained benchmark rather than an observed counterfactual. In addition, the unrestricted estimates reject the design-load restriction imposed by the constrained regression. Accounting for wave resistance reduces the constrained regression's response to speed reductions to roughly one-third of its uncorrected magnitude, while leaving its response to speed increases largely unchanged. Furthermore, supplementary exercises addressing extrapolation, measurement error, and structural nonlinearity do not fully explain departures from the benchmark. Finally, the cost-benefit exercise evaluates speed changes while holding distance, fleet composition, and trade fixed, capturing neither operators' behavioral responses to specific policies nor general equilibrium adjustments.

The central implication is that predictive performance and the credibility of a counterfactual response require separate assessment.

\newpage
\begin{singlespace}
  \raggedright
  \printbibliography[heading=bibintoc]

@misc{mepc2021guidelines,
  author = {{The Marine Environment Protection Committee}},
  title  = {Guidelines for Determining Minimum Propulsion Power to Maintain the Manoeuvrability of Ships in Adverse Conditions},
  shortauthor = {MEPC},
  institution = {International Maritime Organization},
  year        = {2021},
  month       = {7},
  url         = {https://wwwcdn.imo.org/localresources/en/OurWork/Environment/Documents/Air%20pollution/MEPC.1-Circ.850-Rev.3.pdf},
  urldate     = {2026-01-05}
}

@misc{cmems2024waverys,
  author       = {{E.U. Copernicus Marine Service Information}},
  title        = {Global Ocean Waves Reanalysis},
  organization = {Copernicus Marine Service},
  year         = {2024},
  date         = {2024-11-26},
  version      = {202411},
  doi          = {10.48670/moi-00022},
  url          = {https://doi.org/10.48670/moi-00022},
  urldate      = {2025-12-08},
  note         = {Product ID: GLOBAL\_MULTIYEAR\_WAV\_001\_032; DATASET\_ID: cmems\_mod\_glo\_wav\_my\_0.2deg\_PT3H-i}
}

@misc{cmems2023glorys,
  author       = {{E.U. Copernicus Marine Service Information}},
  title        = {Global Ocean Physics Reanalysis},
  organization = {Copernicus Marine Service},
  year         = {2023},
  month        = {11},
  version      = {202311},
  doi          = {10.48670/moi-00021},
  url          = {https://doi.org/10.48670/moi-00021},
  urldate      = {2025-12-01},
  note         = {Product ID: GLOBAL\_MULTIYEAR\_PHY\_001\_030 (GLORYS12V1); DATASET\_ID: cmems\_mod\_glo\_phy\_my\_0.083deg\_P1D-m}
}

@article{hummels2013time,
  title={Time as a trade barrier},
  author={Hummels, David L and Schaur, Georg},
  journal={American Economic Review},
  volume={103},
  number={7},
  pages={2935--2959},
  year={2013},
  publisher={American Economic Association}
}

@article{chernozhukov2018double,
    author = {Chernozhukov, Victor and Chetverikov, Denis and Demirer, Mert and Duflo, Esther and Hansen, Christian and Newey, Whitney and Robins, James},
    title = {Double/debiased machine learning for treatment and structural parameters},
    journal = {The Econometrics Journal},
    volume = {21},
    number = {1},
    pages = {C1-C68},
    year = {2018},
    month = {01},
    issn = {1368-4221},
    doi = {10.1111/ectj.12097},
    url = {https://doi.org/10.1111/ectj.12097},
    eprint = {https://academic.oup.com/ectj/article-pdf/21/1/C1/27684918/ectj00c1.pdf},
}

@techreport{man2018basic,
  title={Basic principles of ship propulsion},
  author={{MAN Energy Solutions}},
  shortauthor={{MAN}},
  year={2018}
}

@techreport{peters2024beached,
  title={Beached Assets? Capital Turnover and Emissions in Shipping},
  author={Peters, Allen},
  year={2024},
  type={Working Paper}
}

@misc{andrews2025transfer,
  title={The transfer performance of economic models},
  author={Andrews, Isaiah and Fudenberg, Drew and Lei, Lihua and Liang, Annie and Wu, Chaofeng},
  year={2025},
  eprint={2202.04796},
  eprinttype={arxiv},
  eprintclass={econ.TH},
  version={5},
  doi={10.48550/arXiv.2202.04796},
  note={arXiv working paper, version 5}
}

@article{eu2015regulation,
  title={{Regulation (EU) 2015/757 of the European Parliament and of the Council}},
  author={{European Union}},
  shortauthor={EU},
  url={https://eur-lex.europa.eu/legal-content/EN/TXT/?uri=CELEX:32015R0757},
  year={2015},
  month = {4},
  day = {29},
  journal={Official Journal of the European Union}
}

@article{clarke2023co2,
  title={{CO2} emissions from global shipping: A new experimental database},
  author={Clarke, Daniel and Chan, Philip and Dequeljoe, Matthew and Kim, Yuri and Barahona, Sarah},
  year={2023},
  publisher={OECD}
}

@techreport{faber2020fourth,
	author = {Faber, J and Hanayama, S and Zhang, S and Pereda, P and Comer, B and Hauerhof, E and Yuan, H},
	institution = {International Maritime Organization},
	title = {{Fourth IMO Greenhouse Gas Study 2020}},
	url = {https://www.imo.org/en/OurWork/Environment/Pages/Fourth-IMO-Greenhouse-Gas-Study-2020.aspx},
	year = {2020}}

@article{mullainathan2017machine,
	author = {Mullainathan, Sendhil and Spiess, Jann},
	title = {Machine Learning: An Applied Econometric Approach},
	journal = {Journal of Economic Perspectives},
	volume = {31},
	number = {2},
	pages = {87--106},
	year = {2017}}

@article{kleinberg2015prediction,
	author = {Kleinberg, Jon and Ludwig, Jens and Mullainathan, Sendhil and Obermeyer, Ziad},
	title = {Prediction Policy Problems},
	journal = {American Economic Review},
	volume = {105},
	number = {5},
	pages = {491--495},
	year = {2015}}

@article{adland2020optimal,
	author = {Adland, Roar and Cariou, Pierre and Wolff, Fran{\c{c}}ois-Charles},
	title = {Optimal Ship Speed and the Cubic Law Revisited: {Empirical} Evidence from an Oil Tanker Fleet},
	journal = {Transportation Research Part E: Logistics and Transportation Review},
	volume = {140},
	pages = {101972},
	year = {2020}}

@article{ollila2024effect,
	author = {Ollila, Saana and Merkel, Axel and B{\"o}rjesson, Maria},
	title = {Effect of Fuel Prices on Sailing Speeds in Short-Sea Shipping},
	journal = {Applied Economics},
	volume = {57},
	number = {33},
	pages = {4938--4954},
	doi = {10.1080/00036846.2024.2364106},
	year = {2024}}

@article{almeida2023machine,
	author = {Almeida, Caio and Fan, Jianqing and Freire, Gustavo and Tang, Francesca},
	title = {Can a Machine Correct Option Pricing Models?},
	journal = {Journal of Business \& Economic Statistics},
	volume = {41},
	number = {3},
	pages = {995--1009},
	doi = {10.1080/07350015.2022.2099871},
	year = {2023}}

@article{prest2023rcts,
	author = {Prest, Brian C. and Wichman, Casey J. and Palmer, Karen},
	title = {{RCTs} against the Machine: Can Machine Learning Prediction Methods Recover Experimental Treatment Effects?},
	journal = {Journal of the Association of Environmental and Resource Economists},
	volume = {10},
	number = {5},
	pages = {1231--1264},
	doi = {10.1086/724518},
	year = {2023}}

@article{burlig2020machine,
	author = {Burlig, Fiona and Knittel, Christopher and Rapson, David and Reguant, Mar and Wolfram, Catherine},
	title = {Machine Learning from Schools about Energy Efficiency},
	journal = {Journal of the Association of Environmental and Resource Economists},
	volume = {7},
	number = {6},
	pages = {1181--1217},
	doi = {10.1086/710606},
	year = {2020}}

@article{abrell2022effective,
	author = {Abrell, Jan and Kosch, Mirjam and Rausch, Sebastian},
	title = {How Effective Is Carbon Pricing? {A} Machine Learning Approach to Policy Evaluation},
	journal = {Journal of Environmental Economics and Management},
	volume = {112},
	pages = {102589},
	doi = {10.1016/j.jeem.2021.102589},
	year = {2022}}

@article{kalouptsidi2014time,
	author = {Kalouptsidi, Myrto},
	title = {Time to Build and Fluctuations in Bulk Shipping},
	journal = {American Economic Review},
	volume = {104},
	number = {2},
	pages = {564--608},
	doi = {10.1257/aer.104.2.564},
	year = {2014}}

@article{klotz2022local,
	author = {Klotz, Richard and Berazneva, Julia},
	title = {Local Standards, Behavioral Adjustments, and Welfare: Evaluating {California}'s Ocean-Going Vessel Fuel Rule},
	journal = {Journal of the Association of Environmental and Resource Economists},
	volume = {9},
	number = {3},
	pages = {383--424},
	doi = {10.1086/717585},
	year = {2022}}

@article{borenstein2019expecting,
	author = {Borenstein, Severin and Bushnell, James and Wolak, Frank A. and Zaragoza-Watkins, Matthew},
	title = {Expecting the Unexpected: {Emissions} Uncertainty and Environmental Market Design},
	journal = {American Economic Review},
	volume = {109},
	number = {11},
	pages = {3953--3977},
	doi = {10.1257/aer.20161218},
	year = {2019}}

@article{fowlie2018does,
	author = {Fowlie, Meredith and Greenstone, Michael and Wolfram, Catherine},
	title = {Do Energy Efficiency Investments Deliver? {Evidence} from the {Weatherization Assistance Program}},
	journal = {The Quarterly Journal of Economics},
	volume = {133},
	number = {3},
	pages = {1597--1644},
	doi = {10.1093/qje/qjy005},
	year = {2018}}

@article{christensen2023performance,
	author = {Christensen, Peter and Francisco, Paul and Myers, Erica and Souza, Mateus},
	title = {Decomposing the Wedge between Projected and Realized Returns in Energy Efficiency Programs},
	journal = {Review of Economics and Statistics},
	volume = {105},
	number = {4},
	pages = {798--817},
	doi = {10.1162/rest_a_01087},
	year = {2023}}

@article{nevo2010taking,
	author = {Nevo, Aviv and Whinston, Michael D.},
	title = {Taking the Dogma out of Econometrics: Structural Modeling and Credible Inference},
	journal = {Journal of Economic Perspectives},
	volume = {24},
	number = {2},
	pages = {69--82},
	doi = {10.1257/jep.24.2.69},
	year = {2010}}

@article{robinson1988root,
	author = {Robinson, Peter M.},
	title = {Root-{N}-Consistent Semiparametric Regression},
	journal = {Econometrica},
	year = {1988},
	volume = {56},
	number = {4},
	pages = {931--954}}

@techreport{chen2026teaching,
	author = {Chen, Hui and Cheng, Yuhan and Liu, Yanchu and Tang, Ke},
	title = {Teaching Economics to the Machines},
	institution = {National Bureau of Economic Research},
	year = {2026},
	number = {w34713}}

@inproceedings{prokhorenkova2018catboost,
	author = {Prokhorenkova, Liudmila and Gusev, Gleb and Vorobev, Aleksandr and Dorogush, Anna Veronika and Gulin, Andrey},
	title = {{CatBoost}: Unbiased Boosting with Categorical Features},
	booktitle = {Advances in Neural Information Processing Systems},
	year = {2018},
	volume = {31}}

@misc{eu2023directive959,
  author = {{European Union}},
  title = {Directive ({EU}) 2023/959 amending {Directive} 2003/87/{EC} establishing a system for greenhouse gas emission allowance trading within the {Union}},
  year = {2023},
  url = {http://data.europa.eu/eli/dir/2023/959/oj},
  note = {Official Journal of the European Union, L~130, 16 May 2023, pp.~134--202}
}

@misc{imo2023strategy,
  author = {{IMO}},
  title = {{2023 IMO Strategy on Reduction of GHG Emissions from Ships}},
  year = {2023},
  note = {Resolution MEPC.377(80), adopted 7 July 2023},
  url = {https://wwwcdn.imo.org/localresources/en/KnowledgeCentre/IndexofIMOResolutions/MEPCDocuments/MEPC.377(80).pdf}
}

@misc{unctad2017review,
  author = {{UNCTAD}},
  title = {Review of Maritime Transport 2017},
  year = {2017},
  note = {United Nations, Geneva}
}

@article{duan1983smearing,
  title = {Smearing estimate: a nonparametric retransformation method},
  author = {Duan, Naihua},
  journal = {Journal of the American Statistical Association},
  volume = {78},
  number = {383},
  pages = {605--610},
  year = {1983}
}

@article{silva2006log,
  title = {The log of gravity},
  author = {{Santos Silva}, J. M. C. and Tenreyro, Silvana},
  journal = {The Review of Economics and Statistics},
  volume = {88},
  number = {4},
  pages = {641--658},
  year = {2006}
}

@article{gneiting2011making,
  title = {Making and evaluating point forecasts},
  author = {Gneiting, Tilmann},
  journal = {Journal of the American Statistical Association},
  volume = {106},
  number = {494},
  pages = {746--762},
  year = {2011},
  doi = {10.1198/jasa.2011.r10138}
}

@inproceedings{lee2018simple,
  author={Lee, Kimin and Lee, Kibok and Lee, Honglak and Shin, Jinwoo},
  title={A simple unified framework for detecting out-of-distribution samples and adversarial attacks},
  booktitle={Advances in Neural Information Processing Systems},
  editor={Bengio, Samy and Wallach, Hanna and Larochelle, Hugo and Grauman, Kristen and Cesa-Bianchi, Nicol{\`o} and Garnett, Roman},
  volume={31},
  year={2018},
  publisher={Curran Associates},
  url={https://proceedings.neurips.cc/paper_files/paper/2018/hash/abdeb6f575ac5c6676b747bca8d09cc2-Abstract.html}
}

@book{hastie2009elements,
  author={Hastie, Trevor and Tibshirani, Robert and Friedman, Jerome},
  title={The Elements of Statistical Learning: Data Mining, Inference, and Prediction},
  edition={2},
  publisher={Springer},
  address={New York},
  year={2009},
  series={Springer Series in Statistics},
  doi={10.1007/978-0-387-84858-7}
}

@inproceedings{xu2021how,
  author={Xu, Keyulu and Zhang, Mozhi and Li, Jingling and Du, Simon S. and Kawarabayashi, Ken-ichi and Jegelka, Stefanie},
  title={How Neural Networks Extrapolate: From Feedforward to Graph Neural Networks},
  booktitle={International Conference on Learning Representations},
  year={2021},
  url={https://openreview.net/forum?id=UH-cmocLJC}
}

@book{carroll2006measurement,
  author={Carroll, Raymond J. and Ruppert, David and Stefanski, Leonard A. and Crainiceanu, Ciprian M.},
  title={Measurement Error in Nonlinear Models: A Modern Perspective},
  edition={2},
  publisher={Chapman and Hall/CRC},
  address={Boca Raton, FL},
  year={2006},
  series={Monographs on Statistics and Applied Probability},
  number={105},
  doi={10.1201/9781420010138}
}

@misc{oecd2025maritime,
  author       = {{OECD}},
  title        = {Maritime Transport {CO$_2$} Emissions: Methodology Update},
  year         = {2025},
  month        = may,
  howpublished = {Methodological note for the 2025-Q1 release of the OECD maritime transport CO$_2$ emissions database},
  url          = {https://www.oecd.org/content/dam/oecd/en/data/methods/eea/OECD-maritime-transport-CO2-emissions-methodology.pdf}
}
\end{singlespace}

\clearpage
\appendix
\section*{Online Appendix}
\setcounter{figure}{0}
\setcounter{table}{0}
\renewcommand{\thefigure}{A.\arabic{figure}}
\renewcommand{\thetable}{A.\arabic{table}}
\renewcommand{\theHfigure}{appendix.\arabic{figure}}
\renewcommand{\theHtable}{appendix.\arabic{table}}
\setcounter{section}{1}

\subsection{Engineering Calculation Details}\label{app:calculation}

\paragraph*{Operating Phase Assignment.}
To assign auxiliary engine and boiler power values, we use a slightly simplified version of the \ac{IMO}'s procedure \parencite[see][]{faber2020fourth}. For simplicity, we combine the `at berth' and `anchored' operational phases. At speeds up to and including 3 knots, we consider a ship anchored, regardless of its distance from the coast. Between 3 and 5 knots, the ship is maneuvering if it is within 5 \ac{nm} of the coast. At further distances and/or higher speeds, the ship is considered to be at sea. Power values are then taken as per \textcite[68]{faber2020fourth}.

\paragraph*{Hourly Observation Assumption.}
We assume operating conditions are constant over each hour (the frequency of the interpolated observational data) so that hourly fuel consumption in \eqref{eqn:fc_inst} is obtained by multiplying instantaneous fuel consumption by one hour. This is a standard approximation in the maritime emissions literature \parencite{faber2020fourth}.

\paragraph*{Parameter Values.}
All parameter values required for the calculation are taken from \textcite{faber2020fourth}, including $\mathrm{SFC}_{\text{base}}$ (Table 19) and $C_i$, which combines correction factors for weather and hull fouling.

\paragraph*{Calibrated and Reported Consumption.}
\autoref{tab:calstats} summarizes reported and calibrated fuel consumption of dry bulk ships by year for the training and test samples, and \autoref{fig:twoway_cal} plots the two against each other.

\begin{figure}[H]
    \centering
    \includegraphics[width=0.55\textwidth]{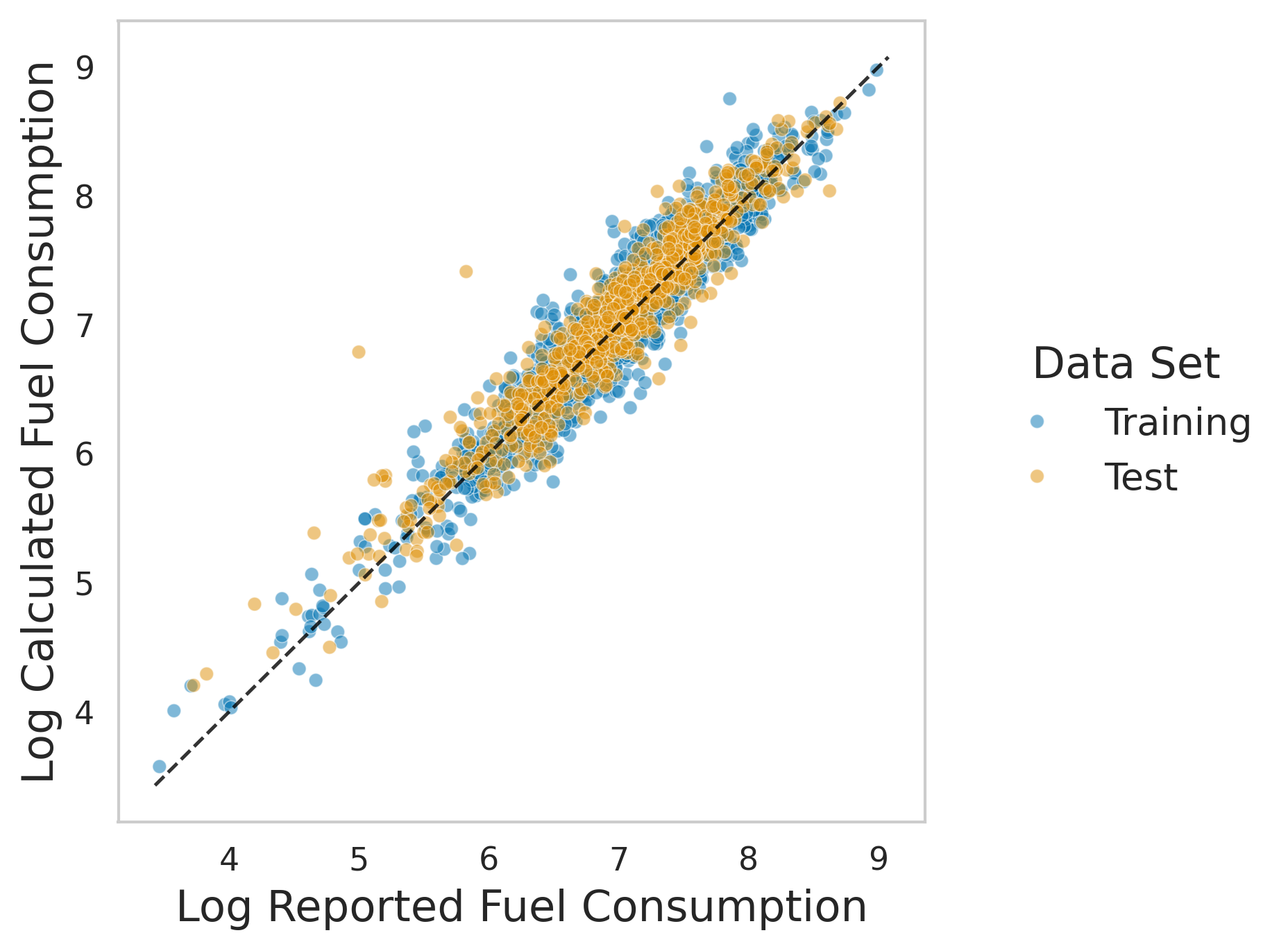}
    \caption{Calibrated versus reported fuel consumption, dry bulk ships.}
    \label{fig:twoway_cal}
\end{figure}

\begin{table}[H]
    \centering
    \begin{threeparttable}
        \caption{Reported and calibrated fuel consumption by year, dry bulk ships.}
        \label{tab:calstats}
        
\begin{tabular}[t]{>{\centering\arraybackslash}p{4em}>{\raggedleft\arraybackslash}p{4em}>{\raggedleft\arraybackslash}p{4em}>{\raggedleft\arraybackslash}p{4em}cr}
\toprule
\multicolumn{2}{c}{ } & \multicolumn{2}{c}{Reported} & \multicolumn{2}{c}{Calibrated} \\
\cmidrule(l{3pt}r{3pt}){3-4} \cmidrule(l{3pt}r{3pt}){5-6}
\multicolumn{1}{>{\centering\arraybackslash}p{4em}}{Year} & \multicolumn{1}{>{\centering\arraybackslash}p{4em}}{count} & \multicolumn{1}{>{\centering\arraybackslash}p{4em}}{mean (t)} & \multicolumn{1}{>{\centering\arraybackslash}p{4em}}{sd (t)} & \multicolumn{1}{c}{mean (t)} & \multicolumn{1}{c}{sd (t)}\\
\midrule
2019 & 819 & 1375 & 955 & 1408 & 985\\
2020 & 773 & 1285 & 876 & 1372 & 955\\
2021 & 851 & 1337 & 913 & 1424 & 976\\
\bottomrule
\end{tabular}

    \end{threeparttable}
\end{table}

\subsection{Data Cleaning Details}\label{app:MLDataCleaning}

\paragraph*{Raw Data.}
\ac{AIS} transmissions are recorded by both land- and satellite-based receivers, and our data set contains the last transmission in each hour for each transceiver identifier (\ac{MMSI}).

\paragraph*{Ship Identifier Assignment.}
Raw \ac{AIS} data comprises dynamic position reports (message types 1, 2, 3, and 27) containing ship movements, and static voyage messages (type 5) from which we extract the ship's \ac{IMO} number and draft. The \ac{IMO} number is a unique identifier for a ship, but the \ac{MMSI} identifies a transceiver, which may be transferred between ships, be spoofed, or otherwise change over time. To ease computational requirements, we process dry bulk and container ship data separately, relying on the \ac{WFR} to identify the ship type for each \ac{IMO} number. For each ship type, we retain all \ac{AIS} observations (static and dynamic) with an \ac{MMSI} that is associated at any point with an \ac{IMO} number of a ship of the given type.\footnote{We consider dry bulk to include all ships over 10,000 \acs{DWT} categorized in the \acs{WFR} under Ore Carrier, Bulk Carrier, Chip Carrier, Open Hatch Carrier, Forest Product Carrier, Aggregates Carrier, Cement Carrier, Nickel Carrier, and Slurry Carrier.} Since \ac{IMO} numbers are entered in \ac{AIS} data manually, they are often incorrect or missing. We validate them by checksum and discard ships that cannot be matched to the \ac{WFR}. Next, we split trajectories into segments where an \ac{MMSI}-ship association may have changed, as indicated by a physically infeasible movement: an implied speed above 25 knots (35 knots for container ships), or a distance between consecutive observations exceeding eight hours of travel at that speed. Throughout, we employ the haversine formula to calculate distances between location observations. Implied speed is computed using the distance and elapsed time between observations, which is distinct from the instantaneous speed provided in dynamic messages. We merge draft and \ac{IMO} number from each static message into the subsequent dynamic message (within each \ac{MMSI}) and drop the static messages. Any \ac{IMO} numbers missing after this step are filled with the modal \ac{IMO} number within each trajectory segment, resulting in a data set with an \ac{IMO} number for each observation. We drop any ship with fewer than ten observations, fill missing draft observations within each ship, and bound draft by the ship's maximum draft from the \ac{WFR}. Observations with a reported speed above 1.3 times the ship's rated speed are dropped rather than clipped to the bound as in \textcite{faber2020fourth}: most are the \ac{AIS} `not available' sentinel transmitted by stationary vessels, and clipping would concentrate observations at the bound, where the cubic load term is most sensitive.

\paragraph*{Trajectory Cleaning.}
Within this \ac{IMO}-labeled data set, we still observe physically infeasible jumps in location, either as single anomalous observations or as a ship appearing to jump to a new location, pursue a feasible trajectory for some time, and then jump back. We first drop single observations whose implied speed exceeds 1.3 times the ship's rated speed, then split each ship's trajectory wherever the implied speed between the remaining consecutive observations exceeds the same threshold. We recombine even segments together and odd segments together, with unique identifiers.

A second important error is missing dynamic data, i.e. hours for which there is no observation for a ship. We interpolate these missing data points as per the \ac{IMO}'s strategy \parencite{faber2020fourth}, creating an observation for each missing hour, assuming a constant speed and using spherical linear interpolation to impute location coordinates. Speed is interpolated as the distance between consecutive locations divided by the time difference, capped at the same 1.3 multiple of rated speed, and the previous draft value is assigned.

\paragraph*{Port Call Detection.}
We define a port call as when a ship does not exceed a speed of one knot for at least 12 hours and is within 5 km from a coast. We exclude port calls detected near the entrances to the Panama and Suez Canals where ships frequently queue before passing through. Our detection strategy may flag stops for refueling or other purposes, while the \ac{MRV} includes only stops where cargo is loaded or unloaded. However, this is not reliably observable from the tracking data \parencite{eu2015regulation}. For inclusion in the \ac{MRV} data, ports in the United Kingdom ceased to be included as of the beginning of 2021, and we adjust accordingly.

\paragraph*{Rated Speed Correction.}
The reference speed in \eqref{eqn:admiralty} is the ship's rated speed, taken from the service speed field of the \ac{WFR}. Listed values are occasionally erroneous, so we check each for internal consistency with the register's rated power and quoted fuel consumption at that speed. Where the engine load implied by the quote falls well below the design point, we replace the listed speed with the value consistent with the design load. Where no consumption quote is available, we use a register-side regression prediction instead when both the prediction and the listed operational speed indicate the listed value is understated. This corrects 1.6\% of dry bulk and 4.1\% of container ship register entries, by median factors of 1.15 and 1.24 respectively.

\subsection{Weather Correction}\label{app:weather}

To correct for weather effects, we construct speeds corrected for ocean currents and main engine power corrected for wave height,\footnote{While there are other potentially important factors, such as wind resistance, we focus on waves as this is the most important factor for dry bulk ships \parencite{man2018basic}.} drawing on reanalysis data from the Copernicus Marine Environment Monitoring Service, which reconstructs historical conditions by combining physics models with observational data. The ocean current data is daily, at 1/12$\degree$ resolution \parencite{cmems2023glorys}, and wave height data is 3-hourly at 1/5$\degree$ resolution \parencite{cmems2024waverys}. We interpolate this data to each \ac{AIS} observation. Speed through water is calculated by vector arithmetic of the current velocity from the \ac{AIS} speed over ground (heading is calculated from the preceding location observation). The effects of waves are more complex, depending on wave frequency, hull dimensions and shape, and speed regime. As a tractable approximation, we take the added resistance in head waves from \textcite{mepc2021guidelines}, in newtons,

\begin{equation*}
R_{\textit{wave}}=
\begin{cases}
1336\,(5.3+v)\left(\dfrac{B\,t}{L}\right)^{0.75} H^{2},
& \text{if wave encounter angle} \le 30\degree,\\
0 & \text{otherwise,}
\end{cases}
\end{equation*}

where $v$ is speed through water in m/s, $t$ is draft, $B$ is the ship width, $L$ is the length between perpendiculars, and $H$ is the significant wave height, all in meters. The corresponding power, $W_{\textit{wave}} = R_{\textit{wave}}\,v$, is required for traveling into waves incident at 30$\degree$ or less and is proportional to the square of significant wave height. This value is added to the main engine power from \eqref{eqn:admiralty}, with the weather correction factor set to one to avoid double-counting.
As the main concern is attenuation of the speed relation, we compare the corrected to the uncorrected speed counterfactual of \autoref{sec:elasticity}. Results using speed through water are virtually indistinguishable from those using speed over ground. Additionally accounting for wave resistance brings the calibration closer to reported consumption, with an \ac{RTE} of $-4.6\%$ on the test set against 6.5\% without it. It also attenuates the speed elasticity, but the degree of attenuation of the `Flexible' class and the unrestricted structural fits relative to the corresponding calibration remains similar (\autoref{fig:elasticity_weather}). The constrained structural fits are the exception: under the wave correction their response to speed reductions falls to about a third of its uncorrected size while their response to increases is little changed, so the correction changes the shape of the constrained fit rather than its overall attenuation.
\begin{figure}[H]
    \centering
    \includegraphics[width=\textwidth]{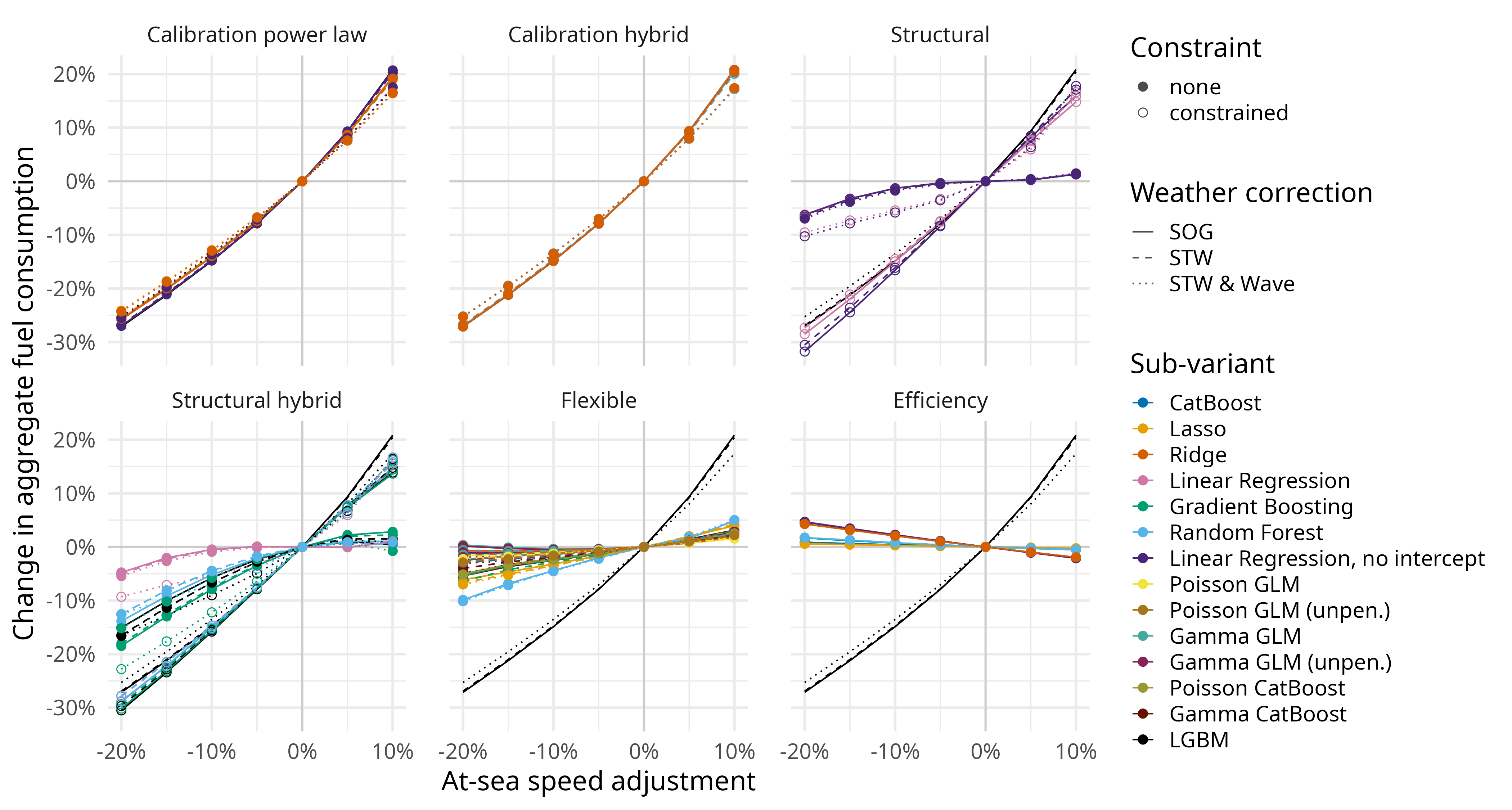}
    \caption{Predicted change in aggregate fuel consumption in the speed-adjusted counterfactual exercise under three treatments of weather, dry bulk ships: speed over ground (solid), speed through water (dashed), and speed through water with wave-corrected power (dotted). The black line of each type is the theoretical benchmark under the same treatment.}
    \label{fig:elasticity_weather}
\end{figure}

\subsection{Alternative Data Inclusion Criterion}\label{app:distance}

Our main results use a data subset for which observed distances match reported distances within 500 \ac{nm}. For robustness, we repeat our analyses for dry bulk ships with an alternative criterion of $\pm 10\%$, which yields a substantially larger data set of 3,703 ship-years, against 2,443 under the main criterion, with greater spread and more outliers. \autoref{fig:samplescores_alt} shows the out-of-sample results under this criterion, \autoref{fig:support_quartiles_alt} the quartile splits, and \autoref{fig:elasticity_alt} the speed counterfactual.

\begin{figure}[H]
    \centering
    \includegraphics[width=0.9\textwidth]{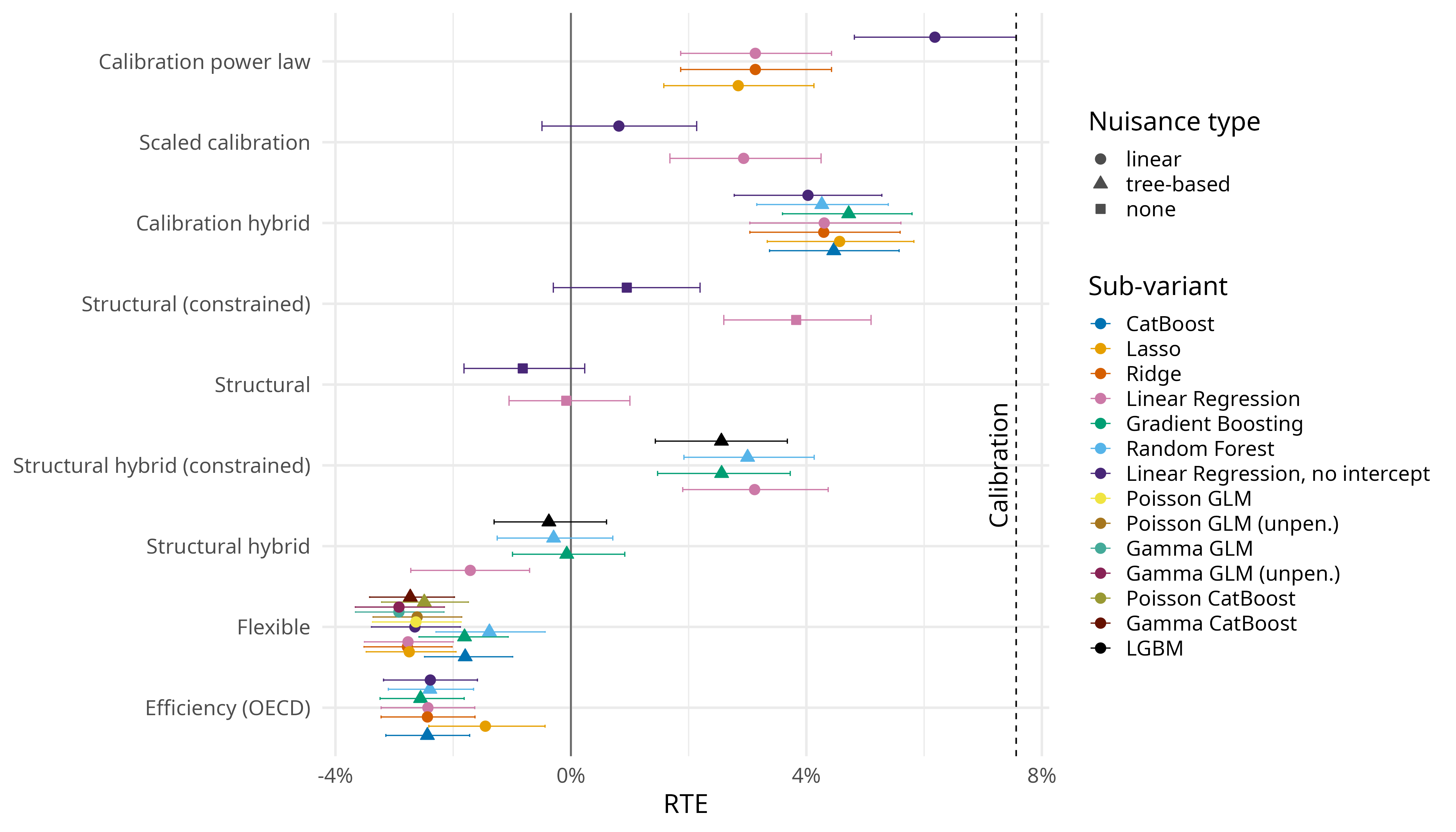}
    \caption{Out-of-sample \ac{RTE} by model specification and sub-variant, dry bulk ships, alternative inclusion criterion. The dashed line is the calibration. Bars are 95\% ship-clustered bootstrap intervals over the test sample, fitted models held fixed.}
    \label{fig:samplescores_alt}
\end{figure}

\begin{figure}[H]
    \centering
    \includegraphics[width=\textwidth]{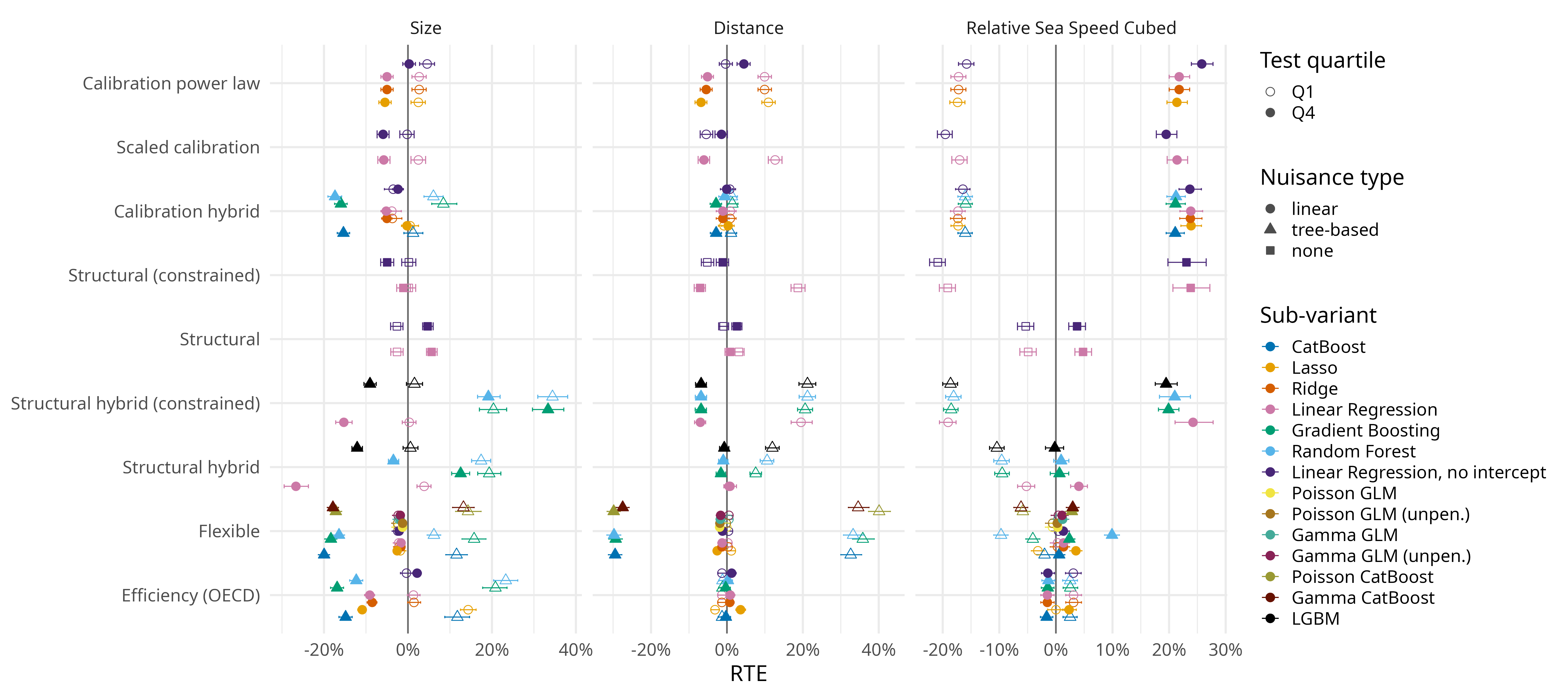}
    \caption{Out-of-support \ac{RTE} for three quartile splits, dry bulk ships, alternative inclusion criterion. Each model is trained on three quartiles and tested on the fourth. Bars are 95\% ship-clustered bootstrap intervals over the test quartile, fitted models held fixed.}
    \label{fig:support_quartiles_alt}
\end{figure}

\begin{figure}[H]
    \centering
    \includegraphics[width=\textwidth]{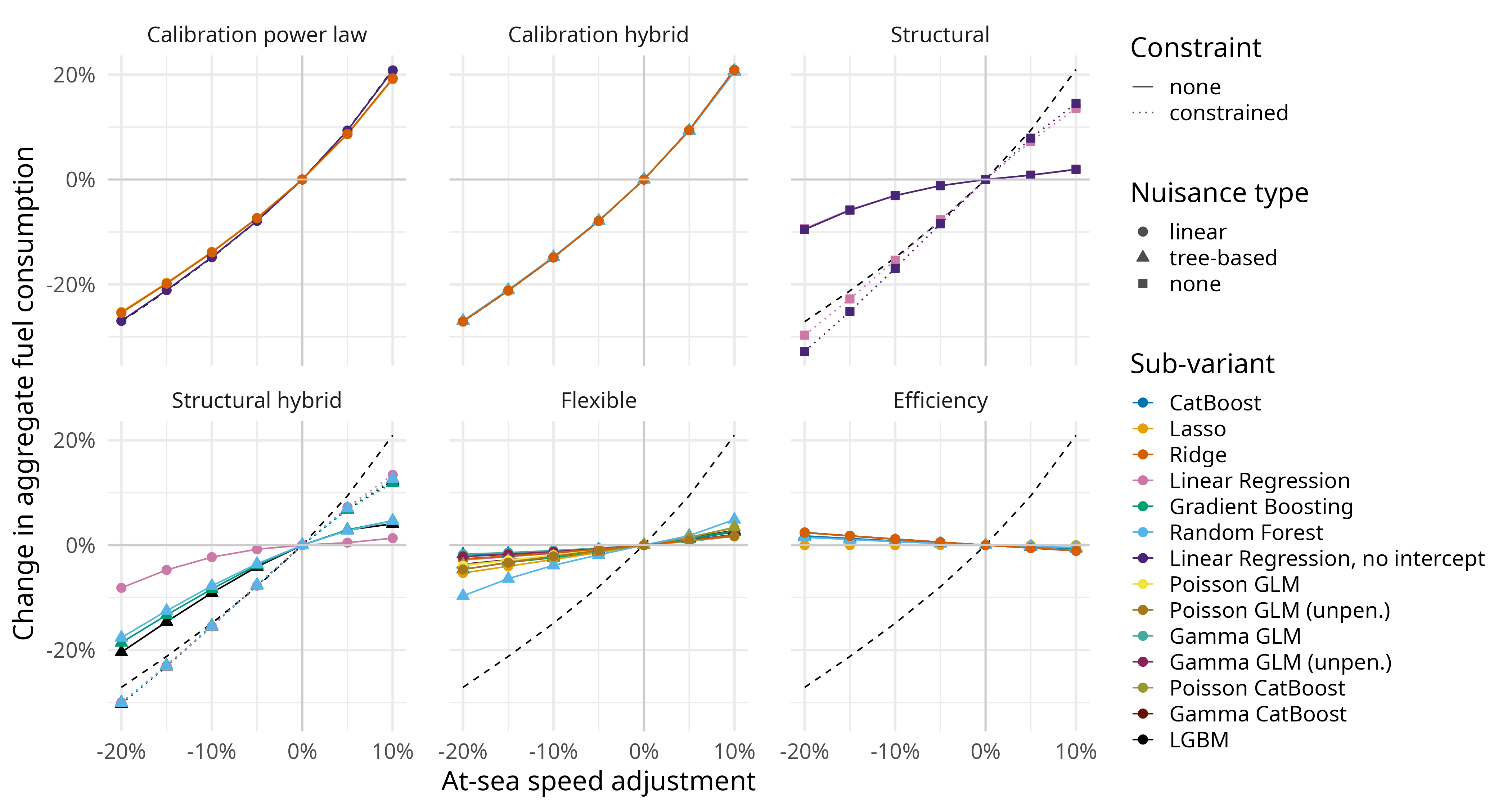}
    \caption{Predicted change in aggregate fuel consumption in the speed-adjusted counterfactual exercise, dry bulk ships, alternative inclusion criterion. The dashed line is the theoretical benchmark.}
    \label{fig:elasticity_alt}
\end{figure}

\subsection{Summary Statistics}\label{app:MLsummarystats}
Observation counts and histograms for container ships, corresponding to \autoref{tab:obscounts} and \autoref{fig:datasupporthistograms} for dry bulk ships, are provided in \autoref{tab:obscounts_container} and \autoref{fig:datasupporthistograms_container}. Summary statistics of the `Flexible' feature set (\autoref{tab:features}) over the training and test samples are in \autoref{tab:MLsummarystats} for dry bulk ships and \autoref{tab:MLsummarystats_container} for container ships.

\begin{table}[H]
    \centering
    \begin{threeparttable}
        \caption{Observation counts for data subsets, container ships.}
        \label{tab:obscounts_container}
        
\begin{tabular}[t]{lrr}
\toprule
Data Subset & \# Ships & \# Observations\\
\midrule
EU Trips & 2221 & 5823\\
Train + Test & 849 & 1036\\
Monthly Prediction & 4766 & 176622\\
Annual Route Prediction & 4747 & 222131\\
\bottomrule
\end{tabular}

    \end{threeparttable}
\end{table}

\begin{figure}[H]
    \centering
    \includegraphics[width=0.9\textwidth]{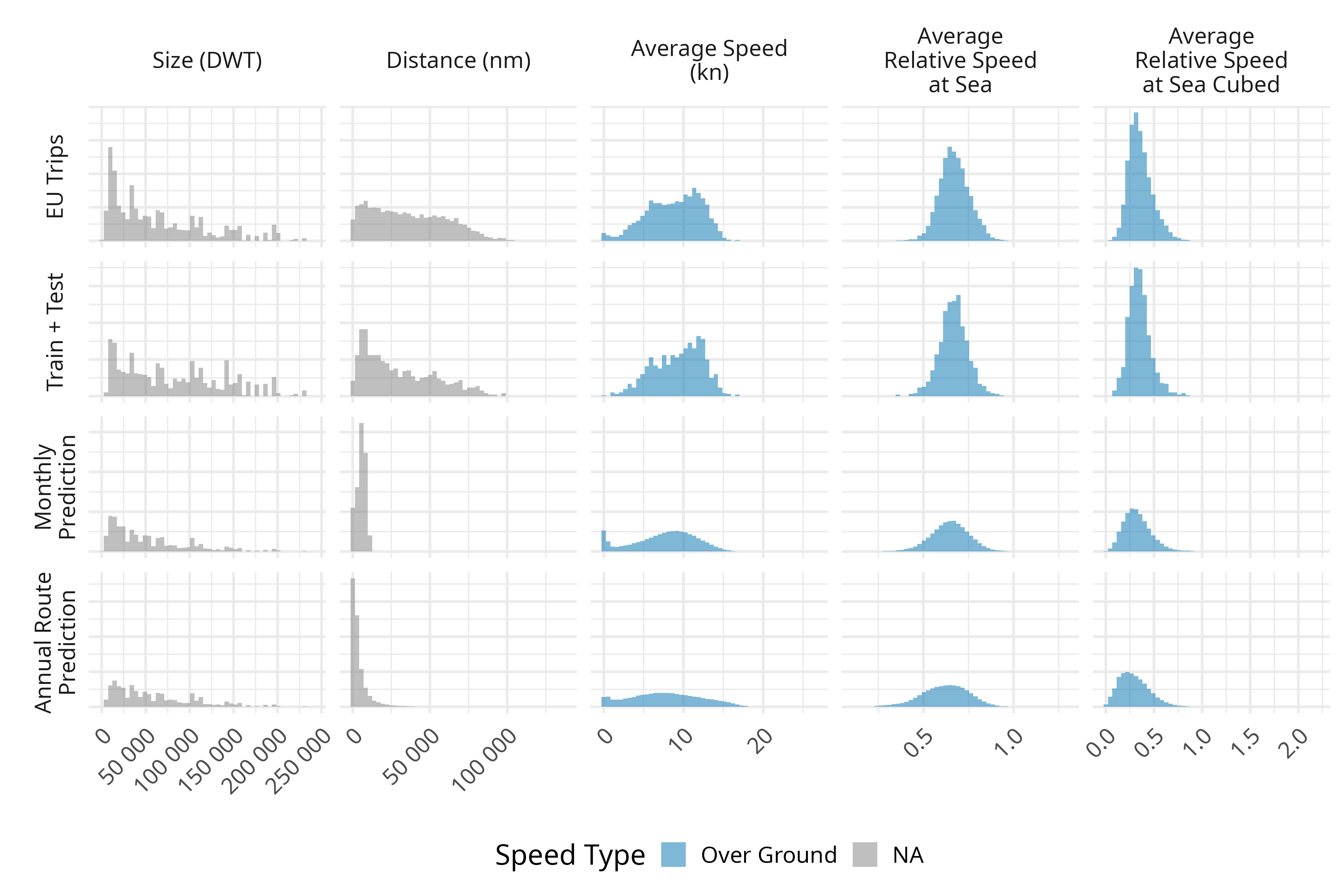}
    \caption{Histograms of variables of interest relevant for prediction applications across data subsets, container ships. Axes vary.}
    \label{fig:datasupporthistograms_container}
\end{figure}

\begin{table}[H]
    \centering
    \begin{adjustbox}{max width=\textwidth}
    \begin{threeparttable}
        \caption{Summary statistics of the `Flexible' feature set, dry bulk ships.}
        \label{tab:MLsummarystats}
        
\begin{tabular}[t]{lrrrrrr}
\toprule
\multicolumn{1}{c}{ } & \multicolumn{3}{c}{Training} & \multicolumn{3}{c}{Test} \\
\cmidrule(l{3pt}r{3pt}){2-4} \cmidrule(l{3pt}r{3pt}){5-7}
Variable & count & mean & std & count & mean & std\\
\midrule
Age (years) & 1592 & 9.5 & 5.5 & 851 & 10.7 & 5.3\\
Beam Moulded (meters) & 1592 & 32.4 & 5.3 & 851 & 32.3 & 5.7\\
Deadweight Tonnage (tonnes) & 1592 & 6.70e+04 & 4.20e+04 & 851 & 6.72e+04 & 4.47e+04\\
Draft (meters) & 1592 & 12.8 & 2.5 & 851 & 12.7 & 2.6\\
Gross Tonnage & 1592 & 3.76e+04 & 2.05e+04 & 851 & 3.77e+04 & 2.20e+04\\
\addlinespace
Length Between Perpendiculars (meters) & 1592 & 199 & 35.7 & 851 & 199 & 37.3\\
Length Overall (meters) & 1592 & 206 & 35.9 & 851 & 205 & 37.5\\
Light Displacement Tonnage (tonnes) & 339 & 1.19e+04 & 5911 & 176 & 1.16e+04 & 5847\\
Main Bunker Capacity (cubic meters) & 1117 & 2499 & 1181 & 579 & 2503 & 1340\\
Main Engine Power (kilowatts) & 1592 & 9600 & 3375 & 851 & 9581 & 3575\\
\addlinespace
Net Tonnage & 1588 & 2.22e+04 & 1.37e+04 & 850 & 2.23e+04 & 1.45e+04\\
Speed (knots) & 1592 & 14.1 & 0.81 & 851 & 14.1 & 0.79\\
Tonnage Per Centimeter (tonnes/cm) & 1372 & 64.4 & 22.8 & 725 & 63.3 & 22.6\\
At-Port Fraction & 1592 & 0.39 & 0.15 & 851 & 0.43 & 0.15\\
Distance Traveled (nautical miles) & 1592 & 1.41e+04 & 9285 & 851 & 1.37e+04 & 8977\\
\addlinespace
Number Of Trips & 1592 & 6.4 & 7.3 & 851 & 7.2 & 10.7\\
Admiralty Draft Term (Mean) & 1592 & 0.87 & 0.06 & 851 & 0.87 & 0.07\\
Admiralty Rel Power & 1592 & 577 & 405 & 851 & 599 & 422\\
Admiralty Rel Power Cubed & 1592 & 198 & 200 & 851 & 235 & 259\\
Admiralty Rel Power Squared & 1592 & 318 & 262 & 851 & 356 & 302\\
\addlinespace
Admiralty Speed Term & 1592 & 660 & 457 & 851 & 686 & 479\\
Calibrated Fuel Consumption (tonnes) & 1592 & 1390 & 970 & 851 & 1424 & 976\\
Main Engine Reference SFC & 1592 & 176 & 3.2 & 851 & 176 & 3.2\\
Relative Draft (Mean) & 1592 & 0.82 & 0.09 & 851 & 0.82 & 0.09\\
Relative Speed & 1592 & 992 & 654 & 851 & 969 & 646\\
\addlinespace
Infeasible Distance (nautical miles) & 1592 & 0.00e+00 & 0.00e+00 & 851 & 0.00e+00 & 0.00e+00\\
Interpolated Fraction At-Sea & 1592 & 0.26 & 0.14 & 851 & 0.15 & 0.12\\
\bottomrule
\end{tabular}

        \begin{tablenotes}[flushleft]\small
            \item Training sample 2019--2020, test sample 2021. All derived values aggregated annually over \ac{EU} trips.
        \end{tablenotes}
    \end{threeparttable}
    \end{adjustbox}
\end{table}

\begin{table}[H]
    \centering
    \begin{adjustbox}{max width=\textwidth}
    \begin{threeparttable}
        \caption{Summary statistics of the `Flexible' feature set, container ships.}
        \label{tab:MLsummarystats_container}
        
\begin{tabular}[t]{lrrrrrr}
\toprule
\multicolumn{1}{c}{ } & \multicolumn{3}{c}{Training} & \multicolumn{3}{c}{Test} \\
\cmidrule(l{3pt}r{3pt}){2-4} \cmidrule(l{3pt}r{3pt}){5-7}
Variable & count & mean & std & count & mean & std\\
\midrule
Age (years) & 663 & 11.8 & 6.4 & 373 & 12.7 & 5.9\\
Beam Moulded (meters) & 663 & 38.2 & 11.0 & 373 & 38.0 & 10.6\\
Deadweight Tonnage (tonnes) & 663 & 8.16e+04 & 5.74e+04 & 373 & 7.95e+04 & 5.41e+04\\
Draft (meters) & 663 & 12.8 & 2.7 & 373 & 12.8 & 2.7\\
Gross Tonnage & 663 & 7.56e+04 & 5.79e+04 & 373 & 7.34e+04 & 5.45e+04\\
\addlinespace
Length Between Perpendiculars (meters) & 660 & 259 & 81.4 & 370 & 256 & 79.6\\
Length Overall (meters) & 663 & 273 & 84.4 & 373 & 271 & 82.8\\
Light Displacement Tonnage (tonnes) & 138 & 2.29e+04 & 1.85e+04 & 70 & 2.02e+04 & 1.53e+04\\
Main Bunker Capacity (cubic meters) & 177 & 7346 & 5163 & 92 & 7431 & 5265\\
Main Engine Power (kilowatts) & 663 & 4.15e+04 & 2.37e+04 & 373 & 4.21e+04 & 2.43e+04\\
\addlinespace
Net Tonnage & 663 & 3.82e+04 & 2.89e+04 & 372 & 3.69e+04 & 2.71e+04\\
Speed (knots) & 663 & 22.2 & 2.7 & 373 & 22.2 & 2.8\\
Tonnage Per Centimeter (tonnes/cm) & 38 & 54.2 & 26.7 & 14 & 55.7 & 21.7\\
At-Port Fraction & 663 & 0.37 & 0.16 & 373 & 0.41 & 0.15\\
Distance Traveled (nautical miles) & 663 & 3.06e+04 & 2.39e+04 & 373 & 2.95e+04 & 2.24e+04\\
\addlinespace
Number Of Trips & 663 & 29.9 & 31.5 & 373 & 31.3 & 34.8\\
Admiralty Draft Term (Mean) & 663 & 0.92 & 0.05 & 373 & 0.93 & 0.05\\
Admiralty Rel Power & 663 & 679 & 615 & 373 & 702 & 623\\
Admiralty Rel Power Cubed & 663 & 170 & 256 & 373 & 201 & 302\\
Admiralty Rel Power Squared & 663 & 309 & 361 & 373 & 343 & 393\\
\addlinespace
Admiralty Speed Term & 663 & 735 & 661 & 373 & 751 & 663\\
Calibrated Fuel Consumption (tonnes) & 663 & 6682 & 6072 & 373 & 6826 & 6328\\
Main Engine Reference SFC & 663 & 177 & 5.1 & 373 & 177 & 5.2\\
Relative Draft (Mean) & 663 & 0.88 & 0.07 & 373 & 0.89 & 0.07\\
Relative Speed & 663 & 1407 & 1113 & 373 & 1362 & 1069\\
\addlinespace
Infeasible Distance (nautical miles) & 663 & 5.4 & 20.2 & 373 & 3.7 & 15.1\\
Interpolated Fraction At-Sea & 663 & 0.24 & 0.14 & 373 & 0.14 & 0.10\\
\bottomrule
\end{tabular}

        \begin{tablenotes}[flushleft]\small
            \item Training sample 2019--2020, test sample 2021. All derived values aggregated annually over \ac{EU} trips.
        \end{tablenotes}
    \end{threeparttable}
    \end{adjustbox}
\end{table}

\subsection{Disaggregation Details}\label{app:disaggregation}

For model training, we prefer to subset data rather than correct infeasible observations, and allow the estimation process to handle errors. For prediction exercises, we split trajectories at observations with infeasibly high implied speeds (above 25 knots for bulk ships and 35 knots for container ships) and treat the speed and distance of these observations as missing in aggregation. The prediction data are built from the same hourly quantities as the training features. For the temporal exercise, each hourly observation is assigned to the calendar month of its timestamp, so a trip spanning two months contributes to both. For route attribution, each trip is assigned an origin and a destination country from the port calls at either end, using the port call detection of Appendix~\ref{app:MLDataCleaning}, and we label the infeasible points as `jump' points (JMP). Trips with the same origin and destination country are not reported, as many of these arise from undetected port calls. Trips with at least one end labeled as `JMP' constitute around 0.02\% of total emissions for both dry bulk and container ships.

\subsection{Machine Learning Features and Training}\label{app:MLfeatures}

\autoref{tab:features} lists the `Flexible' feature set, which is identical for dry bulk and container ships. Missing feature values are imputed by their medians before training; most features are complete, but light displacement tonnage, bunker capacity, and tonnage per centimeter are missing for a substantial share of ships (\autoref{tab:MLsummarystats}). The `Efficiency' feature set reproduces that of \textcite{clarke2023co2} as closely as possible; because we use a different fleet register, some ship characteristics differ slightly or are absent. It comprises ship characteristics (main engine power, beam, deadweight, gross and net tonnage, light displacement, length overall, bunker capacity, moulded depth, and cubic capacity) together with annual distance traveled, mean speed, and mean draft. The nuisance feature sets of the hybrids are subsets of the `Flexible' set. The calibration hybrid's nuisance drops the calibrated fuel consumption, the Admiralty and \ac{SFC} load terms, the speed aggregates, rated speed, and the at-port fraction, retaining 19 features. The structural hybrid's controls further drop reference power, base \ac{SFC}, distance, number of trips, and the mean draft terms, leaving the 13 ship characteristics and data-quality indicators.

\begin{table}[H]
    \centering
    \begin{adjustbox}{max width=\textwidth}
    \begin{threeparttable}
        \caption{Feature set for the `Flexible' class.}
        \label{tab:features}
        
\begin{tabular}[t]{>{\raggedright\arraybackslash}p{16em}>{\raggedright\arraybackslash}p{30em}}
\toprule
Variable & Description\\
\midrule
\addlinespace[0.3em]
\multicolumn{2}{l}{\textbf{Ship Characteristics}}\\
\hspace{1em}Age & Age at observation (years)\\
\hspace{1em}Beam Moulded & Maximum width (meters)\\
\hspace{1em}Deadweight Tonnage & Capacity (tonnes)\\
\hspace{1em}Draft & Maximum draft (meters)\\
\hspace{1em}Gross Tonnage & Nonlinear measure of internal volume\\
\hspace{1em}Length Between Perpendiculars & Length between the forward and aft perpendiculars (meters)\\
\hspace{1em}Length Overall & Total length (meters)\\
\hspace{1em}Light Displacement Tonnage & Weight excluding most contents (tonnes)\\
\hspace{1em}Main Bunker Capacity & Fuel capacity (cubic meters)\\
\hspace{1em}Main Engine Power & Power rating of main engine (kilowatts)\\
\hspace{1em}Net Tonnage & Dimensionless index based on the total volume of cargo spaces\\
\hspace{1em}Speed & Representative speed from \ac{WFR} (knots)\\
\hspace{1em}Tonnage Per Centimeter & Load change required to change draft by one centimeter (tonnes/cm)\\
\addlinespace[0.3em]
\multicolumn{2}{l}{\textbf{Ship Activity}}\\
\hspace{1em}At-Port Fraction & Fraction of hourly observations with operational phase ‘at-port’\\
\hspace{1em}Distance Traveled & Distance calculated from AIS data (nautical miles)\\
\hspace{1em}Number Of Trips & Number of trips detected from AIS data\\
\addlinespace[0.3em]
\multicolumn{2}{l}{\textbf{Calibrated}}\\
\hspace{1em}Admiralty Draft Term (Mean) & Sea-hours mean of draft term of (\ref{eqn:admiralty}), calculated as $\overline{(t/t_{\textit{ref}})^{0.66}}$\\
\hspace{1em}Admiralty Rel Power & Aggregation of relative power term of (\ref{eqn:admiralty}), calculated as $\sum (t/t_{\textit{ref}})^{0.66}*(v/v_{\textit{ref}})^{3}$\\
\hspace{1em}Admiralty Rel Power Cubed & Aggregation of relative power term of (\ref{eqn:admiralty}) cubed, calculated as $\sum ((t/t_{\textit{ref}})^{0.66}*(v/v_{\textit{ref}})^{3})^3$\\
\hspace{1em}Admiralty Rel Power Squared & Aggregation of relative power term of (\ref{eqn:admiralty}) squared, calculated as $\sum ((t/t_{\textit{ref}})^{0.66}*(v/v_{\textit{ref}})^{3})^2$\\
\hspace{1em}Admiralty Speed Term & Aggregation of speed term of (\ref{eqn:admiralty}), calculated as $\sum (v/v_{\textit{ref}})^{3}$\\
\hspace{1em}Calibrated Fuel Consumption & Theoretical fuel consumption, calculated as per (\ref{eqn:annual_fc})\\
\hspace{1em}Main Engine Reference SFC & Base value of fuel consumption rating of main engine as per IMO\\
\hspace{1em}Relative Draft (Mean) & Sea-hours mean of relative draft, calculated as $\overline{(t/t_{\textit{ref}})}$\\
\hspace{1em}Relative Speed & Aggregation of relative speed, calculated as $\sum (v/v_{\textit{ref}})$\\
\addlinespace[0.3em]
\multicolumn{2}{l}{\textbf{Data Quality}}\\
\hspace{1em}Infeasible Distance & Sum of infeasible distances between consecutive locations after interpolation (nautical miles)\\
\hspace{1em}Interpolated Fraction At-Sea & Fraction of hourly observations with operational phase ‘at-sea’ that are interpolated\\
\bottomrule
\end{tabular}

        \begin{tablenotes}[flushleft]\small
            \item All derived values aggregated annually over \ac{EU} trips.
        \end{tablenotes}
    \end{threeparttable}
    \end{adjustbox}
\end{table}
\paragraph*{Retransformation.}
For a model fit to $\log(1+Y)$ with residuals $\hat{u}_i$, the smearing estimate of \textcite{duan1983smearing} multiplies the exponentiated prediction by the mean exponentiated residual, $\hat{\varphi} = n^{-1}\sum_i e^{\hat{u}_i}$, which we compute from cross-fitted training residuals for each model. The correction is valid when the residual distribution does not depend on the features, in particular on speed, which cannot be verified outside the training support. The pseudo-maximum likelihood fits use the Poisson or gamma deviance as the objective, each for both a generalized linear model and a CatBoost learner on the `Flexible' feature set.
\paragraph*{Ship-Level Error.}
The \ac{RTE} of \autoref{sec:MLevaluation} measures the error of the aggregate only. On a single test set, mean squared error decomposes as
\begin{equation*}
    \mathrm{RMSE}^2 = (\mathrm{RTE} \cdot \bar{y})^2 + \Var(e),
\end{equation*}
where $\bar{y}$ is mean reported consumption and $e$ the ship-level prediction error, so models with similar aggregate bias can differ widely in dispersion. \autoref{fig:samplermse} shows the \ac{RMSE} of the models in \autoref{fig:samplescores}. Dispersion falls with flexibility: the flexible and efficiency classes are tightest, the structural class and structural hybrid next, and the calibration-anchored models and the constrained structural variants widest. The bootstrap intervals of the flexible class lie below those of the calibration-anchored models and overlap those of the structural class only at the margin. The structural class therefore combines the smallest aggregate bias with a dispersion well above the flexible models.

\begin{figure}[H]
    \centering
    \includegraphics[width=\textwidth]{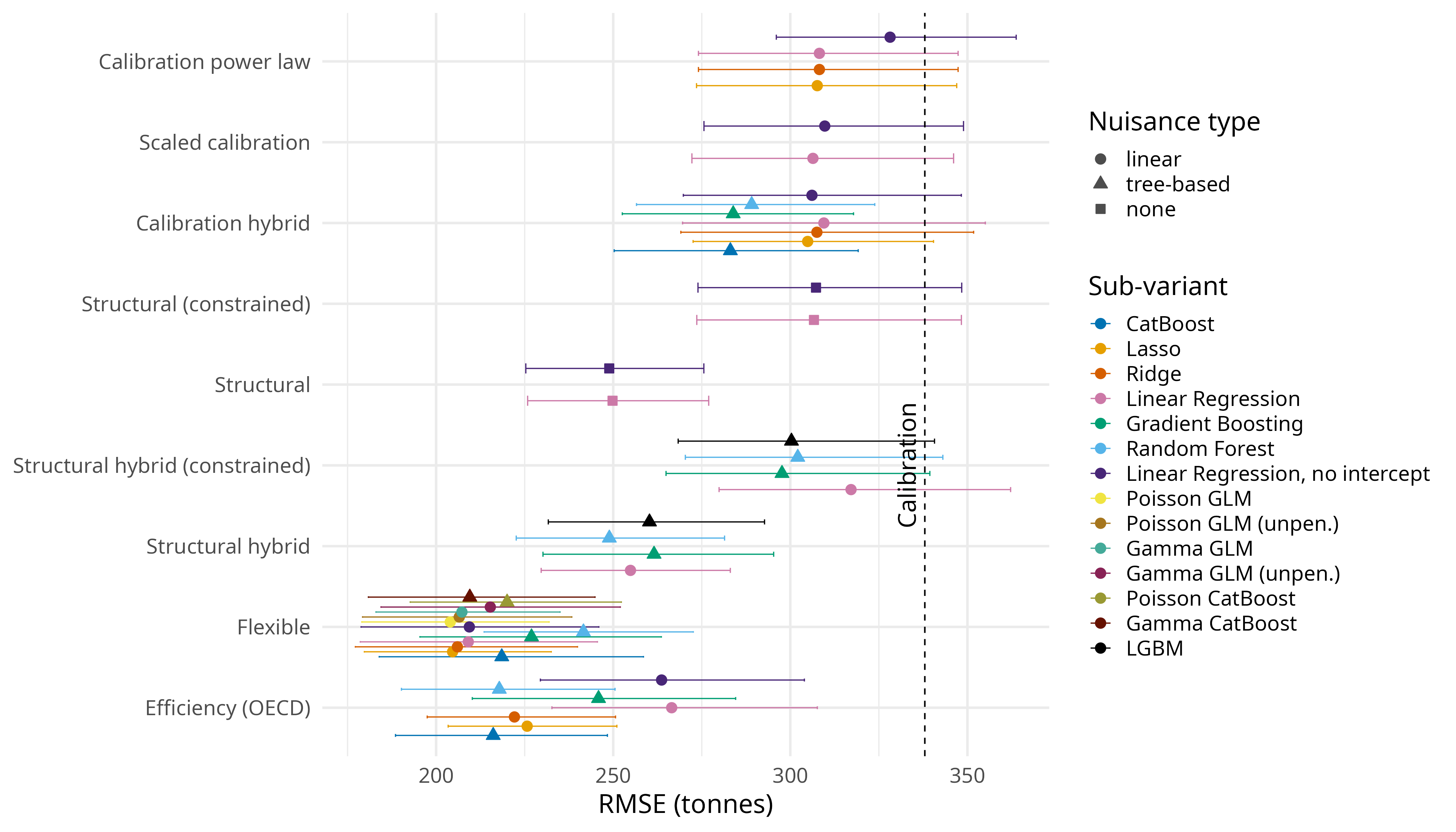}
    \caption{Out-of-sample \ac{RMSE} by model specification and sub-variant, dry bulk ships. The dashed line is the calibration. Bars are 95\% ship-clustered bootstrap intervals over the test sample, fitted models held fixed.}
    \label{fig:samplermse}
\end{figure}

\paragraph*{Tuning and Software.}
Hyperparameters are selected by 5-fold cross-validation over a grid for each learner: the penalty for lasso, ridge, and the penalized Poisson and gamma GLMs; tree depth for random forest; and learning rate and tree depth for the gradient boosting machine and CatBoost. The nuisance learners of the structural hybrid search coarser grids within each cross-fitting fold. The grids and all remaining settings are in the replication package. All models are fit in Python 3.12 with scikit-learn 1.7.1, CatBoost 1.2.8, and LightGBM 4.6.0.

\subsection{Hybrid Model Properties}\label{app:hybrid}

The two hybrids of \autoref{sec:hybridmodels} are
\begin{align*}
    \log \hat{Y}_i &= \log f_c(X_{f,i}) + \hat{g}(X_{g,i}) && \text{(calibration hybrid),}\\
    \hat{Y}_i &= \sum_k \hat{\beta}_k f_k(X_{f,i}) + \hat{g}(X_{g,i}) && \text{(structural hybrid),}
\end{align*}
where $X_g$ contains no variable derived from speed $v$, so that $\partial \hat{g}/\partial v = 0$ identically. All derivatives below hold $X_g$ fixed.

\begin{proposition}[Preservation under feature exclusion]\label{prop:preservation}
For every realization of $\hat{g}$, the calibration hybrid satisfies $\partial \log \hat{Y} / \partial \log v = \partial \log f_c / \partial \log v$, and the structural hybrid satisfies $\partial \hat{Y} / \partial v = \sum_k \hat{\beta}_k\, \partial f_k / \partial v$.
\end{proposition}
\begin{proof}
Differentiate each expression with respect to $\log v$ or $v$; the $\hat{g}$ term vanishes because it has no speed argument.
\end{proof}

\begin{corollary}[Bounded error under misspecification]\label{cor:preservation}
Let $f_0$ denote the true structural function. If $\sup_v \bigl|\partial \log f_c/\partial \log v - \partial \log f_0/\partial \log v\bigr| \le \delta$, then the elasticity of the calibration hybrid is within $\delta$ of the true elasticity for all $v$ and all $\hat{g}$. The analogous statement holds for the level derivative of the structural hybrid.
\end{corollary}

The preserved object is set by the composition. For the calibration hybrid, $\partial \hat{Y}/\partial v = e^{\hat{g}}\, \partial f_c/\partial v$, so the level derivative is rescaled by the nuisance; for the structural hybrid, the elasticity $v\, \partial \hat{Y}/\partial v \big/ \hat{Y}$ has $\hat{g}$ in its denominator and is shifted by it.

\paragraph*{Structural Hybrid Estimation.}
The structural hybrid is the partially linear model $Y_i = \sum_k \beta_k f_{k,i} + g(X_{g,i}) + u_i$ with the three aggregate load terms as regressors of interest, and we estimate $\beta$ by double/debiased machine learning \parencite{chernozhukov2018double}. The training sample is split into five folds; on the complement of each, nuisance learners for the target and for each load term are fit as functions of $X_g$, with hyperparameters selected within the complement, and predictions are formed for the held-out fold. The coefficients are the regression of the out-of-fold residual of the target on those of the load terms, pooled across folds, with standard errors clustered by ship. The procedure is repeated over five random splits and combined as in \textcite[Sec.~3.4]{chernozhukov2018double}. For prediction, the nuisance is refit on the full training sample to $Y_i - \sum_k \hat{\beta}_k f_{k,i}$. We match the nuisance learner to those used elsewhere as far as possible: gradient-boosted trees, random forest, and linear regression are available in the orthogonalized estimator, whereas CatBoost and the penalized linear learners are not, so where a flexible model uses one of the latter we present the structural hybrid with the most similar available learner.

Feature exclusion is what makes Proposition~\ref{prop:preservation} hold by construction; the transfer of the nuisance term to a counterfactual requires more.

\begin{assumption}[Nuisance invariance]\label{ass:invariance}
The conditional distribution of $(X_g, u)$ given the remaining features is unchanged by the counterfactual change in speed, where $u$ is the error on the scale where the composition is additive.
\end{assumption}

\begin{remark}[Transfer]
On the additive scale, the prediction error of a hybrid is $g_0(X_g) - \hat{g}(X_g) + u$ for a true nuisance $g_0$, which does not depend on $v$. Under Assumption~\ref{ass:invariance}, the distribution of this error is the same under the counterfactual as in the training data, which is the sense in which the hybrids are expected to transfer out of the speed support.
\end{remark}

\subsection{Container Ship Evaluation Results}\label{app:container_evaluation}

This section reports the container ship counterparts of the dry bulk evaluation results in \autoref{sec:MLevaluation}. The calibration overpredicts reported consumption by more than for dry bulk ships (\autoref{fig:twoway_cal_container}; \autoref{tab:calstats_container}), with an \ac{RTE} of 13.9\% on the test set. \autoref{fig:samplescores_container} shows the out-of-sample results.

\begin{figure}[H]
    \centering
    \includegraphics[width=0.55\textwidth]{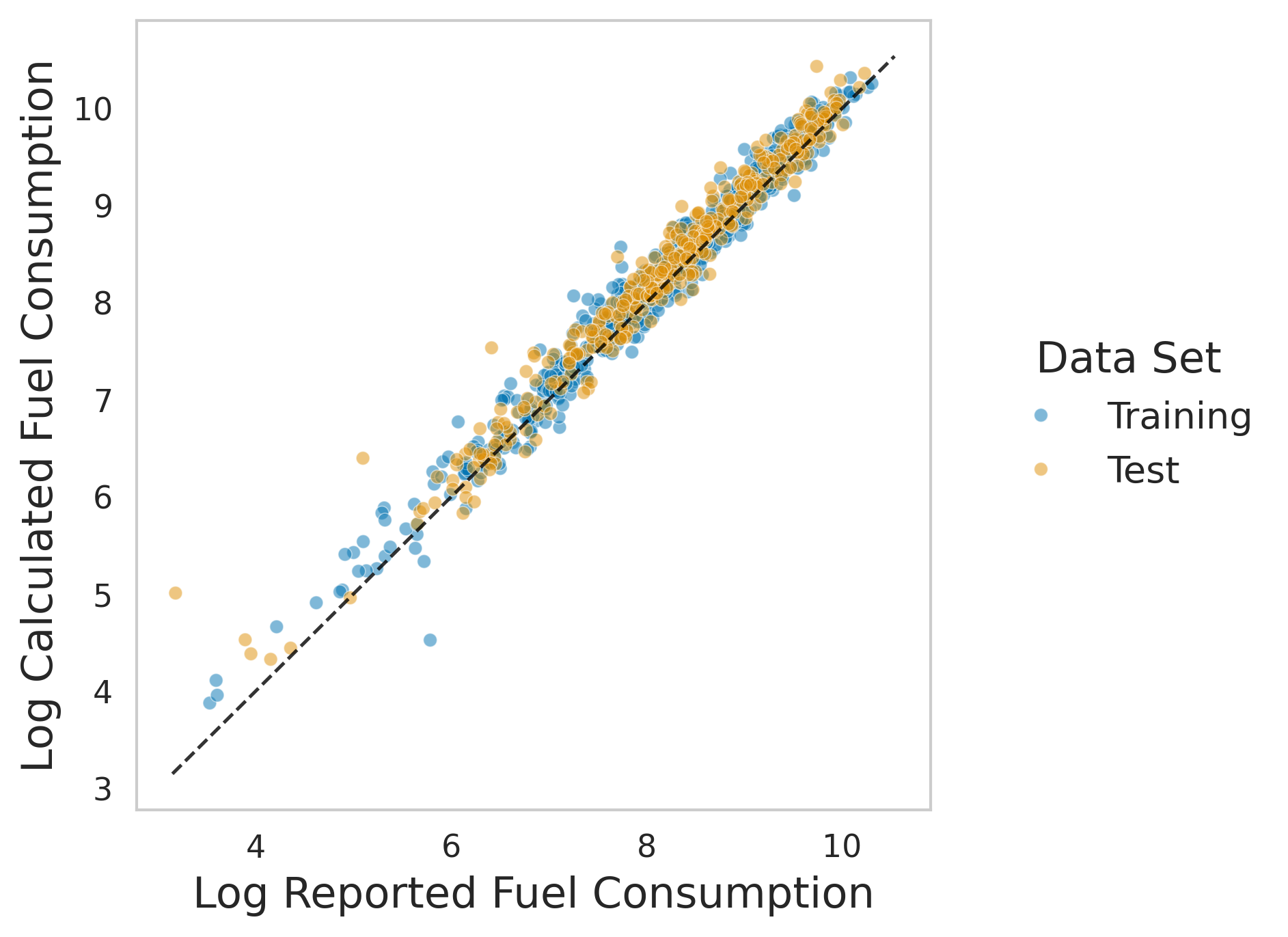}
    \caption{Calibrated versus reported fuel consumption, container ships.}
    \label{fig:twoway_cal_container}
\end{figure}

\begin{table}[H]
    \centering
    \begin{threeparttable}
        \caption{Reported and calibrated fuel consumption by year, container ships.}
        \label{tab:calstats_container}
        
\begin{tabular}[t]{>{\centering\arraybackslash}p{4em}>{\raggedleft\arraybackslash}p{4em}>{\raggedleft\arraybackslash}p{4em}>{\raggedleft\arraybackslash}p{4em}cr}
\toprule
\multicolumn{2}{c}{ } & \multicolumn{2}{c}{Reported} & \multicolumn{2}{c}{Calibrated} \\
\cmidrule(l{3pt}r{3pt}){3-4} \cmidrule(l{3pt}r{3pt}){5-6}
\multicolumn{1}{>{\centering\arraybackslash}p{4em}}{Year} & \multicolumn{1}{>{\centering\arraybackslash}p{4em}}{count} & \multicolumn{1}{>{\centering\arraybackslash}p{4em}}{mean (t)} & \multicolumn{1}{>{\centering\arraybackslash}p{4em}}{sd (t)} & \multicolumn{1}{c}{mean (t)} & \multicolumn{1}{c}{sd (t)}\\
\midrule
2019 & 263 & 6227 & 5357 & 6830 & 5857\\
2020 & 400 & 5872 & 5594 & 6585 & 6215\\
2021 & 373 & 5993 & 5532 & 6826 & 6328\\
\bottomrule
\end{tabular}

    \end{threeparttable}
\end{table}

\begin{figure}[H]
    \centering
    \includegraphics[width=0.9\textwidth]{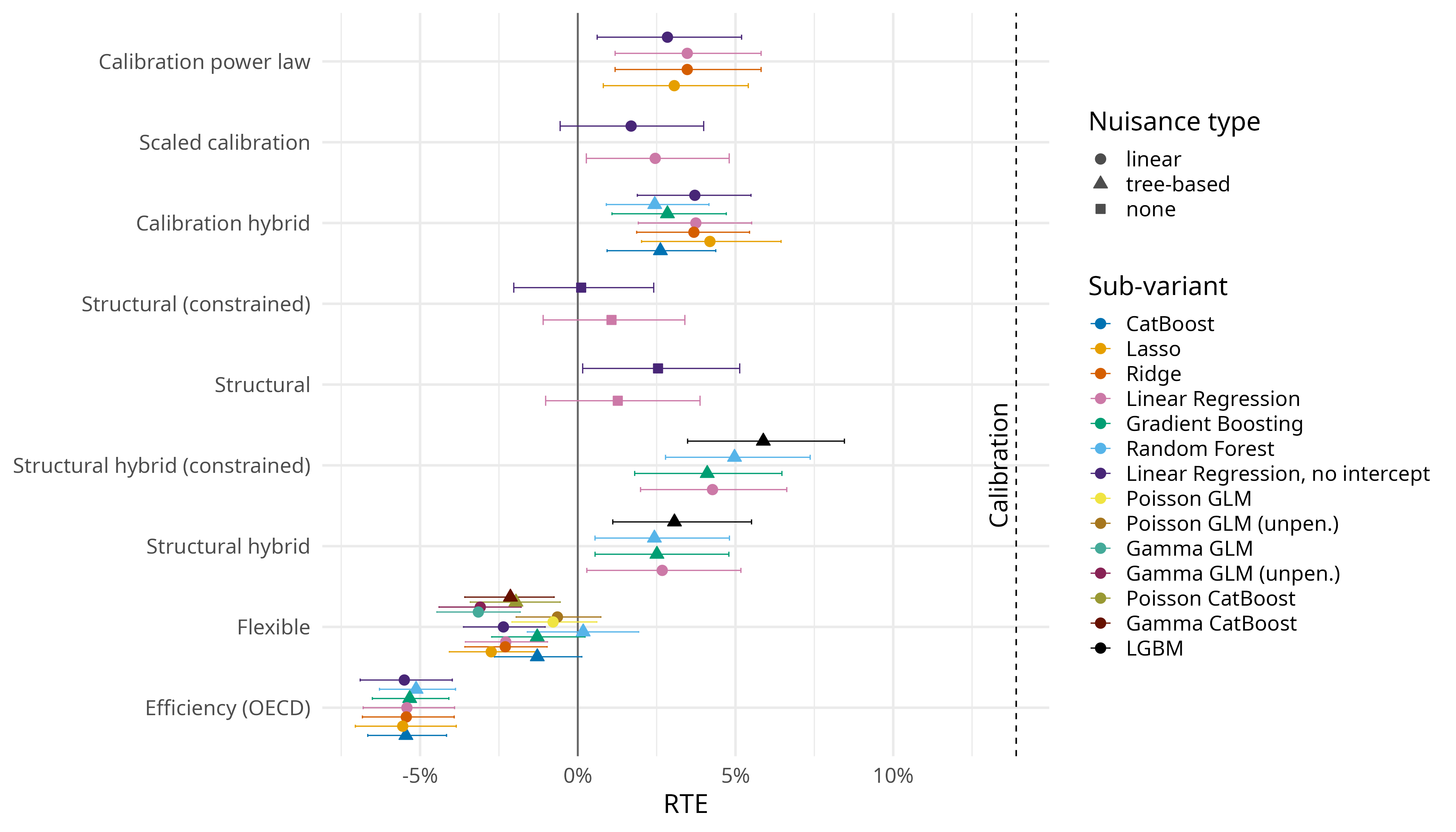}
    \caption{Out-of-sample \ac{RTE} by model specification and sub-variant, container ships. The dashed line is the calibration. Bars are 95\% ship-clustered bootstrap intervals over the test sample, fitted models held fixed.}
    \label{fig:samplescores_container}
\end{figure}

\autoref{fig:support_quartiles_container} shows the quartile splits. The size and distance splits match the dry bulk pattern, with the tree-based learners failing on both outer quartiles. The speed split does not: the calibration-anchored models stay within about 10\% of reported consumption on both outer quartiles, while the unrestricted structural fit underpredicts the fastest quartile and the `Efficiency' class misses both by about 10\%.

\begin{figure}[H]
    \centering
    \includegraphics[width=\textwidth]{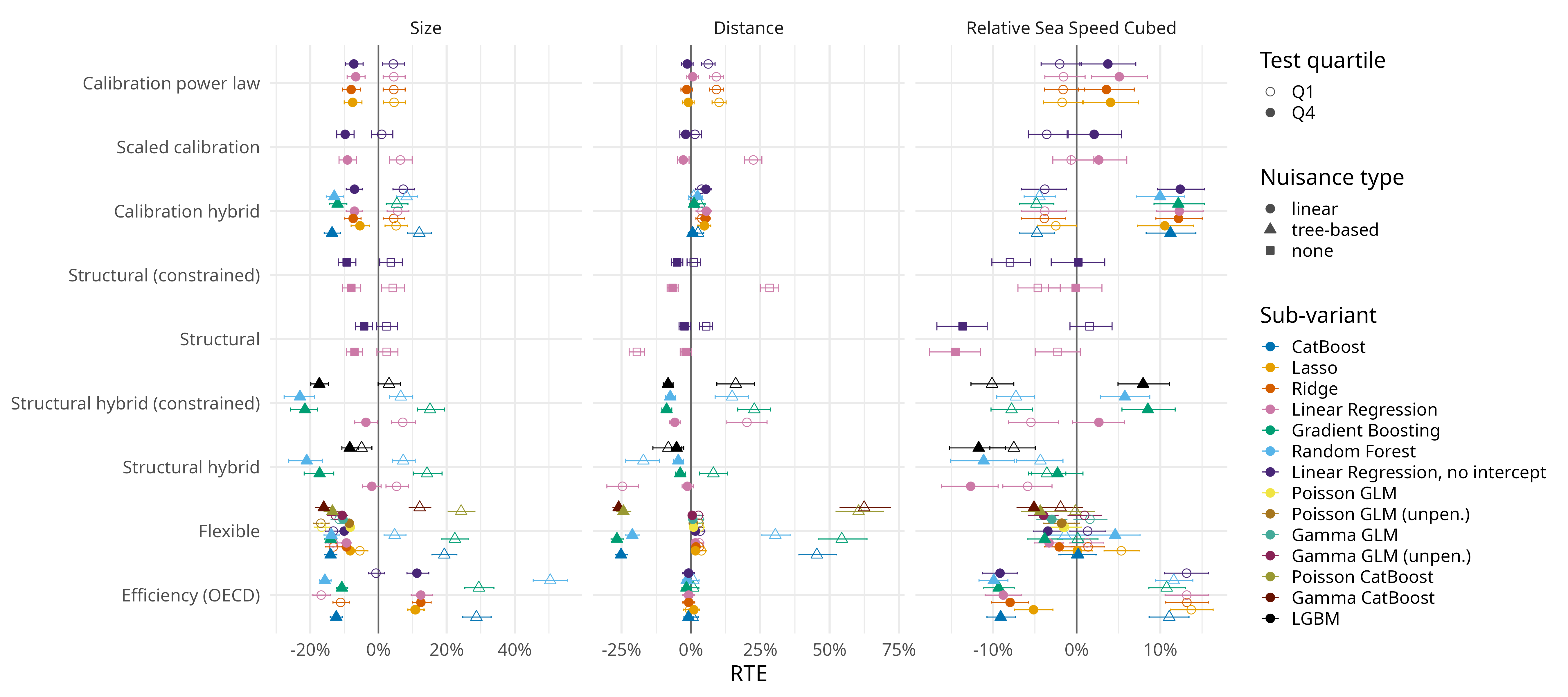}
    \caption{Out-of-support \ac{RTE} for three quartile splits, container ships. Each model is trained on three quartiles and tested on the fourth. Bars are 95\% ship-clustered bootstrap intervals over the test quartile, fitted models held fixed.}
    \label{fig:support_quartiles_container}
\end{figure}

\autoref{fig:elasticity_container} shows the speed counterfactual and \autoref{fig:bias_elasticity_container} the comparison of the two performance axes. \autoref{fig:support_quartiles_mahal_container} shows the Mahalanobis distances of the quartile splits. \autoref{tab:mahal-counterfactual-container} reports the support statistics for the speed counterfactual.

\begin{figure}[H]
    \centering
    \includegraphics[width=\textwidth]{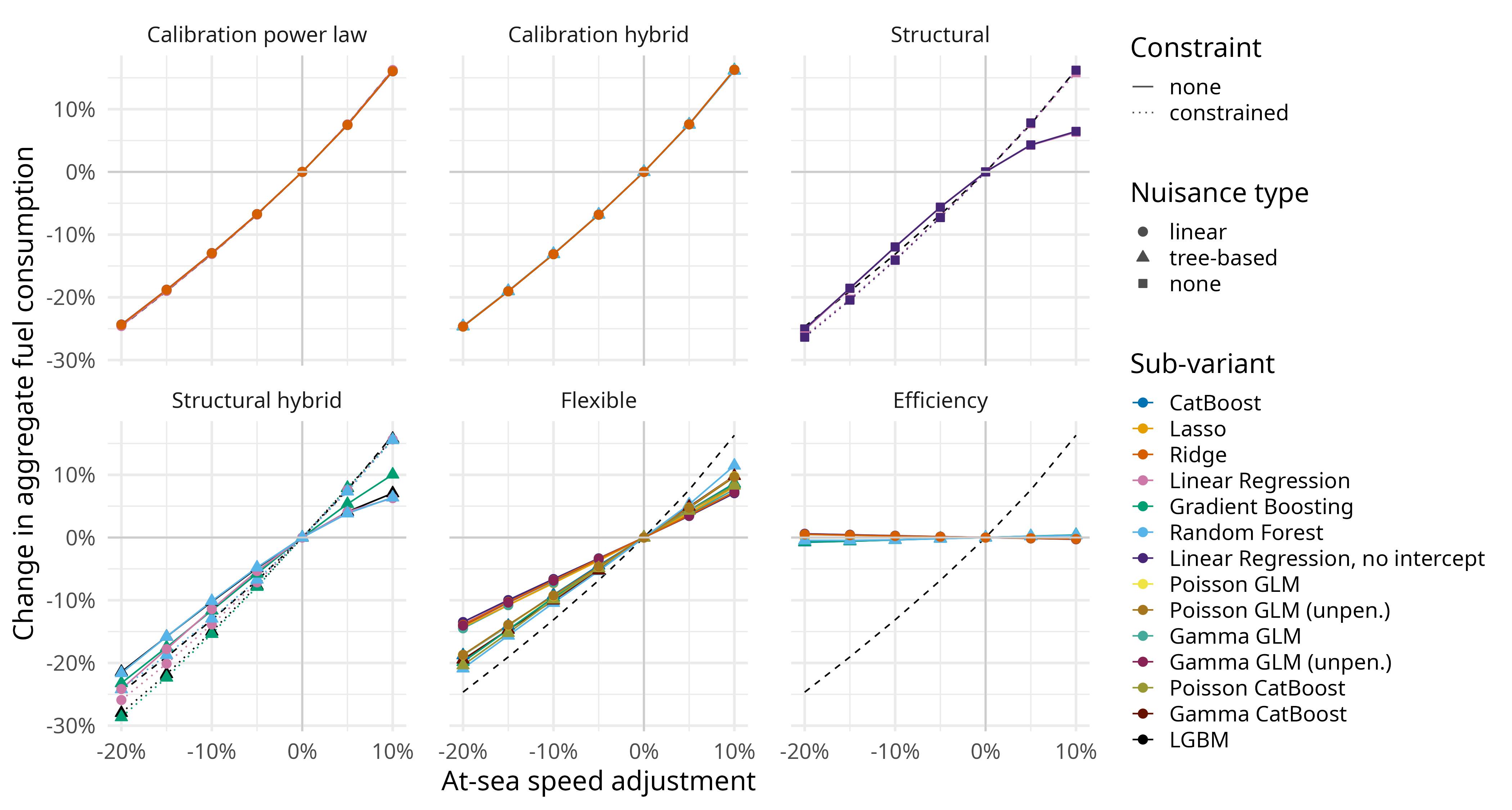}
    \caption{Predicted change in aggregate fuel consumption in the speed-adjusted counterfactual exercise, container ships. The dashed line is the theoretical benchmark.}
    \label{fig:elasticity_container}
\end{figure}

\begin{figure}[H]
    \centering
    \includegraphics[width=0.8\textwidth]{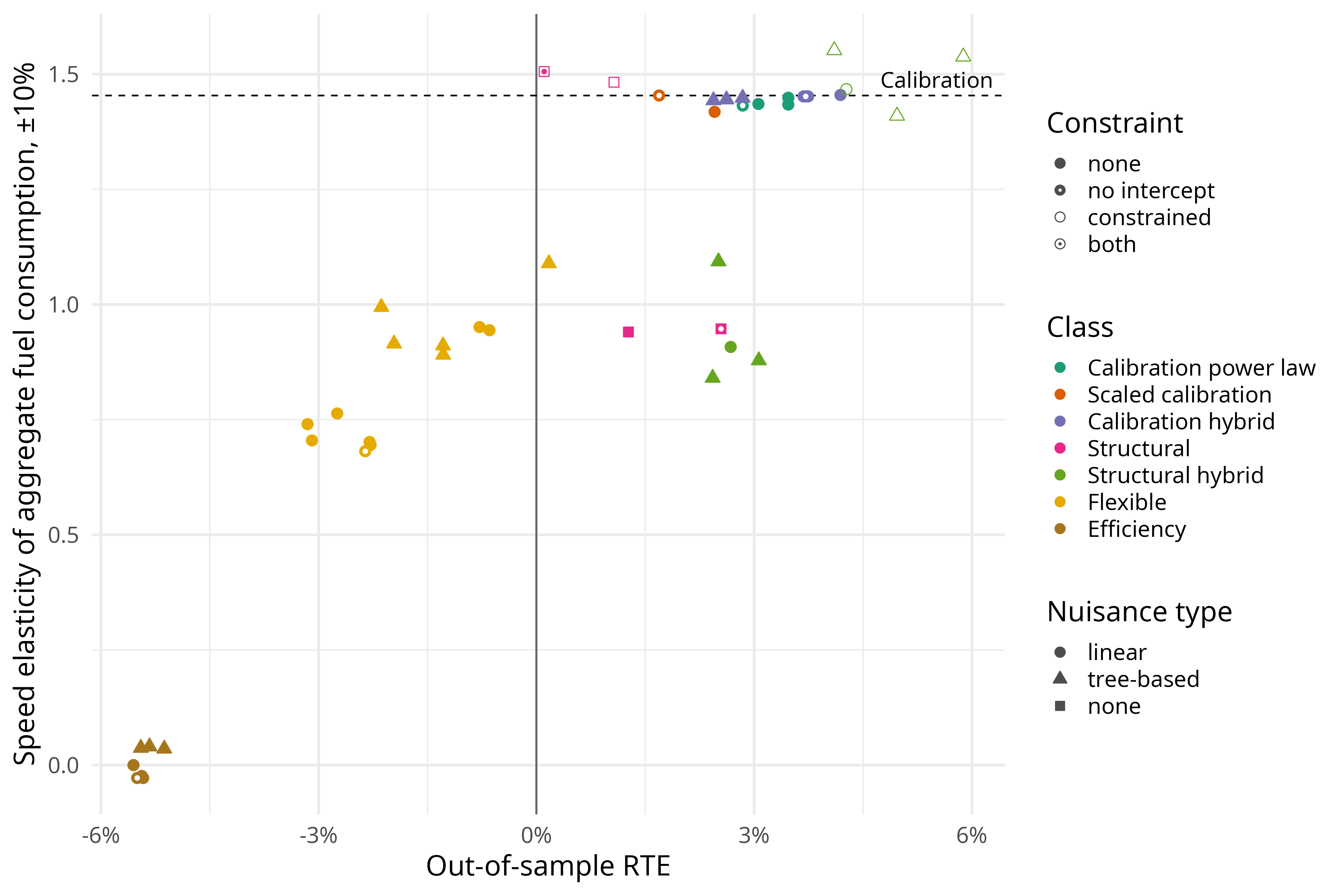}
    \caption{Out-of-sample \ac{RTE} against the speed elasticity of aggregate fuel consumption, container ships. The elasticity is the arc elasticity over the $\pm 10\%$ counterfactual; the dashed line denotes the calibration.}
    \label{fig:bias_elasticity_container}
\end{figure}

\begin{figure}[H]
    \centering
    \includegraphics[width=\textwidth]{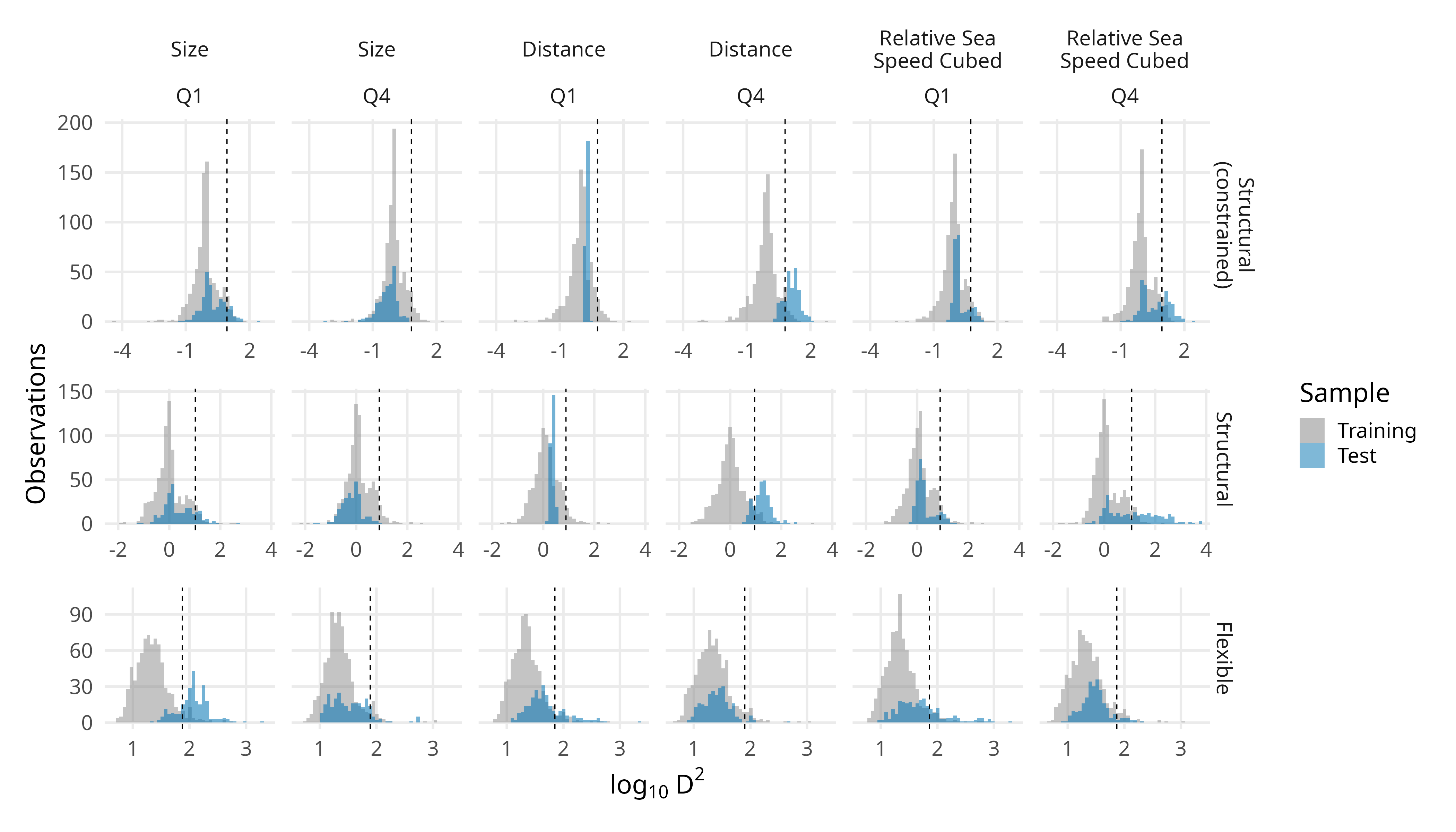}
    \caption{Mahalanobis distance from the training sample for the three quartile splits and three feature sets, container ships. Gray: training observations, each measured leave-one-out against the rest of the training sample; blue: the held-out quartile. The dashed line is the 95th percentile of the training distances.}
    \label{fig:support_quartiles_mahal_container}
\end{figure}

\begin{table}[H]
    \centering
    \begin{threeparttable}
        \caption{Support of the speed-adjusted counterfactual by feature set, container ships.}
        \label{tab:mahal-counterfactual-container}
        
\begin{tabular}[t]{lrrrrrr}
\toprule
\multicolumn{1}{c}{ } & \multicolumn{6}{c}{Speed adjustment} \\
\cmidrule(l{3pt}r{3pt}){2-7}
 & -20\% & -15\% & -10\% & -5\% & +5\% & +10\%\\
\midrule
\addlinespace[0.3em]
\multicolumn{7}{l}{\textbf{Share beyond the observed 95th percentile}}\\
\hspace{1em}Structural (constrained) & 0.20 & 0.13 & 0.08 & 0.05 & 0.08 & 0.11\\
\hspace{1em}Structural & 0.16 & 0.12 & 0.08 & 0.05 & 0.08 & 0.14\\
\hspace{1em}Flexible & 0.08 & 0.07 & 0.05 & 0.05 & 0.04 & 0.03\\
\addlinespace[0.3em]
\multicolumn{7}{l}{\textbf{Mean log excess (training SDs)}}\\
\hspace{1em}Structural (constrained) & 0.086 & 0.051 & 0.028 & 0.020 & 0.041 & 0.107\\
\hspace{1em}Structural & 0.066 & 0.040 & 0.023 & 0.018 & 0.073 & 0.173\\
\hspace{1em}Flexible & 0.065 & 0.046 & 0.034 & 0.028 & 0.022 & 0.022\\
\bottomrule
\end{tabular}

        \begin{tablenotes}
            \small
            \item Statistics as in \autoref{tab:mahal-counterfactual-full}.
        \end{tablenotes}
    \end{threeparttable}
\end{table}

\subsection{Support Metric Details}\label{app:support}

Writing the training covariance as $\hat{\Sigma} = V \Lambda V^\top$, where the columns of $V$ are the principal component directions and $\Lambda$ is the diagonal matrix of the training variances along them, the inverse is $\hat{\Sigma}^{-1} = V \Lambda^{-1} V^\top$. With $z_i = V^\top (x_i - \hat{\mu})$ the coordinates of the centered observation in the principal component basis, the squared distance is
\begin{equation*}
    D_i^2 = z_i^\top \Lambda^{-1} z_i = \sum_k \frac{z_{ik}^2}{\lambda_k},
\end{equation*}
the squared Euclidean distance after rotating into the principal components, which removes the correlation between features, and dividing each coordinate by its standard deviation. A displacement along a direction in which the training data barely vary is divided by a small $\lambda_k$ and counts for a lot; one along a high-variance direction counts for little. A direction with $\lambda_k = 0$ would make any displacement along it infinite, and collinear size measures make the training covariance numerically singular, so the inverse is taken on the non-null directions only.

Distances are computed on the model's feature set after the same transformation and median imputation used in fitting, with the medians taken from the training sample. Training observations are scored leave-one-out, so that each is measured against a centroid and covariance that exclude it; test observations are measured against the full training sample. The out-of-support threshold is the 95th percentile of the leave-one-out training distances. The standardized mean log excess is computed against the same 95th percentile and scaled by the standard deviation of $\log_{10} D^2$ over the leave-one-out training distances.

\subsection{Measurement Error Simulation}\label{app:errorsim}

The simulation of \autoref{sec:structuralexplanation} perturbs two calibrated quantities for each dry bulk ship when generating the outcome and leaves the features unaltered. Writing $C_i$ for the drag constant that scales the load $w_{ih}$ and $E_i = W_{\text{ME},\textit{ref},i}\,\mathrm{SFC}_{\text{base},i}$ for the engine term, the outcome is generated with $C_i(1+u_{C,i})$ and $E_i(1+u_{E,i})$ in their place, where $u_{C,i} \sim N(0,\sigma_C^2)$ and $u_{E,i} \sim N(0,\sigma_E^2)$ are drawn once per ship. The simulated main engine consumption is
\begin{equation*}
    \mathrm{FC}_{\text{ME},i} = E_i (1+u_{E,i}) \sum_{k=1}^{3} \beta_k (1+u_{C,i})^k \sum_h w_{ih}^k,
\end{equation*}
with the calibrated $\beta_k$, so the drag error enters the $k$-th load term at the $k$-th power while the engine error scales all three alike. The features carry the calibrated values and the outcome the perturbed ones, so the error is of the Berkson kind.

To choose appropriate error variances $\sigma_C$ and $\sigma_E$ we draw on a simple model of ship operation paired with observed speed variation. First, a linear approximation of a ship's speed around the no-error point gives:
\begin{equation*}
    \log v_i = \log v_i^{0} - \gamma_C u_{C,i} - \gamma_E u_{E,i},
\end{equation*}
where $v_i^{0}$ is the speed the ship would sail with zero error and $-\gamma_C$ and $-\gamma_E$ are the elasticities of speed with respect to the drag constant and the engine term. Since $v_i^{0}$ is independent of the errors, and the errors of each other, the approximation implies (a) $\operatorname{cov}(u_{C,i}, \log v_i) = -\gamma_C \sigma_C^2$ (and symmetrically for $E$), the covariances the speed-correlated scenarios impose on the draws, and (b) $\operatorname{var}(\log v_i) = \operatorname{var}(\log v_i^{0}) + \gamma_C^2\sigma_C^2 + \gamma_E^2\sigma_E^2$, which bounds the error variances by the fleet's dispersion in speed. Both require the two elasticities.

We obtain them from a simple model of operator behavior. An operator who knows $C_i$ and $E_i$ chooses the speed for a voyage of given distance that minimizes fuel cost, at price $p_f$, plus the cost of time at sea, at $r$ per hour. With the load $w = a_i v^3$, where $a_i$ is proportional to $C_i$, the first-order condition equates the marginal fuel cost of sailing faster to the value of the time saved, and can be written
\begin{equation*}
    p_f E_i\, g(w^*) = r, \qquad g(w) = \sum_{k=1}^{3} (3k-1)\,\beta_k w^k.
\end{equation*}
The drag constant does not appear, so the optimal load $w^*$ is the same for every $C_i$, and since $v^* = (w^*/a_i)^{1/3}$, $\gamma_C = 1/3$. Differentiating implicitly in $E_i$, $v^*$ has elasticity $-1/(3\eta_g)$, where $\eta_g \equiv w^* g'(w^*)/g(w^*)$ denotes the load elasticity of $g$ at the optimum. Evaluated at the median load of 0.36, the calibrated coefficients give $\eta_g = 0.809$ and $\gamma_E = 0.412$.

The speed variance due to the errors should not exceed the observed variance: $\gamma_C^2\sigma_C^2 + \gamma_E^2\sigma_E^2 \le \sigma_v^2$, where $\sigma_v = 0.081$ is the between-ship standard deviation of mean log at-sea speed. With $\sigma_C = \sigma_E = \sigma$ this gives $\sigma \le 0.153$. We take $\sigma = 0.15$, near the limit.

Empirical speed elasticities are lower than the model predicts: $-\gamma_E$ is also the elasticity of speed with respect to fuel price, since $E_i$ and $p_f$ enter the first-order condition only as a product, and estimates of the latter are around $-0.10$ \parencite{ollila2024effect}. We therefore also run a damped response, scaling both elasticities by $\phi = 0.10/\gamma_E = 0.243$, alongside the theoretical $\phi = 1$.

The errors are drawn from the normal distribution conditional on each ship's mean log at-sea speed with the covariances in (a), scaled by $\phi$; at $\phi = 0$ the draws are independent of speed. Five scenarios are run, with 50 replications each: both errors and the drag error alone independent of speed, both errors at $\phi = 0.243$ and at $\phi = 1$, and the drag error alone at $\phi = 1$. On every replication we refit the calibration power law, the unconstrained and constrained structural fits, and, for the `Flexible' class, the two log-target fits with the largest and smallest $\pm 10\%$ arc elasticity on the reported data, ridge regression and random forest.

\end{document}